\documentclass[12pt]{article}
\usepackage{amsmath}
\usepackage{amssymb}
\usepackage{amsthm}
\usepackage{fullpage}
\usepackage[utf8]{inputenc}
\usepackage{mathtools}
\usepackage{url}
\usepackage[usenames,svgnames]{xcolor}
\usepackage{cite}
\usepackage{enumitem}
\usepackage{comment}
\usepackage{mathdots}
\usepackage{graphicx}
\usepackage{float}
\usepackage{enumitem}

\usepackage[pagebackref=false]{hyperref} %hyperref wants to be loaded last. Set = false if no pagebackref needed
\usepackage[all]{hypcap} %fix hyperref links to tables/figures/etc.
\theoremstyle{plain}
\newtheorem{theorem}{Theorem}[section]
\newtheorem{lemma}[theorem]{Lemma}
\newtheorem{definition-theorem}[theorem]{Definition-Theorem}
\newtheorem{definition-proposition}[theorem]{Definition-Proposition}

\newtheorem{proposition}[theorem]{Proposition}
\newtheorem{corollary}[theorem]{Corollary}
\newtheorem{example}{Example}[section]
\newtheorem{examples}{Example}[subsection]

\newtheorem{remark}{Remark}[section]

\theoremstyle{definition}
\newtheorem{definition}{Definition}[section]

\numberwithin{equation}{section} % requires package amsthm

\hypersetup{
breaklinks,
colorlinks,
citecolor=Green,
linkcolor=BlueViolet,
urlcolor=DarkBlue,
pdftitle={Weighted Hurwitz numbers , constellations and topological recursion},
pdfauthor={A. Alexandrov, G. Chapuy, B. Eynard,  J. Harnad } }

 \newcommand{\sgn}{\mathop{\mathrm{sgn}}}

\newcommand{\mat}[1]{\mathbf{#1}}

\DeclareMathOperator{\cyc}{cyc}

\DeclareMathOperator{\aut}{aut}
\DeclareMathOperator{\Res}{Res}

\def\ra{{\rightarrow}}

\def\Tr{\mathrm {Tr}}

\def\det{\mathrm {det}}

\def\span{\mathrm {span}}

\def\span{\mathrm {span}}

\def\Res{\mathop{\mathrm {res}}\limits}

\DeclarePairedDelimiter{\abs}{|}{|}

\def\be{\begin{equation}}
\def\ee{\end{equation}}

\def\bea{\begin{eqnarray}}
\def\eea{\end{eqnarray}}

\def\bt{\begin{theorem}}
\def\et{\end{theorem}}

\def\bex{\begin{example}\small \rm}
\def\eex{\end{example}}

\def\bexs{\begin{examples}\small \rm}
\def\eexs{\end{examples}}

\def\ra{\rightarrow}
\def\ss{\subset}
\def\deq{\coloneqq}

\def\br{\begin{remark}\small \rm}
\def\er{\end{remark}}

\def\Res{\mathop{\mathrm{Res}}\limits}

\def\&{&{\hskip -20pt}}

\def\CC{\mathcal{C}}

\def\HH{\mathcal{H}}

\def\NN{\mathcal{N}}
\def\OO{\mathcal{O}}

\def\WW{\mathcal{W}}

\def\Cb{\mathbf{C}}

\def\Eb{\mathbf{E}}

\def\Ib{\mathbf{I}}

\def\Kb{\mathbf{K}}
\def\Lb{\mathbf{L}}

\def\Nb{\mathbf{N}}

\def\Nb{\mathbf{N}}
\def\Pb{\mathbf{P}}
\def\Qb{\mathbf{Q}}

\def\Zb{\mathbf{Z}}

\def\pp{\partial}

\newcommand{\branch}{{\mathcal{L}}}

\begin{document}
\baselineskip 16pt
%\begin{flushright}
%CRM-xxxx(2016)
%\end{flushright}
\medskip
\begin{center}
\begin{Large}\fontfamily{cmss}
\fontsize{17pt}{27pt}
\selectfont
	\textbf{Weighted Hurwitz numbers and topological recursion}
	\end{Large}
\\
\bigskip \bigskip
\begin{large}  A. Alexandrov$^{1, 2}$\footnote{e-mail: alexandrovsash@gmail.com}, G. Chapuy$^{3}$\footnote{e-mail: guillaume.chapuy@irif.fr. Supported by the grant ERC-2016-STG 716083 ``CombiTop''.},  B. Eynard$^{4, 5}$\footnote{e-mail:  bertrand.eynard@cea.fr} and J. Harnad$^{4, 6}$\footnote {e-mail: harnad@crm.umontreal.ca  }
 \end{large}\\
\bigskip
\begin{small}
$^{1}${\em Center for Geometry and Physics, Institute for Basic Science (IBS), Pohang 37673, Korea}\\

 \smallskip
 $^{2}${\em ITEP, Bolshaya Cheremushkinskaya 25, 117218 Moscow, Russia} \\
\smallskip
$^{3}${\em CNRS, IRIF UMR 8243, Universit\'e de Paris\\
Case 7014,
75205 Paris Cedex 13 France}\\
 \smallskip
$^{4}${\em Centre de recherches math\'ematiques,
Universit\'e de Montr\'eal\\ C.~P.~6128, succ. centre ville, Montr\'eal,
QC H3C 3J7  Canada}\\
	\smallskip
 $^{5}${\em 
Université Paris-Saclay, CNRS, CEA,  Institut de physique théorique, 
	\\91191, Gif-sur-Yvette, France.}\\ 
\smallskip
$^{6}${\em Department of Mathematics and Statistics, Concordia University\\ 1455 de Maisonneuve Blvd.~W.~Montreal, QC H3G 1M8  Canada}
\smallskip
\end{small}
\end{center}

\begin{small}\begin{center}
\today
\end{center}\end{small}
\medskip
%%%%%%%%%%%%%%%%  Abstract  %%%%%%%%%%%%%%%%
\begin{abstract}
   The  KP and 2D Toda $\tau$-functions of hypergeometric type that  serve as generating functions  for weighted single and double Hurwitz numbers are  related to the topological recursion programme. A graphical  representation of such weighted Hurwitz numbers is given  in terms of weighted  constellations.  The associated  classical and quantum spectral curves are derived, and these are interpreted combinatorially  in terms of the graphical model. The pair correlators are given a finite Christoffel-Darboux representation and  determinantal expressions are obtained for the multipair correlators. The genus expansion of the multicurrent correlators  is shown to  provide generating series for weighted Hurwitz numbers of fixed ramification profile lengths. The WKB series for the Baker function  is derived and used to deduce the loop equations and the topological recursion relations in the case of polynomial weight functions.
        \end{abstract}

\break

%%%%%%%%%%%%%%%%% Contents %%%%%%%%%%%%%%%%%
\tableofcontents

%%%%%%%%%%%%%%%% Section 1. Introduction %%%%%%%%%%%%%%%%

\section{Introduction and main result}

%%%%%%%%%%%%%%%% Section 1.1 Introduction %%%%%%%%%%%%%%%%

\subsection{Introduction}

In their original geometric sense, Hurwitz numbers enumerate  $N$-fold ramified coverings of the Riemann sphere
with given ramification types at the branch points. They can also be interpreted combinatorially as enumerating factorizations of 
the identity element of the symmetric group $\mathfrak{S}_N$ into a product of elements belonging to given conjugacy classes. They were first introduced and studied by Hurwitz~\cite{Hu1, Hu2} and subsequently related to the structure and characters of the symmetric group by Frobenius \cite{Frob1, Frob2}. 

Many variants and refinements have been studied in recent years~\cite{GGN1, GGN2,  Ok, Pa, OP, GH1, GH2, H1, H2, HO2, AC1, AC2, KZ,  Z, AMMN,AMMN1, NaOr, BEMS, AM, MM, AEG, BM, MSS, ALS, Or2, Or3}, culminating in the introduction of  \emph{weighted} Hurwitz numbers~\cite{GH1, GH2, H2, HO2}, which are weighted sums of Hurwitz numbers depending on a finite or infinite number of weighting parameters. All recently studied variants are special cases of these, or suitably defined limits. Combinatorially, the weighted enumeration of branched coverings is equivalent to the weighted enumeration of families of embedded graphs such as maps, {\em dessins d'enfants}, or more generally, {\em constellations}~\cite{LZ}. There has also been important progress in relating Hurwitz numbers to other classes of enumerative geometric invariants ~\cite{AEG,Ok, Pa, ELSV, EO3, GJV, K, BM, Or2, Or3} and matrix models ~\cite{AC1,AC2, AM, AMMN, AMMN1, BEMS, BM, MSh, NaOr, Or2, Or3, Dunin}.

A key development was the identification by Pandharipande \cite{Pa} and Okounkov \cite{Ok} that certain special 
$\tau$-functions for integrable hierarchies of the KP and 2D Toda type may serve as generating functions
for {\em simple} (single and double) Hurwitz numbers (i.e., those for which all branch points, with the possible exception of one, or two,  have simple ramification profiles).   It was  shown  subsequently  \cite{AMMN1, AMMN, GH1, GH2, H2, HO2} that all weighted (single and double) Hurwitz numbers 
have KP or 2D Toda $\tau$-functions of  the special {\em hypergeometric} type \cite{Or1,  OrSc1, OrSc2}  as generating functions.

An alternative approach,  particularly useful for studying  genus dependence and recursive relations between the invariants 
involved~\cite{GJV}, consists of using multicurrent correlators as generating functions for weighted Hurwitz numbers having
a fixed ramification profile length $n$. These  may be defined in a number of equivalent ways: either 
as the coefficients in multivariable Taylor series expansions of  the $\tau$-function at suitably defined  
$n$-parameter families of evaluation points, in terms of  pair correlators, or as fermionic expectation values  
of products of  current operators evaluated at $n$ points \cite{ACEH2}. 

A remarkable way of computing Hurwitz numbers, which provides strong results about their structure, follows from the method of Topological Recursion (TR), introduced by Eynard and Orantin in~\cite{EO1}. This approach, originally 
inspired by results arising naturally in random matrix theory~\cite{Eyn05}, has been shown applicable  to many enumerative geometry problems, such as the  counting of maps~\cite{Eyn:book} or the computation of Gromov-Witten invariants~\cite{BKMP}.   It has received a great deal of attention in recent years and found  to have many far-reaching implications. 
The fact that simple Hurwitz numbers satisfy the TR relations was  conjectured by Bouchard and Mari\~{n}o~\cite{BM}, and proved in~\cite{BEMS, EMS}.  This provides an algorithm that allows them to be computed by recursion in the Euler characteristic, starting from  initial data corresponding to the so-called ``disk'' and ``cylinder'' case (which  in the notation of the present paper correspond to  genus $g=0$ and  $n=1, 2$, respectively).

 The basic recursive algorithm is the same in all these problems, the only difference being the so-called {\em spectral curve} that corresponds to the initial data. In the case of weighted Hurwitz numbers, it implies connections between their structural properties and relates them to other areas of enumerative geometry. In particular, it  implies the existence of formulae of ELSV type~\cite{ELSV, EO3}, relating Hodge invariants, $\psi$-classes and Hurwitz numbers.  It is interesting to note that, although TR and its consequences are universal (only the spectral curve changes), the  detailed proofs of its validity in the various models are often distinct, model-dependent and, to some extent, {\em ad hoc}.

In the present work we prove, under certain technical assumptions, that weighted Hurwitz numbers satisfy the TR relations.  For brevity and simplicity, we assume that the weight generating function $G(z)$, and  the exponential factor $S(z)$ determining the second set of KP flows in the 2D Toda model are polynomials, leaving the extension of these results to more general cases to further work.  For combinatorialists, we emphasize that the main result may be interpreted equivalently  as applying to suitably weighted enumeration of constellations, as explained below. 

The main conclusions were previously announced in the overview paper~\cite{ACEH1} and are summarized in Subsections \ref{main_result} and \ref{outline}.   They rely in part on results proved  in~\cite{ACEH1} and in the companion paper~\cite{ACEH2} on fermionic representations.  In certain cases, new  proofs are
provided that have a different form from those in~\cite{ ACEH2}.

\medskip

\begin{remark}
The fact that weighted Hurwitz numbers satisfy the TR relations deeply involves their integrable structure and requires a rather intricate sequence of preparatory results, each of which has its own independent interest.  Such results have always been a challenge to prove, with  characteristically different difficulties for every enumerative geometry problem considered. This can perhaps be understood,  considering  that TR is interpretable as a form of mirror symmetry \cite{BKMP}.
\end{remark}

\medskip
We proceed through the following sequence of preparatory steps.
\begin{itemize}[itemsep=0pt, topsep=0pt, parsep=0pt, leftmargin=20pt]
	\item[1.] The generating functions for weighted Hurwitz numbers are related to certain integral kernels (pair correlators, or $2$-point functions) and their
	$2n$-point generalizations. These  are shown to satisfy determinantal formulae that follow from  their expression in terms of 
	$\tau$-functions,  which themselves are given by Fredholm determinants.
	\item[2.] The integral kernels are shown to have a Christoffel-Darboux-like form (sometimes called  {\em integrable kernels} \cite{Szeg, IIKS, HI}),  with numerators consisting of a finite sum over bilinear combinations of solutions of  a linear differential system with rational coefficients, and a Cauchy-type denominator.
	\item[3.] This property, together with the expansion of the generating functions and Baker functions in a small parameter $\beta$ appearing in the definition of the $\tau$-function, is used to  derive a WKB-like expansion, with powers corresponding to the Euler characteristic; hence, a \emph{topological expansion}. Moreover, using the differential system, the generating functions are shown to have poles only at the branch points of the spectral curve, a non-trivial property.
	\item[4.] The differential system is also used to prove that generating functions for weighted Hurwitz numbers satisfy a set of  consistency conditions,  the {\em loop equations}.
	\item[5.] These equations are used, together with the fact that the generating functions are analytic away 
	from the branch points, to prove the Topological Recursion (TR) relations, along lines developed in~\cite{EO1}.
\end{itemize} 

%%%%%%%%%%%%%%%% Section 1.2 Main result %%%%%%%%%%%%%%%%

\subsection{Main result}
\label{main_result}   

The main result of this work is that the coefficients  $\tilde W_{g,n}(x_1,\dots, x_n)$  of  multicurrent correlators  in a genus expansion serve as generating functions  for weighted double Hurwitz numbers $H^d_G(\mu, \nu)$, with weights determined by a weight  generating function $G(z)$,  and satisfy the  topological recursion (TR) relations. 
  The complete statement of this fact  requires a number of preparatory  definitions and  results; it is given in Theorem~\ref{thm:toprec}, Section~\ref{sec:toprec}.
For genus $g\geq 0$ and $n\geq 1$ points in the correlator, let
\be
\tilde W_{g,n}(x_1,\dots, x_n)\equiv \tilde W_{g,n}^G({\bf s}; \gamma; x_1,\dots,x_n)
\ee
denote the generating function (see eqs.~(\ref{eq:defFn}) -  (\ref{eq:defWGgn_tilde}))  for  weighted double Hurwitz numbers $H^d_G(\mu, \nu)$ associated to the $(g,n)$ step of the recursion. The  $\tilde W_{g,n}$'s are identified in Section~\ref{sec:fermionicAndBosonic} as coefficients in the $\beta$-expansion of the multicurrent correlation function $\tilde W_{n}(x_1,\dots, x_n)$ associated to the underlying 2D-Toda  $\tau$-function $\tau^{(G,\beta,\gamma)}({\bf t}, {\bf s})$ of hypergeometric type.  The $\beta$ dependence  follows from evaluations of the function $G(z)$ that generate the weighting factor 
$\WW_G(\mu^{(1)}, \dots, \mu^{(k)})$ in the definition (\ref{Hd_G})  of $H^d_G(\mu, \nu)$. 
The variables $\{x_i\}_{i=1, \dots, n}$ are viewed as evaluations of the spectral parameter,  and the second set of 2D-Toda flow parameters, denoted ${\bf s} = (s_1, s_2, \dots )$, serve as bookkeeping parameters that record the part lengths in the ``second partition''  $\nu$ of the double Hurwitz numbers $H^d_G(\mu, \nu)$. Powers of  $\beta$ in the expansion
of $\tilde W_{n}(x_1,\dots, x_n)$ keep track of the genus of the covering curve, and $\beta$ also serves  as 
the small parameter in the WKB expansion of the wave function (or Baker function). Powers of an auxiliary parameter 
$\gamma$  in the multiple parameter expansion keep track of the degree of the covering.

The following is an abbreviated version of the main result. (See  Section~\ref{sec:toprec}, Theorem~\ref{thm:toprec}, for a 
more precise statement.)
\begin{theorem}
\label{thm:toprec_abbreviated}
	 Choosing both $G(z)$  and $S(z)=\sum_{k \geq 1} ks_k z^k$ as polynomials in $z$, the 
 ramification points  of the algebraic plane curve (the {\em spectral curve})
\be
x y = S\left( \gamma x G(xy) \right),
\label{spectral_curve_intro}
\ee
with rational parametrization:
\be
\label{zcurveIntro}
	X(z):=\frac{z}{\gamma G(S(z))}, \quad
	Y(z):= \frac{S(z)}{z}\,\gamma G(S(z)),
\ee
under the  projection map $(x,y) \ra x$  are given by the zeros of the polynomial
	\bea 
	\label{virtualMainTheorem}
	G(S(z))-zG'(S(z))S'(z), 
\eea
which  are assumed to all be simple. The (multicurrent) correlators $\tilde W_{g,n}(x_1,\dots, x_n)$ then satisfy the Eynard-Orantin topological recursion relations given by eqs.~(\ref{omega_tilde_K} - \ref{W2gn_TR}), with spectral curve~\eqref{spectral_curve_intro}, \eqref{zcurveIntro}.
\end{theorem}
\begin{remark}
We restrict ourselves to the case of simple branch points and polynomial $G(z)$ and $S(z)$. The first restriction is mainly for the
sake of simplicity; as will appear, the proofs are already quite technical. For higher ramification types, an extended version of topological recursion is expected to hold~\cite{BouchardEynardlocalglobal}. The second restriction is more essential, and we do not expect that our results can be immediately extended to the most generic non-polynomial case.  Although it is natural to expect that many of the techniques applied here can be adapted when these restrictions are removed, we do not pursue this here. (See the remarks at the end of Section~\ref{sec:F03}.). 
	We also stress that many of our intermediate results are proved here without any  assumption of polynomiality or simplicity. 
	This is the case for most results in Sections~\ref{sec:weightedHurwitzNumbers} to~\ref{subsec:matricesPQ}, with the exception of Section~\ref{sec:constellations} (although it would be a simple exercise to extend the definition of constellations to infinitely many colours) and Section~\ref{sec:CDR}, where the polynomiality of the Christoffel-Darboux relation follows from the polynomiality of $G$ and $S$. 
However, the non-polynomial Christoffel-Darboux relation is valid in general. The same is true for the quantum and classical curve equations.
Further details on non-polynomial weight generating functions $G$, including the computation of pair correlators, current correlators,
the classical and quantum spectral curves and a matrix representation of the $\tau$ function for both rational  $G$ and
the the $q$-exponential function may be found in refs. \cite{BH, BHR, HR}.
\end{remark}
\begin{remark}
	As will be seen, the combinatorial model underlying the choice of a polynomial $G(z)$ is essentially given by the constellations of Bousquet-Mélou and Schaeffer~\cite{MBM-Schaeffer}. The numbers we study can thus be considered as a refinement of the Bousquet-Mélou Schaeffer (BMS) numbers with a finer control on the number of vertices of each color. However the original (single) BMS numbers correspond to $G(z)=(1+z)^M$, which is not directly covered by our results as it does not satisfy the simplicity assumption.
\end{remark}

\begin{remark}
	Probably the best-known example of  weighted Hurwitz numbers for polynomial weight generating function $G(z)$ is 
given by strictly monotone double Hurwitz numbers, or {\em dessins d'enfants} (or  {\em hypermaps}, or {\em bipartite maps}), 
corresponding to ramified covers of $\Pb^1$ with at most three branch points, for which
\be
G(z)=1+z.
\ee
Another case that also gives rise to a weighted enumeration involving only three branch points, and hence {\em dessins d'enfants},
is the {\em single} weighted Hurwitz number $H^d_G(\mu)= H^d_G(\mu, (1)^N)$ obtained by setting  the second unweighted 
partition $\nu$ equal to the conjugacy class of the identity element (and hence corresponding to a nonramified  point) 
and choosing the weight generating function as a quadratic polynomial
\be
G(z) = (1+c_1 z)(1+c_2 z).
\ee
(See, e.g. \cite{AC1}.)

	The orbifold case corresponds to the monomial $S(z)=z^r$. The quantum curve for this case was obtained in \cite{Do} using combinatorics of hypermaps and in \cite{Dunin} using the loop equations for hypermaps. Topological recursion was discussed in \cite{Dunin, KZ}. For the more general case of double strictly monotone Hurwitz numbers, relevant to this paper, the quantum spectral curve equation was derived in \cite{ALS}. Our approach also enables consideration of the more general orbifold case, both simple or double, by choosing $S(z)=z^r$ and a generic polynomial $G(z)$.
	
	For example, strictly monotone orbifold Hurwitz numbers correspond to $G(z)=1+z$ and $S(z)=z^r$. They have been considered in~\cite{Do,Dunin}, see also \cite{HO2, ALS}. 
 If we put $\gamma=1$, we have 
 \be
 X(z)=\frac{z}{1+z^r}, \quad Y(z)=z^{r-1}(1+z^r).
 \ee
After  identification of the global parameter $z$ with that of \cite{Do}, this leads to the relations 
 \be
 x(z)=\frac{1}{X(z)}=\frac{1}{z}+z^{r-1}, \quad y(z)=X(z)+X(z)^2Y(z)=z
 \ee 
	between the functions $(X(z),Y(z))$ and those used in \cite{Do} (see also Remark~\ref{rem:orbifold}).
	
	These results fit into a growing literature on Hurwitz numbers and topological recursion. Readers interested in other recent developments may consult, for example \cite{BSLM, DLN} for the Hurwitz orbifold case, \cite{MSS,SSZ} for $r$-spin Hurwitz numbers, our previous papers~\cite{ACEH1, ACEH2} which, among other things, discuss quantum and classical spectral curves for 
simple (single and double) Hurwitz numbers, strongly and weakly monotone Hurwitz numbers, and  \cite{BH, BHR, HR}, which deal with weight generating functions $G$ that are either rational, or a $q$-exponential, generating simple quantum Hurwitz numbers.

\end{remark}

\begin{remark}
 To help readers with a background mainly in combinatorics in understanding the meaning of this theorem, we note that the function $\tilde W_{g,n}(x_1, \dots, x_n)$ is just the exponential generating function, in the parameter $\gamma$, of double weighted Hurwitz numbers of genus $g$, with weighting function $G(z)$. The first partition has $n$ parts whose lengths are marked by the variables $x_1,\dots, x_n$, and for each $i\geq 1$ the variable  $s_i$ marks the parts of length $i$ in the second partition. It can also be thought of as a generating function of constellations of genus $g$ with $n$ vertices of a given colour, as explained below. Topological recursion gives an explicit way to compute these functions recursively in closed form, and provides much information about their structure.
 \end{remark}

%%%%%%%%%%%%%%%% Subsection 1.3 Outline  %%%%%%%%%%%%%%%%

\subsection{Outline}
\label{outline}

In Section~\ref{sec:weightedHurwitzNumbers}  weighted Hurwitz numbers are defined and the 
parametric family of 2D Toda $\tau$-functions  that serve as generating functions for these is introduced. Section~\ref{sec:constellations} presents a classical graphical model for weighted enumeration of branched covers  consisting of weighted {\em constellations}. Section~\ref{sec:fermionicAndBosonic}   introduces several families of functions,  identified either as  {\em fermionic} or {\em bosonic}, associated to the $\tau$-function and gives the relations between these; namely: 1) the  multicurrent correlators $\tilde W_{n}$  and their coefficients  $\tilde W_{g,n}$ in the genus expansion that appears in our main result; 2) a pair of dual bases $\{\Psi^\pm_k\}$ that extend the Baker function and its dual
and are adapted to a basis of the infinite Grassmannian element that determines the $\tau$-function,
 and 3) the pair correlation kernel $K$ and its $n$-pair generalization $K_n$.
In Section~\ref{sec:recursionOperators}  an operator formalism adapted to these quantities is developed and used to derive a bilinear formula of Christoffel-Darboux type for $K$. In Section~\ref{section5}  infinite and finite linear differential systems satisfied by the $\Psi^\pm_k$'s, are studied, together with recursion relations  that allow  the infinite to be mapped onto the finite ones by ``folding''.  
In particular, this leads to the quantum spectral curve equation. Section~\ref{sec:classicalCurve} 
concerns basic geometric properties of the classical spectral curve and its branch points,  a necessary step in establishing the topological recursion relations. In Section~\ref{sec:WKB}  a key technical result about the rational structure and poles of the multidifferentials $\tilde \omega_{g,n}$  corresponding to the $\tilde W_{g,n}$'s is proved. This is done by studying the $\beta$-expansions (or WKB-expansions) using the tools of the previous sections and delicate inductions. This leads
to the proof of a version of the loop equations in Section~\ref{loop_eqs_fundamental_sys}. Finally, in Section~\ref{sec:toprec}, the topological recursion relations are stated and proved, together with certain corollaries and examples.

The paper is largely self-contained, except for some proofs that have already been given either in the overview paper \cite{ACEH1} or the companion paper \cite{ACEH2} on fermionic representations. These include  some explicit formulae for the kernel $K$, the recursion relations for the adapted bases  $\{\Psi^\pm_k\}$, and the relation between the correlators $\tilde W_{n}$ and the $\tau$-function.  In order not to interrupt the flow of the development, most of the detailed proofs 
have been placed in the Appendix.

%%%%%%%%%%%%%%%% Section 2  Weighted Hurwitz numbers %%%%%%%%%%%%%%%%

\section{Weighted Hurwitz numbers and $\tau$-functions as generating functions }
\label{sec:weightedHurwitzNumbers}

In the following, we introduce notation and definitions needed for dealing with  $\tau$-functions in the setting of formal power series.

Single  $[\, . \, ]$,  respectively, double $[[\, .\, ]]$  square brackets in a set of variables (or indeterminates)  are used to denote spaces of polynomials (resp.~formal power series), and single $(\,. \,)$,  respectively double $((\, . \, ))$ 
round brackets denote the space of rational functions (resp.~formal Laurent series).
 For example $\Lb(x)[[\gamma]]$ is the set of formal power series in $\gamma$ whose coefficients are rational functions of $x$ over the basefield $\Lb$.
 The usual pair of infinite sequences of 2D Toda flow parameters 
  \be
 {\bf t}:= (t_1, t_2, \dots), \quad {\bf s} = (s_1, s_2, \dots ) 
 \ee
 will be viewed here as ``bookkeeping'' parameters when using the 2D Toda $\tau$-function \cite{Ta, UTa, Takeb} as a generating series for weighted Hurwitz numbers.   We use the standard notation $\{s_\lambda,   e_\lambda, h_\lambda, p_\lambda, m_\lambda, f_\lambda\}$ for the six standard bases for the space of symmetric functions:  Schur functions,   elementary and complete symmetric functions, power sum symmetric  functions, monomial symmetric functions and ``forgotten'' symmetric functions, respectively. (See, e.g.~\cite{Stanley:EC2, Mac}.) We view all symmetric functions as expressed in terms of the scaled power sums 
 \be
 t_i :=\tfrac{p_i}{i}, \quad s_i := \tfrac{p'_i}{i},
 \label{flow_variables_power_sums}
 \ee
 which play the roles of KP and 2D Toda flow variables in the $\tau$-function. For example, the notation $\{s_\lambda({\bf t})\}$ means the Schur functions expressed as polynomials in terms of the quantities ${\bf t}= (t_1, t_2, \dots)$. The Cauchy-Littlewood generating function expression for these is then
\be
e^{\sum_{i=1}^\infty i t_i s_i} = \sum_{\lambda}s_\lambda({\bf t}) s_\lambda({\bf s}) ,
\ee
where the sum is over all integer partitions.

%%%%%%%%%%%%%%%% Subsection 2.1 Weighted Hurwitz numbers %%%%%%%%%%%%%%%%

\subsection{Weighted Hurwitz numbers}

Multiparametric weighted Hurwitz numbers, as introduced in \cite{GH1, GH2, HO2, H1, H2}  are determined by weight generating functions
\be
G(z) = 1 + \sum_{i=1}^\infty g_i z^i,
\ee
which may also be expressed as infinite products
\be
G(z) =\prod_{i=1}^\infty(1+c_i z)
\label{G_inf_product}
\ee
or limits thereof, in terms of an infinite set of parameters ${\bf c} = (c_1, c_2, \dots)$. A {\em dual} class of weight generating
functions 
\be
\tilde{G}(z) =\prod_{i=1}^\infty(1-c_i z)^{-1}
\label{tilde_G_inf_product}
\ee
is also used in applications \cite{GH1, GH2, H2, HO2}, but will not be considered here.

Choosing a nonvanishing  small parameter $\beta$, we define the {\em content product } coefficients as
\be
r_\lambda^{(G, \beta)}  \deq   \prod_{(i,j)\in \lambda} r_{j-i}^{(G, \beta)},
\label{r_lambda_G}
 \ee
where
\be
r_j^{(G, \beta)} := G(j \beta ) 
\label{rj_beta}
\ee
and  $(i,j)\in \lambda$ refers to the position of a box in the Young diagram of the partition $\lambda$ 
 in matrix index notation.

Introducing a further nonvanishing parameter $\gamma$, it is convenient to express
these as consecutive ratios
\be
r^{(G, \beta)}_j = {\rho_j\over \gamma \rho_{j-1}} 
\ee
of a sequence of auxiliary coefficients  $\rho_j$  that are finite products of the  $ \gamma G(i\beta )$'s and their inverses  \cite{GH2, H2},
normalized such that $\rho_0=1$
\bea
\rho_j &\&:= \gamma^j\prod_{i=1}^j G(i\beta) , \quad \rho_0 := 1 \cr
\quad \rho_{-j} &\&:= \gamma^{-j} \prod_{i=0}^{j-1} (G(-i \beta))^{-1}  , \quad j=  1, 2, \dots.
\label{rho_j_gamma_G}
\eea
and 
\be
e^{T_j}:=\rho_j.
\ee

In most of the analysis below, the weight generating function $G$ will  be  chosen as a polynomial of degree $M$, 
and hence only the first $M$ parameters $(c_1, c_2, \dots, c_M )$ are taken as nonvanishing
\be
G(z)=1+\sum_{k=1}^M g_k z^k=\prod_{i=1}^M (1+c_i z).
\label{eq:Gpolynomial}
\ee
The coefficients $\{g_j\}_{j=1, \dots, M}$ are then just the elementary symmetric polynomials $\{e_j({\bf c})\}_{j=1, \dots, M}$
 in the parameters ${\bf c} =(c_1, c_2, \dots, c_M)$. We  denote by $\Kb =  \Qb[g_1,\dots, g_M] $
the algebra of polynomials in the $g_k$'s, with rational coefficients or, equivalently,
the algebra of symmetric functions of the $c_i$'s. 

\begin{definition}
For a set of partitions $\{\mu^{(i)} \}_{i=1,\dots, k}$ of weight $|\mu^{(i)}|=N$, 
the pure Hurwitz numbers  $H(\mu^{(1)}, \dots, \mu^{(k)})$ are  defined geometrically  \cite{Hu1, Hu2}  as the number
of inequivalent $N$-fold branched coverings  $\CC \ra \Pb^1$  of the Riemann sphere with $k$ branch points $(Q^{(1)}, \dots, Q^{(k)})$, whose ramification profiles are given by the partitions $\{\mu^{(1)}, \dots, \mu^{(k)}\}$, 
normalized by the inverse  $1/|\aut (\CC)|$ of the order of the automorphism group of the covering. 
\end{definition}
\begin{definition}
 An equivalent combinatorial/group theoretical  definition \cite{Frob1, Frob2, Sch} is that  $H(\mu^{(1)}, \dots, \mu^{(k)})$ is 
 ${1/ N !}$ times the  number of distinct   factorizations of the identity  element $\Ib \in \mathfrak{S}_N$ in the symmetric 
 group into a product of $k$ factors $h_i$, belonging to the conjugacy classes $\cyc(\mu^{(i)})$
\be
\Ib = h_1 \cdots h_k, \quad h_i \in \cyc(\mu^{(i)}).
\label{hurwitz_factoriz}
\ee
\end{definition}
The equivalence of the two follows from  the monodromy homomorphism from the fundamental group of 
$\Pb^1/\{Q^{(1)}, \dots, Q^{(k)}\}$, the Riemann sphere punctured at the branch points, into $\mathfrak{S}_N$,
 obtained by lifting closed loops from the base to the covering.

Denoting by
\be
\ell^*(\mu) := |\mu| - \ell(\mu)
\label{colength_mu}
\ee
the {\em colength} of the partition $\mu$ (the difference between its weight and length),
the Riemann-Hurwitz theorem relates the Euler characteristic $\chi$ of the covering curve to the sum of the colengths 
$\{\ell^*(\mu^{(i)})\}$ of the ramification profiles at the branch points as follows:
\be
\chi = 2-2g  = 2N - d, 
\label{riemann_hurwitz}
\ee
where
\be
d:= \sum_{i=1}^k \ell^*(\mu^{(i)}).
\label{d_def}
\ee

\begin{definition}
Given a pair of partitions $(\mu, \nu)$ of $N$, the \emph{weighted double  Hurwitz} number $H^d_G(\mu, \nu)$ 
with weight generating function $G(z)$ is defined as the weighted sum
   \bea
H^d_G(\mu, \nu) &\&\deq \sum_{k=0}^d \sideset{}{'}\sum_{\substack{\mu^{(1)}, \dots, \mu^{(k)} \\ |\mu^{(i)}| = N \\ \sum_{i=1}^k \ell^*(\mu^{(i)})= d}}
\WW_G(\mu^{(1)}, \dots, \mu^{(k)})H(\mu^{(1)}, \dots, \mu^{(k)}, \mu, \nu) ,
\label{Hd_G}
\eea
where $\sideset{}{'}\sum$ denotes a sum over all  $k$-tuples of partitions 
$\{\mu^{(1)}, \dots, \mu^{(k)}\}$ of $N$ other than the cycle type of the identity element $(1^N)$ and
 the weights $\WW_G(\mu^{(1)}, \dots, \mu^{(k)})$ are given by
 \be
 \WW_G(\mu^{(1)}, \dots, \mu^{(k)}):=
 {1\over k!}\sum_{\sigma \in \mathfrak{S}_{k}} \sum_{1 \le b_1 < \cdots < b_{k }} 
 c_{b_{\sigma(1)}}^{\ell^*(\mu^{(1)})} \cdots c_{b_{\sigma(k)}}^{\ell^*(\mu^{(k)})} = {|\aut(\lambda)|\over k!} m_\lambda ({\bf c}).
 \label{Wg_def}
 \ee
Here  $m_\lambda ({\bf c}) $ is the monomial symmetric function of the parameters ${\bf c}:= (c_1, c_2, \dots)$
\be
m_\lambda ({\bf c}) = {1\over |\aut(\lambda)|}\sum_{\sigma \in \mathfrak{S}_{k}} \sum_{1 \le b_1 < \cdots < b_{k }}
 c_{b_{\sigma(1)}}^{\lambda_1} \cdots c_{b_{\sigma(k)}}^{\lambda_{k}} ,
  \label{m_lambda}
\ee
 indexed by the partition $\lambda$ of weight $|\lambda|=d$ and length  $\ell(\lambda) =k$, whose 
parts $\{\lambda_i\}$  are equal to the colengths $\{\ell^*(\mu^{(i)})\}$ (expressed in weakly decreasing order), 
\be
\{\lambda_i\}_{i=1, \dots k} \sim \{\ell^*(\mu^{(i)})\}_{i=1, \dots k}
\ee
where the symbol $\sim$ is equality of multisets, and
 \be
 |\aut(\lambda)|:=\prod_{i\geq 1} m_i(\lambda)!
 \ee
where $m_i(\lambda)$ is the number of parts of $\lambda$ equal to~$i$. 
\end{definition}

%%%%%%%%%%%%%%%% Subsection 2.2 The $\tau$ function  %%%%%%%%%%%%%%%%

\subsection{The $\tau$-functions $ \tau^{(G, \beta, \gamma)} ({\bf t}, {\bf s}) $ }

Following \cite{GH1, GH2, HO2, H2}, we introduce a parametric family $\tau^{(G, \beta, \gamma)}({\bf t},{\bf s})$ 
of 2D Toda $\tau$-functions  of hypergeometric type   \cite{ KMMM, Or1, Or2, OrSc1, OrSc2} (at the lattice point $0$)
associated to the weight generating function $G(z)$ defined by the double Schur function series
\begin{eqnarray}
\tau^{(G, \beta, \gamma)}({\bf t},{\bf s})
:= \sum_\lambda \gamma^{|\lambda|} r^{(G, \beta)}_\lambda s_\lambda({\bf t}) s_\lambda({\bf s}),
\label{tau_schur_exp}
\end{eqnarray}
where ${\bf t} =(t_1, t_2, \dots)$, ${\bf s}= (s_1, s_2, \dots)$ are the two sets of 2D Toda flow parameters
and the sum is taken over all integer partitions (including $\lambda=\emptyset$).
These will serve as generating functions for the weighted double Hurwitz numbers as explained below.
\begin{remark}
For polynomial generating functions $G$, the $\tau$-function $\tau^{(G, \beta, \gamma)} ({\bf t},{\bf s})$
is viewed in the following as an element of $\Kb[{\bf t},{\bf s},\beta][[\gamma]]$.
\end{remark}
Making a change of basis  from the Schur functions to the power sum symmetric functions 
\be
p_\mu({\bf t}):= \prod_{i=1}^{\ell(\mu) }p_{\mu_i} =  \prod_{i=1}^{\ell(\mu) } \mu_i t_{\mu_i}, \quad  p_\nu({\bf s}) := \prod_{i=1}^{\ell(\nu)} p'_{\nu_i} =  \prod_{i=1}^{\ell(\nu)} \nu_i s_{\nu_i},
\ee 
 using the Frobenius character formula \cite{FH, Sag, Mac},
\be
s_\lambda = \sum_{\mu, \, \abs{\mu} = \abs{\lambda}} z_\mu^{-1} \chi_\lambda(\mu) p_\mu,
\label{frobenius_character}
\ee
where $\chi_\lambda(\mu)$ is the character of the irreducible representation of symmetry
type $\lambda$ evaluated on the conjugacy class of cycle type $\mu$ and
\be
z_\mu = \prod_{i=1}^{\abs{\mu}} i^{m_i} (m_i)!, \qquad m_i = \text{number of parts of $\mu$ equal to $i$}, 
\ee
$  \tau^{(G, \beta, \gamma)} ({\bf t}, {\bf s})$  may equivalently be expressed as a double series
in the power sum symmetric functions, whose coefficients are equal to the $H^d_G(\mu, \nu)$'s (see ~\cite{GH2, H2} for 
details).
\begin{theorem}[\cite{GH2,  H2}]
\label{prop:tauHurwitz}
The function $\tau^{(G, \beta, \gamma)} ({\bf t},{\bf s})\in \Kb[{\bf t},{\bf s},\beta][[\gamma]]$ has the equivalent series expansion
\bea
  \tau^{(G, \beta, \gamma)} ({\bf t}, {\bf s}) &\& =  
\sum_{\substack{\mu, \nu \\ |\mu|=|\nu|}} \gamma^{|\mu|}
\sum_{d=0}^\infty \beta^d  H^d_G(\mu, \nu) p_\mu({\bf t}) p_\nu({\bf s}).
\label{tau_G_H}
\eea
\end{theorem}
Thus $\tau^{(G, \beta, \gamma)}({\bf t},{\bf s})$ is interpretable as a generating function for weighted double Hurwitz numbers $H^d_G(\mu, \nu)$, with the exponents of the variables $\gamma$ and $\beta$ equal  to the quantities $N=|\mu|=|\nu|$ 
and $d$,  as defined in eq.~(\ref{d_def}), respectively.
\begin{remark}
 Applications of particular cases of hypergeometric $\tau$-functions to Hurwitz numbers were studied in   \cite{Pa, Ok, GGN1, GGN2, AC1, AC2, Z, KZ, AMMN, AMMN1, GH1, GH2, H2, HO2, HOrt} and other applications elsewhere  \cite{BAW, NT, HO1}.
Adding a further integer index $n$ to the definition (\ref{r_lambda_G}) of the content product coefficients 
by replacing $j-i \ra n + j-i$ and making the corresponding substitution in eq.~(\ref{tau_schur_exp}),  as in refs.~\cite{GH2, HO2},
we obtain a lattice index $n$ on the $\tau$-function  in addition to the two continuous infinite sets of flow parameters  ${\bf t}$ and ${\bf s}$.
This defines a sequence of $\tau$-functions of the 2D Toda lattice hierarchy \cite{Ta, UTa, Takeb}, each
 of which satisfies the  dynamics of a pair of independent KP hierarchies in the  ${\bf t}$ and ${\bf s}$ flow parameters, 
 as well as the lattice equations.  \end{remark}

%%%%%%%%%%%%%% Section 2.3. Convolution action and dressing: adapted bases  %%%%%%%%%%%%

\subsection{Convolution action and dressing: adapted bases}
\label{convolution_action}

In the analytic model used in \cite{SW}, we consider the Hilbert space $\HH= L^2(S^1)$
whose elements are Fourier series $\{\sum_{i\in \Zb} f_i \zeta^i\}$ on the unit circle in the complex $\zeta$-plane \{$|\zeta |=1\}$, with the usual splitting
\bea
\HH &\&= \HH_+ + \HH_-  \cr
\HH_+ &\&:= \span\{\zeta^i\}_{i\in \Nb}, \quad \HH_-:= \span\{\zeta^{-i}\}_{i\in \Nb^+}
\eea
and complex inner product
\be
\langle f, g \rangle := {1\over 2\pi i}\oint_{S^1}f(\zeta) g(\zeta) {d\zeta}.
\label{dual_pairing_Hirota}
\ee
To define a  dual pairing, we can identify the analytic dual $\HH^*$ of $\HH$ with $\HH$ itself. The dual basis 
in $\HH^*$ corresponding to the monomial basis  $\{\zeta^i\}_{i\in \Zb}$  in $\HH$ is then 
 $\{\zeta^{-i-1}\}_{i \in \Zb}$.

In the formal series setting, we replace the spaces $\HH$ and $\HH^*$ by their formal
analog $\Cb((\zeta))$, viewed as semi-infinite formal Laurent series in $\zeta$,
and the (Hirota) inner product (\ref{dual_pairing_Hirota}) by the {\em formal residue} (i.e., the coefficient of ${1 \over \zeta}$).

For genuine (convergent) Fourier series, $f(\zeta), g(\zeta)$, the convolution product is defined by
\be
f*g (\zeta) = {1\over 2\pi i} \oint_{\xi \in S_1} f(\xi) g\left({\zeta \over \xi}\right) {d\xi\over \xi}.
\ee
and we have the following formal representation in terms of power series
\bea
f(\zeta) &\&= \sum_{i\in \Zb} f_i \zeta^{-i-1}, \quad g(z) = \sum_{i\in \Zb} g_i \zeta^{-i-1} \cr
f*g (\zeta) &\& = \sum_{i\in \Zb} f_i g_i \zeta^{-i-1} .
\eea

Three infinite abelian group actions on $\HH$, or its formal analog,  enter in the definition of 
$\tau$-functions of hypergeometric type. First, there are the two abelian groups of ``shift flows''
\be
\Gamma_+ =\{\gamma_+({\bf t}) := e^{\sum_{i=1} ^\infty t_i \zeta^i} \}, 
\quad \Gamma_- =\{\gamma_-({\bf s}):= e^{\sum_{i=1}^\infty s_i \zeta^{-i}} \},
\ee
which act by multiplication
\bea
\Gamma_{\pm} \times \HH &\&\ra \HH \cr
(\gamma_{\pm}, f) &\& \mapsto \gamma_{\pm} \,f.
\eea
We also have the semigroup of convolution actions $\CC=\{C_\rho\}$,
 defined by the convolution product
\bea
\CC\times \HH &\& \ra \HH \cr
(C_\rho, f) &\& \mapsto \rho * f .
\eea
with elements  $\rho(\zeta)$ that admit a distributional (or formal) Fourier series expansion
\be
\rho(\zeta) := \sum_{i\in \Zb} \rho_i \zeta^{-i-1}
\ee

Letting
\be
x :=1/\zeta,
\ee
and applying the $\Gamma_{\pm}$ and $\CC$ actions to the monomial basis $\{\zeta^k\}_{k \in \Zb}$ for $\HH$ or $\Cb((\zeta))$,
we define the ``dressed'' basis  $\{\Psi^+_k(\zeta)\}_{k\in \Zb}$ for $\HH$ as
\be
\Psi^+_{k}(x= 1/\zeta)  := C_\rho(\gamma_-(\beta^{-1} {\bf s}) (\zeta^{-k})) = \gamma \sum_{j= k}^{\infty} 
 \rho_{j-1} h_{j-k}(\beta^{-1} {\bf s}) x^j.
\label{Psi+_k}
\ee
 Under the pairing (\ref{dual_pairing_Hirota}), the dual basis $\{\Psi^-_k(\zeta)\}_{k\in \Zb}$ for $\HH^*$ (or $\Cb((\zeta))$) is given by
\be
\Psi^-_{k}(x=1/\zeta):= \gamma_-(-\beta^{-1} {\bf s})  (C^{-1}_\rho (\zeta^{-k})) = \sum_{j= k}^{\infty}  \rho^{-1}_{-j} h_{j-k}(-\beta^{-1} {\bf s}) x^j.
\label{Psi-_k}
\ee

These are dual  in the sense that 
\be
 \langle \Psi^+_j, \Psi^-_{k}\rangle  = \gamma \delta_{j, -k+1}.
 \label{hirota_dual}
\ee
In particular all $\Psi^+_j$ for $j\le0$ are orthogonal to all $\Psi^-_k$ for $k\le0$, 
and this is equivalent to the Hirota bilinear equation for the KP hierarchy with respect to the times ${\bf t}$.
A representation of these adapted  bases as fermionic vacuum expectation values (VEV's)
is given in the companion paper \cite{ACEH2}, together with proofs of a number of
their properties.

 In terms of the infinite Grassmannians \cite{SS, Sa, SW, JM, JMD},
the element $W^{(G, \beta, \gamma, {\bf s})}$ that corresponds to the 
$\tau$-function $\tau^{(G, \beta, \gamma)}({\bf t}, \beta^{-1}{\bf s})$
is spanned by the basis elements $\{\Psi^+_{-k}\}_{k\in \Nb}$, whereas $\span(\Psi^-_{-k})_{k\in \Nb}$
is the element of the dual Grassmannnian given by its annihilator $W^{(G, \beta, \gamma, {\bf s})\perp}$
under the pairing (\ref{dual_pairing_Hirota}).

%%%%%%%%%%%%%%%% Section 3. Constellations %%%%%%%%%%%%%%%%

\section{Constellations }
\label{sec:constellations}

%%%%%%%%%%%%%%%% Subsubsection  3.1 Constellations and branched covers %%%%%%%%%%%%%%%%

\subsection{Constellations and branched covers} 

We give here another interpretation of  $\tau^{(G, \beta, \gamma)}({\bf t},{\bf s})$ as  generating function of certain embedded weighted, bipartite graphs on surfaces (introduced in~\cite{LZ}),  called {\em constellations}. Several variants of the same graphical model are used in the combinatorial literature  (see  \cite{LZ, MBM-Schaeffer, Poulalhon-Schaeffer:factorization, Chapuy:constellations} for background).  Here, we enhance the graphical definition by attributing suitable weights to the vertices and edges, with the weight for a given constellation obtained  by multiplying all the vertex and edge weights. 
 Vertices of the first type, called ``coloured'',  correspond to the ramification points
of the branched cover. These are attributed a ``colour'' that determines their weight, as well as the weight of the edges that connect them to the other type of vertices. The latter  are called ``star'' vertices and correspond to $N$ points over an arbitrarily chosen generic (non-branching) point.  Constellations  give a combinatorial interpretation of factorizations of the form~\eqref{hurwitz_factoriz} or, equivalently,  branched covers of the Riemann sphere $\Pb^1$. 
\smallskip

Assuming the weight generating function $G(z)$ is a polynomial of degree $M$ (as in (\ref{eq:Gpolynomial})), we
may start with a slightly different expression for the weighted Hurwitz number $H_G^d(\mu)$ defined in~\eqref{Hd_G}:
 \bea
H^d_G(\mu, \nu) &\&= \sum_{\substack{\mu^{(1)}, \dots, \mu^{(M)} \\ |\mu^{(i)}| = N \\ \sum_{i=1}^M \ell^*(\mu^{(i)})= d}}
c_1^{\ell^*(\mu^{(1)})}\dots c_M^{\ell^*(\mu^{(M)})} H(\mu^{(1)}, \dots, \mu^{(M)}, \mu, \nu).
\label{Hd_G2}
\eea
The difference between the two is that in \eqref{Hd_G2}   ``trivial'' ramification profiles $(1)^N$
are allowed  at each position, whereas in \eqref{Hd_G} this is only allowed for $\mu$ and $\nu$. Therefore the number $k$ of non-trivial profiles is not specified. Since there are only $M$ nonvanishing $c_i$'s, where $M$ is the degree of the weighting polynomial $G(z)$, this is the maximum number of branch points allowed (besides $0$ and $\infty$).
 The equivalence between the two formulae follows from the following facts
\begin{enumerate}
\item[-] The pure Hurwitz number in which $M-k$ of the branching profiles are $(1)^N$ equals the one
in which these are omitted:
\be
H(\mu^{(1)}, \cdots, \mu^{(k)}, \underbrace{(1)^N, \cdots, (1)^N}_{M-k}) = H(\mu^{(1)}, \cdots, \mu^{(k)}) 
\ee
\item[-] The pure Hurwitz numbers $H(\mu^{(1)}, \cdots, \mu^{(k)})$ are invariant under permutations 
of the ordering of the partitions.
\end{enumerate}
 Starting with \eqref{Hd_G2}, if the summation is refined to indicate the  number $k$ of indices $i\in[1, \dots,M]$ 
for which $\mu^{(i)}\neq (1)^N$, $b_1<\dots<b_k$ are the $k$ indices appearing in these,  in increasing order, 
 and  $\lambda$ is the partition whose parts are equal to their colengths $\{\ell^*(\mu^{(b_i)})\}$, we 
obtain~\eqref{Hd_G}, with the indices $\{b_i\}_{i=1, \dots, k}$ in \eqref{m_lambda} summed also 
over all permutation in $\mathfrak{S}_N$ .
\begin{remark}
	As will be seen, the interest of constellations is to reformulate weighted Hurwitz numbers as counting functions for a family of graphs with \emph{local} weights. As a starting point for this section, it is more natural to work with~\eqref{Hd_G2}, rather than \eqref{Hd_G},  since the requirement that $\mu^{(i)}\neq 1^N$ would lead to a nonlocal constraint on the graphical model.
\end{remark}

\begin{remark}
	We  define $(M+2)$-constellations for $M=1, 2, \dots$, motivated by the fact that we are working with \emph{double} Hurwitz numbers; i.e., there are $M+2$ partitions appearing in the Hurwitz number in the RHS of~\eqref{Hd_G2}.
\end{remark}

\begin{definition}\label{def:constellations}
An \emph{$M+2$-constellation} of size $N$ is a graph embedded in a compact oriented surface, in such a way that each face is homeomorphic to a disk, considered up to oriented homeomorphism. It consists of the following data and constraints: 
\begin{itemize}[itemsep=0pt,topsep=1pt,leftmargin=20pt, parsep=0pt]
\item[-] There are two kinds of vertices: \emph{star vertices}, of which there are $N$ in total,  numbered consecutively from $1$ to $N$, and \emph{coloured vertices}. Each coloured vertex carries a \emph{colour}, labelled by $M+2$ indices $(0, \dots, M, M+1)$, but several different vertices can have the same colour.
\item[-] Each edge links a star vertex to a coloured vertex.
\item[-] Each star vertex has degree $M+2$, and the sequence of colours of its neighbours in counterclockwise order is $0,1,\dots,M, M+1$.
\item[-] Each face contains exactly one angular sector of each colour (equivalently it is bounded by $2(M+2)$ edge sides).
\end{itemize}
(See Figure~\ref{fig:exampleConstellationBis} for an example where $N=5, M=3$.)
\end{definition}

The colour $0$ and the last colour, $M+1$, play a special role; anticipating the covering interpretation given in Section~\ref{const_coverings}, we  also denote the last colour $M+1$ as $\infty$ (see Figure~\ref{fig:exampleConstellationBis}).

We now explain the relation between constellations and Hurwitz numbers. Given an $M+2$-constellation of size $N$, define the permutations 
$(h_0,h_1,\dots, h_M, h_{M+1})$ in  $\mathfrak{S}_N$ by  $h_i(j):=j'$,
where $j'$  is the label of the star vertex following the star vertex
 labelled $j$ clockwise around its unique neighbour of colour $i$. In other words, each cycle of the permutation $h_i$ gives the clockwise order of appearance of star vertices around a vertex of colour $i$. It is easy to see that the  distinct cycles of  the elements  $\{h_0, h_1, \dots , h_{M+1}\}$ are in one-to-one correspondence with the \emph{faces} of the graph on the surface and that the product $h_0h_1\dots h_{M+1}$ is the identity. (See Figure~\ref{fig:exampleConstellationBis}.)
Clearly this construction can be inverted;  given an $(M+2)$-tuple of permutations whose product is the identity, one can reconstruct a unique embedded graph as in Definition~\ref{def:constellations} by gluing together vertices in accordance with the rules.  Therefore we have

\begin{lemma}[{see \cite[Chapter 1]{LZ}}]\label{lemma:HurwitzConstellations}
The Hurwitz number $H(\mu^{(1)}, \dots, \mu^{(M)}, \mu, \nu)$ is $\frac{1}{N!}$ times the number of constellations with $N$ star vertices such that  the  partition of $N$  giving the degrees of the vertices of colour $i$ is $\mu^{(i)}$ for $0\leq i \leq M+1$, with $\mu^{(0)}=\mu$, $\mu^{(M+1)}=\nu$. 
\end{lemma}

Note that the Euler characteristic $\chi$ of the surface is given by the Riemann-Hurwitz formula  (cf. eq.~(\ref{riemann_hurwitz}))
\be
\chi = 2-2g  = \ell(\mu) + \ell(\nu) - \sum_{i=1}^M \ell^*(\mu^{(i)}). 
\label{riemann_hurwitz_mu_nu}
\ee

  %%%%%%%%%%%%%%%% Figure 2 figure exampleConstellationBis %%%%%%%%%%%%%%%%%%%
 \begin{figure}
 \label{fig:exampleConstellationBis}
\begin{center}
\includegraphics{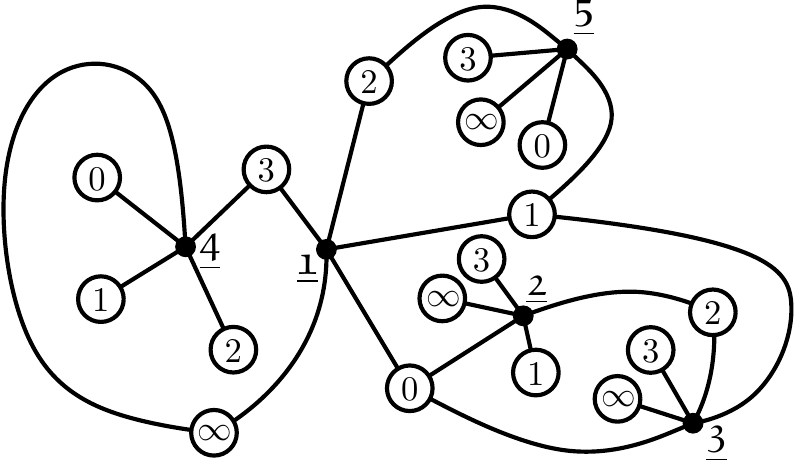}
\end{center}
\caption{\footnotesize  An example of a  constellation with $N=5$, $M=3$. We use $\infty$ to denote the last colour $M+1=4$. 
The example corresponds to the factorization $h_0h_1h_2h_3h_4=1$ with 
$h_0 = (123),  \ h_1  = (153), \ h_2 =  (15) (23), \ h_3 = (14),  \ 
      h_4 = h_\infty = (14)$, with corresponding partitions 
 $\mu^{(1)} = (3, 1, 1)),  \ \mu^{(2)} = (2, 2, 1), \  \mu^{(3)} = (2, 1,1,1),\ 
\mu:= \mu^{(0)}  = (3, 1, 1) ), \  \nu := \mu^{(4)}= (2, 1, 1, 1)$.
 In the picture elements of the ground set $\{1,\dots,N\}$ are indicated with underlined numbers, while numbers corresponding to colours are not underlined. This example has genus $g=0$.}
\end{figure}

%%%%%%%%%%%%%%%%  Subsection 3.2 Weighted constellations %%%%%%%%%%%%%%%%%%%%

\subsection{Weighted constellations}

  We now turn to the weighting of the vertices and edges of a given constellation that will enable us to view the $\tau$-function as a generating function for constellations. Recall that the colours are labelled $1,\dots, M$ where $M$ is the degree of $G(z)$, plus two special colours, white and black, corresponding to $0$ and $\infty$, respectively,  in the branched covering surface interpretation. (See subsection \ref{const_coverings} for details of this correspondence and Figures~\ref{fig:weightsVerticesAllTypes}, 
  \ref{fig:weightsVerticesAllTypesCover} for a visualization of the weighting scheme.)
  
\begin{itemize}[itemsep=0pt,topsep=1pt,leftmargin=20pt, parsep=0pt]
	\item[-]
{ \bf Coloured  vertices.} 
 To the ordinary coloured vertices of colour $i$ for $1\leq i\leq M$,  we assign a weight $(\beta c_i)^{-1}$. To the edges linking them to the star vertices, we assign a weight~$\beta c_i$.
\item[-] {\bf White vertices.} 
	To the white vertices (of colour $0$),  we assign the power sum symmetric functions $p_{j} = j t_{j}$, where $j$ is the vertex degree. The edges that connect them to the star vertices have weight $1$. \item[-] {\bf Black vertices.} 
	To the black vertices (of colour $\infty$),  we assign the power sum symmetric functions $p_{j} = j s_{j}$, where $j$ is the vertex degree. The edges that connect them to the star vertices have weight $1$.   \item[-]
{\bf Star vertices} Finally, to each of the $N$ star vertices we assign a
 weight $\gamma$.

 \end{itemize}

 %%%%%%%%%%%%%%%%%%%%%%% Figure 3 weightsSpecialVertex0Bis  %%%%%%%%%%%%%%%
 \begin{figure}
 \begin{minipage}{0.35\linewidth}
  \begin{center}
   \includegraphics[width=\linewidth]{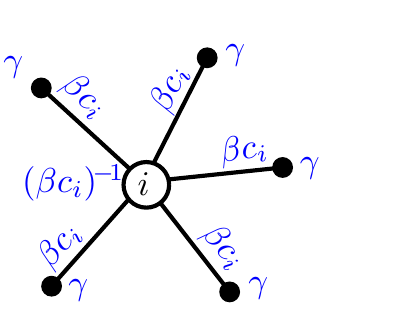}
  \end{center}
 \end{minipage}
 \begin{minipage}{0.315\linewidth}
  \begin{center}
   \includegraphics[width=\linewidth]{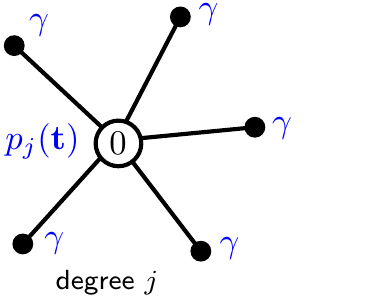}
  \end{center}
 \end{minipage}
 \begin{minipage}{0.315\linewidth}
  \begin{center}
   \includegraphics[width=\linewidth]{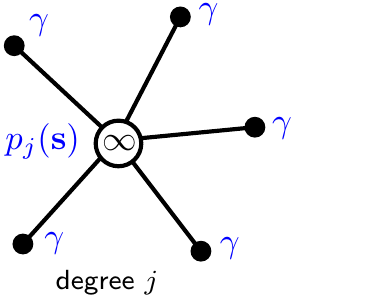}
  \end{center}
 \end{minipage}

\caption{\footnotesize  Weights. Left: coloured vertices and edges. Center: white vertices. Right: black vertices. 
The vertex degree is $j=5$ on the center and right examples.}
	 \label{fig:weightsVerticesAllTypes}
\end{figure}

 %%%%%%%%%%%%%%%%  Subsection 3.3 Reconstructing the $\tau$-function %%%%%%%%%%%%%%%%%%%%

\subsection{Reconstructing the $\tau$-function}

Taking the product over all the weights gives the total weight for the constellation as
\be
\gamma^N \beta^d \prod_{i=1}^M (c_i)^{\ell^*(\mu^{(i)})}p_\mu({\bf t}) p_\nu({\bf s}), 
\ee
where $d$ is given by~\eqref{d_def}. 
Summing the contributions of all constellations, and using~\eqref{Hd_G2} and Lemma~\ref{lemma:HurwitzConstellations}, we recognise the expression \eqref{tau_G_H} for the $\tau$-function and obtain: 

\begin{proposition}\label{prop:tauConstellations}
The function $\tau^{(G, \beta, \gamma)} ({\bf t},{\bf s})\in \Kb[{\bf t},{\bf s},\beta][[\gamma]]$ is the exponential generating function of weighted $(M+2)$-constellations. 
\end{proposition}

Recall that here, ``exponential'' generating function means that each constellation is counted with an extra factor of $\frac{1}{N!}$.
Since each constellation can be uniquely decomposed into connected components, the exp-log principle (see e.g.~\cite{Flajolet}) ensures 
that $\ln\, \tau^{(G, \beta, \gamma)}({\bf t},{\bf s})$ is the generating function of \emph{connected} constellations, with the same weighting scheme as in Proposition~\ref{prop:tauConstellations}. (A constellation is connected if and only if its underlying surface is.) Equivalently, in the 
notation of weighted Hurwitz numbers we have:
\be
\ln\,  \tau^{(G, \beta, \gamma)} ({\bf t},{\bf s}):= \sum_{\mu,\nu \atop |\mu|=|\nu|}  \gamma^{|\mu|} \sum_d  \beta^{d} \tilde H_{G}^d(\mu,\nu)\, p_\mu({\bf t}) p_\nu({\bf s}).
\ee
where $\tilde H_{G}^d(\mu,\nu)$ denotes the \emph{connected} weighted Hurwitz number (see~\cite{GH2, H2, ACEH1, ACEH2}).

 %%%%%%%%%%%%%%%%  Subsection 3.4 Direct correspondence between constellations and coverings %%%%%%%%%% 
\subsection{Direct correspondence between constellations and coverings}
\label{const_coverings}

For completeness, we recall in this section how the link between constellations and branched covers can be made directly, without relying on permutations. This well-known correspondence  (see again~\cite[Chapter 1]{LZ}) enables us to draw the graph directly on the surface of the covering and to interpret the weights of the different vertices as weights assigned to the ramification points.

 %%%%%%%%%%%%%%%%%%%%%%% Figure 3 weightsSpecialVertex0Bis  %%%%%%%%%%%%%%%
 \begin{figure}
 \begin{minipage}{0.32\linewidth}
  \begin{center}
   \includegraphics[width=\linewidth]{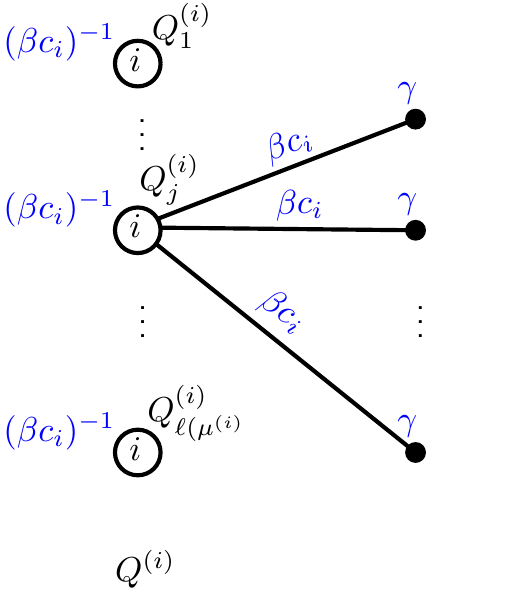}
  \end{center}
 \end{minipage}
 \begin{minipage}{0.32\linewidth}
  \begin{center}
   \includegraphics[width=\linewidth]{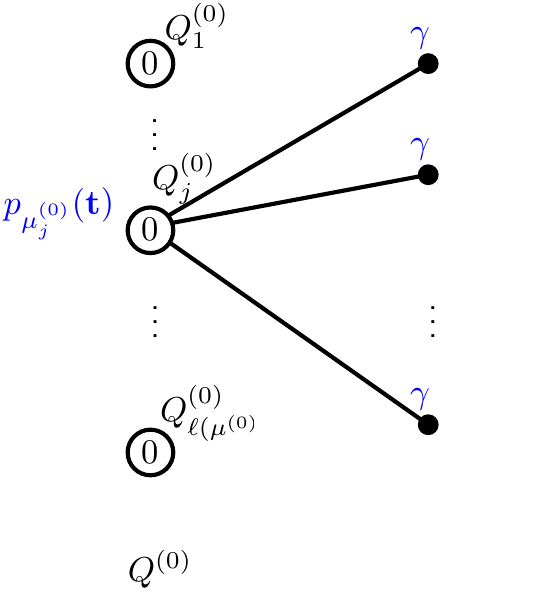}
  \end{center}
 \end{minipage}
 \begin{minipage}{0.32\linewidth}
  \begin{center}
   \includegraphics[width=\linewidth]{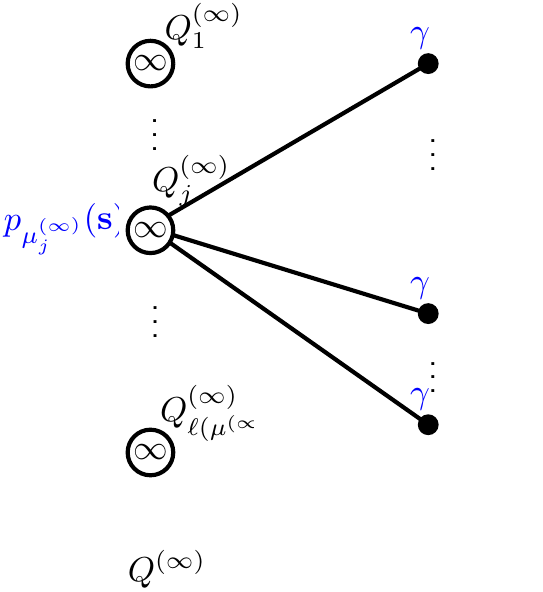}
  \end{center}
 \end{minipage}

\caption{\footnotesize  The weights of the vertices interpreted directly on the covering rather as on the abstract combinatorial objects. 	 }\label{fig:weightsVerticesAllTypesCover}
\end{figure}

Constellations correspond to branched covers with $(M+2)$ marked points $(Q^{(0)}, Q^{(1)},\dots, Q^{(M+1)})$, which include all the branch point of the cover -- however some of these points may be regular.
The vertices coloured $1,\dots,M$ in the constellation correspond to the points over $(Q^{(1)}, \dots, Q^{(M)})$, 
	while the two extra colours: white, for the points over $Q^{(0)}=0$ and  black for the points over $Q^{(M+1)}=Q^{(\infty)}=\infty$. They play a special role only because of the different  type of weights attached to them. To each equivalence class of branched coverings of $\Pb^1$ with $M+2$ marked points including all the branch points, there is associated a constellation that can be obtained as follows.

Given a branched cover of $\Pb^1$, with ordered marked points $(Q^{(0)}=0$, $Q^{(1)}, \dots, Q^{(M)}$, $Q^{(M+1)} =\infty)$ that include the branch points of the cover, and respective ``ramification'' profiles $(\mu^{(0)}= \mu, \mu^{(1)}, \dots, \mu^{(M)}, \mu^{(M+1)}=\nu)$, we associate a coloured vertex to each of the ``ramification'' points $\{Q^{(i)}_1, \dots Q^{(i)}_{\ell(\mu^{(i)})}\}$ over $Q^{(i)}$, $i=0, \dots, M+1$, with ``ramification'' indices $\{\mu^{(i)}_j\}_{j=1, \dots, \ell(\mu^{(i)})}$ and all points over the same base point  $Q^{(i)}$ having the same colour. Here we use quotes since the point $Q^{(i)}$ is allowed to be regular, in which case the ``ramification profile'' $\mu^{(i)}$ is equal to $[1^N]$. Choosing an arbitrary  generic (non-branch) point $P \in \Pb^1$, we order the $N$ points  $(P_1, \dots, P_N)$, in the covering curve over $P$ and associate a star vertex to each. Lifting the simple, positively oriented closed loop $\Gamma_i$ starting and ending at $P$ and looping once around the  point $Q^{(i)}$, for $i=0, \dots , M+1$, we obtain paths (see Figure \ref{path_lift_cycle})  that start at each of the $P_j$'s,  $j=1, \dots, N$ and end either at the same one (contributing nothing to the monodromy) or at another one, and closing in a cycle. Thus, each pair $(P_j, Q^{(i)}_j)$ that are so linked have a pair of half-loops connecting them: an incoming one, and an outgoing one, and a successive sequence that  closes at the starting point $P_j$ and forms a cycle. We put an edge connecting  $(P_j, Q^{(i)}_k)$ for each such pair for which there are loops that are part of a cycle. The product of the disjoint cycles over a given point  defines the corresponding element $h_i \in \mathfrak{S}_N$,  which has cycle lengths given by the parts of the partition $\mu^{(i)}$.
All ``ramification points'' $\{Q^{(i)}_j\}_{j=1 \dots, \ell(\mu^{(i)})}$ over a given point $Q^{(i)}$ correspond to a coloured vertex, with colour $i$, and each of the points $\{P_j\}_{j=1, \dots, N}$ corresponds to a star vertex. The resulting graph is the constellation corresponding to the given branched cover.

  %%%%%%%%%%%%%%%%%%%%%% Figure 1 figure pathABCBis %%%%%%%%%%%%%%%

 \begin{figure}[H]
 \label{path_lift_cycle}
\begin{center}
\includegraphics{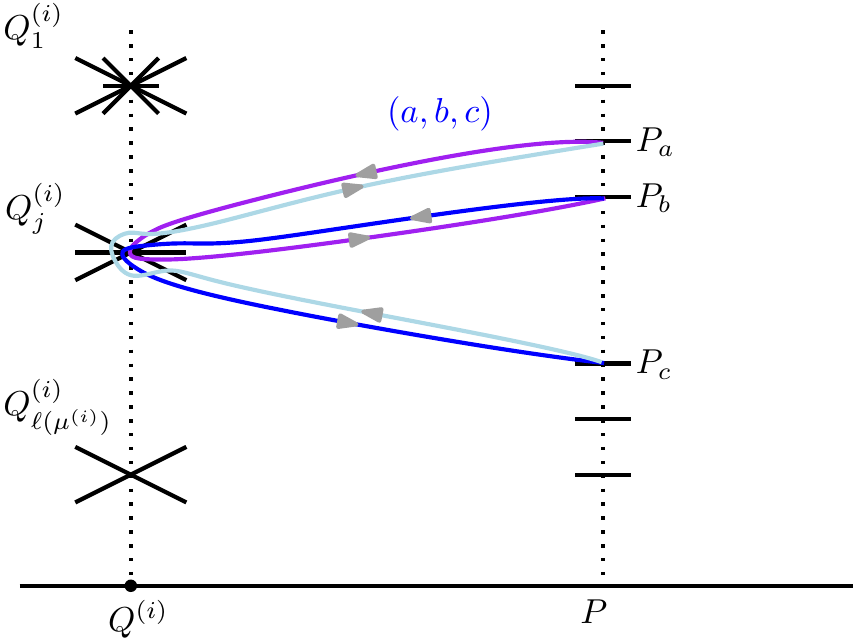}
\end{center}

	 \caption{\footnotesize Lifted loops and edges. Here, $(a,b,c)$ represents a (typical) cycle  in the monodromy factorization at the indicated branch point.}\label{fig:pathABCTer}

\end{figure}
\bigskip

%%%%%%%%%%%%%%%% Section 4. Fermionic and Bosonic functios %%%%%%%%%%%%%%%%

\section{Fermionic and bosonic functions }
\label{sec:fermionicAndBosonic}

In much of the following,  we denote $\tau^{(G, \beta, \gamma)} ({\bf t},\beta^{-1}{\bf s})$ simply as  $\tau ({\bf t})$,
viewing the parameters $(c_1, c_2, \dots )$ defining $G$ as well as  $\beta$, $\gamma$  and ${\bf s}$ all as auxiliary parameters.
The fact that this is the $n=0$ lattice point evaluation of a 2D-Toda $\tau$-function (see~\cite{ACEH2}) when the ${\bf s}$ variables are taken as independent
does not play any role in this section, and the second set of flow variables ${\bf s}=(s_1, s_2, \dots)$ are viewed as
auxiliary parameters in a KP $\tau$-functions, which serve as bookkeeping parameters in the generating function interpretation.
Most of the definitions and several of the results of the following three subsections are, in fact, valid for completely general functions $\tau({\bf t})$,
whether they are KP $\tau$-functions, or not. The sole exceptions are Propositions \ref{prop:Kn_determinant} and  \ref{prop:detConnected},  
which are valid for any KP $\tau$-function $\tau({\bf t})$, and Propositions \ref{F_n_gen_fn} and
\ref{prop:explicitFunctions}, which are only valid for the specific $\tau$-functions $\tau^{(G,\beta, \gamma)}({\bf t}, \beta^{-1}{\bf s})$ 
defined by eqs.~(\ref{tau_schur_exp}), (\ref{tau_G_H}).

%%%%%%%%%%%%%%%% Subsection 4.1 Fermionic functions} %%%%%%%%%%%%%%%%

\subsection{Fermionic functions}

The functions defined in this subsection will be referred to as  \emph{fermionic} functions. 
The definitions remain valid regardless of whether or not $\tau({\bf t})$ is a KP $\tau$-function, but 
Proposition \ref{prop:Kn_determinant} only holds true if it is (regardless of which KP $\tau$-function it is). 
In this case, their representation as vacuum expectation values (VEV's) of products of fermionic operators
is given in  \cite{ACEH2}.

The first of these is the Baker function and its dual,
%which will be denoted $\Psi_{(G, \beta, \gamma)}^-(\zeta,{\bf t}, {\bf s})$  and $\Psi_{(G, \beta, \gamma)}^+(\zeta, {\bf t}, {\bf s})$, respectively,
which are given by the Sato formulae~\eqref{sato_formula}.
\begin{definition}
We define the Baker function and its dual, denoted $\Psi_{(G, \beta, \gamma)}^-(\zeta,{\bf t}, {\bf s})$ and $\Psi_{(G, \beta, \gamma)}^+(\zeta, {\bf t}, {\bf s})$, respectively, by 
 \be
 \Psi_{(G, \beta, \gamma)}^\mp(\zeta, {\bf t}, {\bf s}) = e^{\pm \sum_{i=1}^\infty t_i \zeta^i }{\tau^{(G, \beta, \gamma)}({\bf t} \mp [x] , {\bf s}) \over  \tau^{(G, \beta, \gamma)}({\bf t}, {\bf s}) }
 = e^{\pm \sum_{i=1}^\infty t_i \zeta^i } \left(1 + \OO({1/ \zeta})\right),
\label{sato_formula} 
 \ee
 where  $\zeta$ is the spectral parameter
 \be
x := \zeta^{-1}
 \ee
and $[x]$ is the infinite vector whose components are equal to the terms
  of the Taylor expansion of $-\ln(1-x)$ at $x=0$
 \be
 [x]:= \left(x, {x^2\over 2}, \dots, {x^n \over n}, \dots\right).
 \ee
\end{definition}
We view the Baker function and its dual as associated to the family of KP
$\tau$-functions $\tau({\bf t})$  in the flow variables ${\bf t}= (t_1, t_2, \dots)$
with $(\beta, \gamma)$, ${\bf c}=(c_1, c_2, \dots)$ and ${\bf s }= (s_1, s_2, \dots)$  interpreted as parameters.
The expansion of the Baker function  $\Psi_{(G, \beta, \gamma)}^-(\zeta,{\bf t}, {\bf s})$  and its dual
$\Psi_{(G, \beta, \gamma)}^+(\zeta, {\bf t}, {\bf s})$ in terms of the bases $\{\Psi^+_{-k}\}_{k\in \Nb}$
and $\{\Psi^-_{-k}\}_{k\in \Nb}$ for the subspaces  $W^{(G, \beta, \gamma, {\bf s})\perp}$ and $W^{(G, \beta, \gamma, {\bf s})}$,
respectively, is given in the companion paper \cite{ACEH2}.

Note that, from the Sato  formulae (\ref{sato_formula}), 
and the expansion (\ref{tau_G_H}) of the $\tau$-function, we have the
following expansion of the  ${\bf t} = {\bf 0}$ values of $\Psi_{(G, \beta, \gamma)}^\pm(\zeta,{\bf t}, {\bf s})$
 as a Taylor series in the spectral parameter and a power sum series in the parameters ${\bf s}$, 
 with the weighted Hurwitz numbers as coefficients 
\begin{eqnarray}\label{eq:defPsi0}
\Psi_0^\pm(x)
&{:=}& 
\Psi_{(G, \beta, \gamma)}^\pm\left(x^{-1}, {\bf 0}, \beta^{-1}{\bf s}\right)
\  = \sum_{\mu,\nu}   \sum_d (x\gamma)^{|\mu|} (\pm1)^{\ell(\mu)} \beta^{d-\ell(\nu)} H_{G}^d(\mu,\nu)\, p_\nu(s).
\label{Phi_x_def}
\end{eqnarray}
More generally, for any KP $\tau$-function $\tau({\bf t})$, by Sato's formula, evaluation of  the Baker
function at  ${\bf t}= {\bf 0}$ is given by
\be
\Psi^\pm_0(x) = \tau(\pm [x]).
\label{Psi_KP_Tau_t0}
\ee

\begin {definition}
For an arbitrary KP $\tau$-function $\tau({\bf t})$ (in particular, for the choice
$\tau({\bf t}) = \tau^{(G, \beta, \gamma)}({\bf t}, \beta^{-1}{\bf s})$), the {\em pair correlator} is defined as
\begin{eqnarray}
K(x,x') 
&{:=}& \frac{1}{x-x'}\,
\tau \big([x]-[x']\big),
\label{K_x_x'}
\end{eqnarray}
and for $n\geq 1$, the $n$-pair correlator is
\be
K_n(x_1,\dots,x_n;x'_1,\dots,x'_n)
:= 
\det\left(\frac{1}{x_i-x'_j}\right)_{1\leq i,j\leq n}
\times
\tau \left(\sum_i [x_i]-[x'_i]\right).
\label{K_n_tau_def}
\ee
\end{definition}
\begin{remark}
The justification for this terminology is that, when $\tau({\bf t})$ is a KP $\tau$-function,
 these quantities are all expressible as fermionic VEV's involving products of 
$n$ pairs of creation and annihilation operators, as detailed in \cite{ACEH2}.
\noindent We note that in the particular case where $\tau({\bf t}) = \tau^{(G, \beta, \gamma)}({\bf t}, \beta^{-1}{\bf s})$, 
$\Psi^\pm_0(x) \in  \Kb[x,{\bf s},\beta,\beta^{-1}][[\gamma]]$, while for any $n\geq 1$, 
 $K_n (x_1,\dots,x_n;x'_1,\dots,x'_n) \in \Kb(x_1,\dots,x_n,x'_1,\dots,x'_n)[{\bf s},\beta,\beta^{-1}][[\gamma]]$.
 \end{remark}

The following result is standard, and follows from the Cauchy-Binet identity (or, equivalently, the fermionic Wick theorem).
A proof, valid for all KP $\tau$-functions, is given in the appendix of  \cite{ACEH2}.

\begin{proposition}
\label{prop:Kn_determinant}
 We have
\be
K_n(x_1,\dots,x_n;x'_1,\dots,x'_n) = \det \left(K(x_i,x'_j)\right).
\label{Kn_det}
\ee
\end{proposition}

%%%%%%%%%%%%%%%% Subsection 4.2 Bosonic functions %%%%%%%%%%%%%%%%

\subsection{Bosonic functions}

Define the following derivations:
\begin{definition}
For any parameter $x$
\be
\nabla(x) := \sum_{i=1}^\infty x^{i-1}{\partial \over \partial t_i}, \quad \tilde{\nabla}(x) := \sum_{i=1}^\infty{ x^{i}\over i}{\partial \over \partial t_i} 
\ee
\end{definition}
In terms of these, and any function $\tau({\bf t})$ of the infinite sequence ${\bf t}=(t_1, t_2, \dots)$ of
flow variables (in particular, for the choice $\tau({\bf t})= \tau^{(G, \beta, \gamma)}(\bf t, \beta^{-1}{\bf s})$),  we introduce the following correlators
\bea
W_{n}( x_1,\dots,x_n )  := \left.\left( \Big(\prod_{i=1}^n \nabla(x_i) \Big)
 \tau ({\bf t}) \right) \right\vert_{{\bf t}= {\bf 0}},
 \label{W_G_def}\\
 \tilde{W}_{n}( x_1,\dots,x_n )  := \left.\left( \Big(\prod_{i=1}^n \nabla(x_i) \Big)
\ln\,  \tau({\bf t}) \right) \right\vert_{{\bf t}= {\bf 0}},
\label{W_tilde_G_def} \\
F_{n}( x_1,\dots,x_n )  := \left.\left( \Big(\prod_{i=1}^n \tilde{\nabla}(x_i) \Big)
 \tau ({\bf t}) \right) \right\vert_{{\bf t}= {\bf 0}},
 \label{F_G_def}\\
 \tilde{F}_{n}( x_1,\dots,x_n )  := \left.\left( \Big(\prod_{i=1}^n \tilde{\nabla}(x_i) \Big)
\ln\,  \tau({\bf t}) \right) \right\vert_{{\bf t}= {\bf 0}}.
\label{F_tilde_G_def}
\eea
\begin{remark}
Note that these definitions apply equally for any function $\tau({\bf t})$ admitting a formal, or analytic, power series expansion
in the flow parameters ${\bf t}=(t_1, t_2, \dots)$, since they only refer to the dependence on these parameters,
whereas the parameters ${\bf s} = (s_1, s_2, \dots)$ are just present as ``spectator'' parameters within the definition.
\end{remark}

It follows that
\bea
\label{eq:defWGn}
W_{n}(x_1,\dots,x_n) &\&= \frac{\partial}{\partial x_1}\dots \frac{\partial }{\partial x_n} F_{n}(x_1,\dots,x_n),\\
\label{eq:defWGn_tilde}
\tilde W_{n}(x_1,\dots,x_n) &\&=\frac{\partial}{\partial x_1}\dots \frac{\partial}{ \partial x_n}\tilde F_{n}(x_1,\dots,x_n).
\eea
As shown in \cite{ACEH2},  if $\tau({\bf t})$ is a KP $\tau$-function,  $W_n(x_1, \dots, x_n)$ has a fermionic representation  as a multicurrent correlator. 

Assuming $\tau_g({\bf t})$ to be normalized such that $\tau_g({\bf 0}) = 1$, we have, in particular,
 \bea
 \tilde{W}_1(x_1) &\& = W_1(x_1 ) \cr
  \tilde{W}_2(x_1, x_2) &\& = W_2(x_1, x_2)  - W_1(x_1) W_2(x_2) \cr
   \tilde{W}_3(x_1, x_2, x_3) &\& = W_3(x_1, x_2, x_3 )  - W_1(x_1) W_2(x_2, x_3) - W_1(x_2) W_2(x_1, x_3)
   - W_1(x_3) W_2(x_1, x_2) \cr
   &\&{\hskip 10 pt} + 2 W_1(x_1) W_1(x_2) W_1(x_3)\cr
  \vdots  {\hskip 15 pt}&\& =    {\hskip 20 pt} \vdots 
  \label{tilde_W123}
 \eea
 For general $n$, the moment/cumulant relations between connected and nonconnected functions give
\bea
\label{eq:cumulant2}
W_n(x_1,\dots,x_n)&=&
\sum_{\ell\geq 1} \sum_{I_1\uplus\dots\uplus I_\ell = \{1,\dots,n\}}
\prod_{i=1}^\ell \tilde{W}_{|I_i|}(x_j,j\in I_i),
\eea
 with identical relations holding between the $\tilde{F}_n$'s and $F_n$'s.

For the particular case of $\tau$-functions $\tau^{(G, \beta, \gamma)}({\bf t}, \beta^{-1}{\bf s})$ with expansion (\ref{tau_G_H}), 
we write:
\bea
W_n(x_1, \dots, x_n)&\&= W_n({\bf s}; \beta; \gamma; x_1, \dots, x_n) , \quad \tilde{W}_n(x_1, \dots, x_n)=\tilde{W}_n({\bf s}; \beta; \gamma; x_1, \dots, x_n),  \cr
&\& \\
F_n(x_1, \dots, x_n)&\&= \ \ F_n({\bf s}; \beta; \gamma; x_1, \dots, x_n) , \quad \tilde{F}_n(x_1, \dots, x_n)= \tilde{F}_n({\bf s}; \beta; \gamma; x_1, \dots, x_n) \cr
&\&
\eea
It follows that the $F_n$'s and $\tilde{F}_n$'s may also 
be viewed as generating functions for the weighted Hurwitz numbers $H^d_G(\mu, \nu)$,
encoding the same information as $\tau^{(G, \beta, \gamma)} ({\bf t},\beta^{-1}{\bf s})$, but in a different way. 

\begin{proposition}
\label{F_n_gen_fn}
\bea
\label{eq:defFn}
F_n({\bf s};\beta;\gamma;x_1,\dots,x_n) &\&=  \sum_{d=0}^\infty \sum_{\mu,\nu , |\mu|=|\nu| \atop \ell(\mu)=n}  \gamma^{|\mu|} \beta^{d-\ell(\nu)} H_{G}^d(\mu,\nu)\, |\aut(\mu)| m_\mu(x_1,\dots,x_n) p_\nu({\bf s}), \cr 
&\& 
\label{eq:defFn_tilde}
\\ 
\tilde F_n({\bf s};\beta;\gamma;x_1,\dots,x_n) &\&:= \sum_{d=0}^\infty \sum_{\mu,\nu, |\mu|=|\nu| \atop \ell(\mu)=n} \gamma^{|\mu|} \beta^{d-\ell(\nu)} \tilde H^d_{G}(\mu,\nu)\,|\aut(\mu)|m_\mu(x_1,\dots,x_n) p_\nu({\bf s}) \cr
&\& \\
&\&=  \sum_{g=0}^\infty  \beta^{2g-2+n} \tilde{F}_{g,n}({\bf s}, \gamma; x_1, \dots, x_n), \cr
&\&
\label{eq:defFn_tilde_exp}
\eea
where
\be
\label{eq:defFgn_tilde}
\tilde F_{g,n}({\bf s};\gamma;x_1,\dots,x_n) =  \sum_{\mu,\nu, |\mu|=|\nu|  \atop \ell(\mu)=n}  \gamma^{|\mu|}  \tilde H^{2g-2+n+\ell(\nu)}_{G}(\mu,\nu)\, |\aut(\mu)| m_\mu(x_1,\dots,x_n) p_\nu({\bf s}) 
\ee
and $m_\mu(x_1,\dots, x_n)$ is the monomial symmetric polynomial in the indeterminates $(x_1,\dots , x_n)$.
\end{proposition}
\begin{proof}
Applying the operator   $\prod_{i=1}^n \tilde{\nabla}(x_i) $ to the symmetric function $p_\mu({\bf t})$ in the expansion (\ref{tau_G_H})
gives
\bea
\left.\left(\prod_{i=1}^n \tilde{\nabla}(x_i)\right) p_\mu({\bf t}) \right\vert_{{\bf t}= {\bf 0}}&\& 
=\left. \prod_{i=1}^n \left(\sum_{j_i =1}^\infty {1\over j_i}x_i^{j_i } {\partial \over \partial t_{j_i}}\right)
\left(\prod_{k=1}^{\ell(\mu)} \mu_k t_{\mu_k}\right)\right\vert_{{\bf t}= {\bf 0}} \cr
&\& =\delta_{\ell(\mu), n} \sum_{\sigma \in \mathfrak{S}_n} \prod_{k=1}^n  x_{\sigma(k)}^{\mu_k } 
 =\delta_{\ell(\mu), n}|\aut(\mu)|  m_\mu (x_1, \dots, x_n) 
\eea
by (\ref{m_lambda}).
Hence, applying it to $\tau^{(G, \beta, \gamma)} ({\bf t},\beta^{-1} {\bf s})$, using the expansion (\ref{tau_G_H}), we obtain
from the definition (\ref{F_G_def}) of $F_n({\bf s}, x_1, \dots, x_n)$
\bea
F_n({\bf s};\beta;\gamma; x_1, \dots, x_n) &\&= \sum_{\mu,\nu,\,\ell(\mu)=n}  \sum_d \gamma^{|\mu|} \beta^{d-\ell(\nu)}  H^d_{G}(\mu,\nu)
|\aut(\mu)|   m_\mu (x_1, \dots, x_n)   p_\nu({\bf s}), \cr
&\&
\eea
proving (\ref{eq:defFn}).
 The same calculation proves the connected case   (\ref{eq:defFgn_tilde}).
\end{proof}

Note that $F_n({\bf s};\beta;\gamma; x_1, \dots, x_n)$ and $\tilde F_n({\bf s};\beta;\gamma; x_1, \dots, x_n)$ belong to 
 $\Kb[x_1,\dots,x_n; {\bf s}; \beta, \beta^{-1}][[\gamma]]$ while the $\tilde F_{g,n}({\bf s};\gamma; x_1, \dots, x_n)$'s are in
 $\Kb[x_1,\dots,x_n; {\bf s}][[\gamma]]$.
We also define 
\be
\label{eq:defWGgn_tilde}
\tilde W_{g,n}({\bf s};\gamma;x_1,\dots,x_n) := \frac{\partial}{\partial x_1}\dots \frac{\partial }{ \partial x_n}\tilde F_{g,n}({\bf s};\gamma;x_1,\dots,x_n)
\ee
and therefore, from eq.~(\ref{eq:defFn_tilde_exp}), we have the expansion
\be
\tilde W_{n}({\bf s};\beta;\gamma;x_1,\dots,x_n)  =  \sum_{g=0}^\infty  \beta^{2g-2+n} \tilde{W}_{g,n}({\bf s}, \gamma; x_1, \dots, x_n).
\ee

In the following two subsections, we give explicit relations between the  functions defined above. 
Note that these relations hold in complete generality, valid for  pair correlators and current correlators defined
by the formulae (\ref{K_n_tau_def}), (\ref{W_G_def}) - (\ref{F_tilde_G_def}) for arbitrary $\tau$-functions.
Propositions \ref{prop:BosonsToFermions}  and \ref{prop:Wndeterminant} are in fact tautological; they do not even require 
 that $\tau({\bf t})$ be a KP $\tau$-function.

%%%%%%%%%%%%%%%% Subsection 4.3. From Bosons to Fermions %%%%%%%%%%%%%%%%

\subsection{From bosons to fermions}

The following proposition allows us to write the fermionic functions in terms of the bosonic ones.
We emphasize that it does not even require $\tau({\bf t})$  to be a KP $\tau$-function, it applies to any infinite power series
in the flow variables ${\bf t}$,  whether formal or analytic.
\begin{proposition}
\label{prop:BosonsToFermions}
For all $\tau({\bf t})$ admitting an analytic or formal power series expansion in the parameters ${\bf t}=(t_1, t_2, \dots)$,
with $\Psi^\pm_0(x)$ defined by (eq.~\ref{Psi_KP_Tau_t0}) and $K_n(x_1,\dots,x_n;x'_1,\dots,x'_n)$ by (\ref{K_n_tau_def}),
we have
\begin{eqnarray}
\Psi^\pm_0(x) 
	&=& \sum_n \frac{(\pm 1)^n}{n!}  F_n(x,\dots,x)  \cr
	&=&  e^{\sum_n \frac{(\pm 1)^n}{n!} \tilde F_n(x,\dots,x) } ,
	  \label{eq:PsiRefined}
\label{Psi_fermi_bose}
\end{eqnarray}
\begin{eqnarray}
K(x,x') 
&=& \frac{1}{x-x'}\, \sum_{n,m} \frac{(-1)^m}{n!m!}  F_{n+m}(\overbrace{x,\dots,x}^n,\overbrace{x',\dots,x'}^m)  \cr
&=& \frac{1}{x-x'}\, e^{\sum_{n,m} \frac{(-1)^m}{n!m!} \tilde  F_{n+m}(\overbrace{x,\dots,x}^n,\overbrace{x',\dots,x'}^m) }. \cr
&\&
\label{K_exp_F_exp}
\end{eqnarray}
More generally, we have
\begin{eqnarray}
&& K_k(x_1,\dots,x_k;x'_1,\dots,x'_k) 
/ \det\left(\frac{1}{x_i-x'_j}\right)\cr
&=&
\sum_{n_1,\dots,n_k,\atop m_1,\dots,m_k} \frac{(-1)^{\sum_i m_i}}{\prod_i n_i! \prod_i m_i!} 
 F_{\sum_i n_i+m_i}(\dots,\overbrace{x_i,\dots,x_i}^{n_i},\dots,\overbrace{x'_j,\dots,x'_j}^{m_j},\dots)  
 \label{K_n_fermi_bose1}
\\
&=& 
 \exp \left(\quad 
\sum_{n_1,\dots,n_k,\atop m_1,\dots,m_k} \frac{(-1)^{\sum_i m_i}}{\prod_i n_i! \prod_i m_i!} \tilde F_{\sum_i n_i+m_i}(\dots,\overbrace{x_i,\dots,x_i}^{n_i},\dots,\overbrace{x'_j,\dots,x'_j}^{m_j},\dots) \right).
\label{K_n_fermi_bose2}
\end{eqnarray}
\end{proposition}
\begin{remark}
For the case $\tau({\bf t}) = \tau^{(G, \beta, \gamma)}({\bf t}, \beta^{-1}{\bf s})$ with expansion (\ref{tau_G_H}),  these equalities should 
be interpreted in $\Kb(x,x';x_i,x'_i,1\leq i\leq n)[{\bf s},\beta,\beta^{-1}][[\gamma]]$. Substituting the genus
expansion (\ref{eq:defFn_tilde_exp}) for $\tilde F_{n}({\bf s};\gamma;x_1,\dots,x_n) $  into the exponential sums
 in  (\ref{eq:PsiRefined}), (\ref{K_exp_F_exp}) and (\ref{K_n_fermi_bose2}) gives the genus expansions of the latter.
\end{remark}
\begin{proof}
All the results stated follow from the single tautological identity
\be
e^{\sum_{i=1}^k \tilde{\nabla}(x_i) - \sum_{j=1}^l \tilde{\nabla}(y_j)}\tau({\bf t})|_{{\bf t}={\bf 0}}
= \tau(\sum_{i=1}^k[x_i] - \sum_{j=1}^l [y_j]),
\label{exp_tau_eval}
\ee
valid for any infinitely differentiable function $\tau({\bf t})$ of an infinite sequence of variables ${\bf t}=(t_1, t_2, \dots)$,
or formal power series in these variables, and any set of $k+l$ indeterminates $(x_1, \dots, x_k, y_1, \dots, y_l)$,
by expanding the exponential series for each term in the exponent sum of (commuting)  operators.
 For the case  $\Psi^\pm_0(x)$, to obtain (\ref{Psi_fermi_bose}) we recall eq.~(\ref{Psi_KP_Tau_t0}), 
 and choose either $k=1$, $l=0$, $x_1 =x$ or $k=0$, $l=1$, $y_1 = x$ in (\ref{exp_tau_eval}). 
 The case (\ref{K_n_fermi_bose2}) is obtained by doing the
same for $k=l $, and setting the $y_i=x'_i$, $i=1, \dots , k$.  The connected versions follow from
applying the same exponential of commuting operators to $\ln(\tau({\bf t}))$.
\end{proof}

%%%%%%%%%%%%%%%% Subsection 4.4 From Fermions to Bosons %%%%%%%%%%%%%%%%

\subsection{From fermions to bosons}\label{subsec:fermionsToBosons}

Conversely to the results of the previous subsections, we can also express the bosonic functions in terms of the  fermionic ones. 
The next result is again valid for any function  $\tau({\bf t})$ admitting either a convergent or formal power series
expansion in the flow variables ${\bf t}=(t_1, t_2, \dots)$.
\begin{proposition}
\label{prop:Wndeterminant}
\be
 \label{eq:Wndeterminant}
W_n(x_1,\dots,x_n) = 
[\epsilon_1 \dots \epsilon_n]\left( 
 K_n(x_1,\dots,x_n; x_1-\epsilon_1,\dots,x_n-\epsilon_n)
\big/\det\left(\tfrac{1}{x_i-x_j+\epsilon_j}\right) \right).
\ee
\end{proposition}
\begin{proof}
From the definition (\ref{K_n_tau_def}), this  is equivalent to the relation
\be
{\partial^n\  \tau(\sum_{i=1}^n\left( [x_i] - [x_i - \epsilon_i]\right))  \over \partial \epsilon_1 \cdots \partial \epsilon_n}
\bigg\vert_{\{\epsilon_i = 0\}_{i=1, \dots, n}}
= \prod_{i=1}^n \nabla(x_i) \tau({\bf t})|_{{\bf t} = {\bf 0}},
\ee
which follows from
\be
{1\over j}((x_i)^j - (x_i-\epsilon_i)^j) = \epsilon_i x_i^{j-1} + \OO(\epsilon_i^2)
\ee
by applying the chain rule.
\end{proof}

If $\tau({\bf t})$ is a KP $\tau$-function, by Proposition~\ref{prop:Kn_determinant}, the $n$-pair correlator $K_n(x_1, \dots, x_n; x_1', \dots, x_n')$  
is expressible as an $n\times n$ determinant in terms of the single pair correlators $K(x_i,x'_j)$.
This enables us to express the functions $W_n$ in terms of these quantities alone. For the connected functions, $\tilde{W}_n(x_1, \dots, x_n)$,
we get the following elegant relations.
\begin{proposition}
\label{prop:detConnected}
\be\label{eq:detConnected1}
\tilde W_1(x) = \lim_{x'\to x} \left( K(x,x') - \frac{1}{x-x'}\right),
\ee
\be\label{eq:detConnected2}
\tilde W_2(x_1,x_2) = \left( -K(x_1,x_2)K(x_2,x_1)  - \frac{1}{(x_1-x_2)^2}\right),
\ee
and for $n\ge 3$
\be\label{eq:detConnectedn}
\tilde W_n(x_1,\dots,x_n) =
\sum_{\sigma \in \mathfrak{S}_n^{\rm 1-cycle}}
\sgn(\sigma) \prod_{i} K(x_i,x_{\sigma(i)}),
\ee
where the last sum is over all permutations in $\mathfrak{S}_n$  consisting of a single $n$-cycle.
\end{proposition}
\begin{proof}
The proof is given in Appendix \ref{app_A1}. It is based on a computation of the exact expression for the unconnected correlator 
$W_n(x_1, \dots, x_n)$, using the determinantal identity (\ref{Kn_det}), followed by application of the cumulant identity
(\ref{eq:cumulant2}) relating this to the connected correlators $\{\tilde{W}_k(x_1, \dots, x_n)\}_{k=1, \dots , n}$
\end{proof}
\begin{remark}
Such relations were first noted in the context of complex matrix models~\cite{bergere1}.
If we view eqs.~(\ref{Psi_KP_Tau_t0}), (\ref{K_n_tau_def}) , (\ref{W_G_def}),  (\ref{F_G_def}) (\ref{F_tilde_G_def}),  
as  {\em definitions} of the functions $\Psi_0^{\pm}(x)$,   $K_n(x_1, \dots, x_n, x'_1, \dots, x'_n)$, $W_n(x_1,\dots,x_n)$, 
 $F_n(x_1,\dots,x_n)$ and $\tilde{F}_n(x_1,\dots,x_n)$, 
respectively, Propositions \ref{prop:BosonsToFermions} and \ref{prop:Wndeterminant} are valid for 
completely arbitrary functions $\tau({\bf t})$ admitting either a formal
 or convergent power series expansion in the variables ${\bf t}=(t_1, t_2, \dots)$,  while Proposition \ref{prop:detConnected} 
 remains valid if $\tilde{W}_n(x_1, \dots, x_n)$ is defined by  (\ref{W_tilde_G_def}) in terms of  an arbitrary KP $\tau$-function $\tau({\bf t})$.
\end{remark}
\begin{remark}
For the case $\tau({\bf t}) = \tau^{(G, \beta, \gamma)}({\bf t}, \beta^{-1}{\bf s})$ with expansion (\ref{tau_G_H}),  these equalities should 
be interpreted as being in $\Kb(x,x';x_i,x'_i,1\leq i\leq n)[{\bf s},\beta,\beta^{-1}][[\gamma]]$.
\end{remark}

%%%%%%%%%%%%%%%% Section 4.5  Expansions of $\Psi^+_0(x)$, $\Psi^-_0(x)$, $K(x,x')$ %%%%%%%%%%%%%%%

\subsection{Expansion of $K(x,x')$ for the case $\tau^{(G, \beta, \gamma)}({\bf t}, \beta^{-1}{\bf s})$}

Some properties of Schur functions corresponding to hook
partitions will be needed in the following (see \cite{Mac}):
\begin{lemma}
\label{lemma:hookSchur}
The Schur function $s_{(a|b)}({\bf t})$ corresponding to a hook partition $(1+a, 1^b)$, (denoted
$(a | b)$ in Frobenius notation)  may be expressed in terms of the complete symmetric
functions as
\be
s_{(a|b)}({\bf t}) =  \sum_{j=1}^{b+1}h_{a+j}({\bf t}) h_{b-j+1}({\bf t}).
\ee
\end{lemma}
Substituting the specific evaluations at ${\bf t} = [x] - [x']$, we obtain
\begin{lemma}
\label{lemma:evalSchur}
The Schur function $s_\lambda([x]-[x'])$ vanishes unless $\lambda$ is a hook or the trivial partition
and in that case it takes the value
\be
s_{(a|b)}([x]-[x']) =  (x-x') x^a (-x')^b.
\ee
\end{lemma}

From \eqref{K_x_x'}, \eqref{tau_schur_exp} and Lemma~\ref{lemma:evalSchur} we obtain (see \cite{ACEH1} for details of the proof):
\begin{proposition}
\label{prop:explicitFunctions}
For the particular case where $\tau({\bf t}) = \tau^{(G, \beta, \gamma)}({\bf t}, \beta^{-1}{\bf s})$, 
 the correlation function $K(x,x')$  can be expressed as:
\bea
K(x,x') &\&= {1 \over x - x'} +\sum_{a,b \geq 0}  \rho_a \rho^{-1}_{-b-1} x^a (-x')^b s_{(a|b)}(\beta^{-1}{\bf s}) \cr
&\& ={1 \over x - x'} +  \sum_{a=0}^\infty \sum_{b=0}^\infty \sum_{j=1}^{b+1} \
\rho_a h_{a+j} (\beta^{-1}{\bf s}) x^a \rho^{-1}_{-b-1} h_{b-j+1}(-\beta^{-1} {\bf s}) (x')^b.
\label{Khook} 
\eea
\end{proposition}

%%%%%%%%%%%%%%%% Section 5. Adapted bases, Recursion opertaors %%%%%%%%%%%%%%%%

\section{Adapted bases, recursion operators and the Christoffel-Darboux relation}
\label{sec:recursionOperators}

%%%%%%%%%%%%%%%% Section 5.1. Recursion operators %%%%%%%%%%%%%%%%

\subsection{Recursion operators and adapted basis}
\label{rec_ops_adapted_basis}

In this section we describe, for the case of the $\tau$-function $\tau^{(G, \beta, \gamma)}({\bf t}, \beta^{-1}{\bf s})$,
 an alternative approach to the adapted bases $\{\Psi^+_i(x)\}_{i\in \Zb}$ and their duals $\{\Psi^-_i(x)\}_{i\in \Zb}$.
 All subsequent developments concern only  this specific case: the hypergeometric $\tau$-function 
 $\tau^{(G, \beta, \gamma)}({\bf t}, \beta^{-1}{\bf s})$ defined in (\ref{tau_schur_exp})  which serves, 
 by  (\ref{tau_G_H}), as generating function for weighted Hurwirz numbers.

\begin{definition}
Denote the Euler operator as
\be
D:= x\frac{d}{dx}
\ee
and define the pair of dual recursion operators 
\be
\label{Rdef}
R_{\pm}:=\gamma x G(\pm\beta D).
\ee
The same operators in the variable $x'$ are denoted $D'$ and $R_{\pm}'$
\end{definition}
Since $D$ commutes with $G(\pm\beta D)$, it follows that these operators satisfy the commutation relations
\bea\label{Rcom}
\left[D,R_{\pm}\right]=R_{\pm}.
\eea

\begin{definition}
Denote the (normalized) power sums of the variables  $\{c_i\}_{i\in \Nb^+}$:
\be
A_k:=\frac{1}{k} \sum_{i=1}^\infty c_i^k,\,\,\,\,\, k>0.
\label{A_k_def}
\ee
\end{definition}

Then the following lemma is easily proved (see Appendix \ref{app_A2}).
\begin{lemma} \label{lemmarec}
There exists a unique formal power series $T(x)$ such that
\be\label{TGrel}
e^{T(x)-T(x-1)}=\gamma G(\beta x)
\ee
with  $T(x)-x\log\gamma \in{\Kb}[x][[\beta]]$ and $T(0)=0$.
Explicitly,
\bea
\label{Tdef}
T(x) =x\log\gamma  + \sum_{k=1}^\infty (-1)^k A_k \beta^k \frac{B_{k+1}(x)-B_{k+1}(0)}{k+1},
\eea
where $\{B_{k}(x)\}$ are the Bernoulli polynomials. 
\end{lemma}

\begin{corollary}
The operators $R_\pm$ can be expressed as
\bea\label{Tcom}
R_+=e^{T\left(D-1\right)}\circ x \circ e^{-T\left(D-1\right)},\\
R_-=e^{-T\left(-D\right)} \circ x \circ e^{T\left(-D\right)}.
\label{Tcom1}
\eea
\end{corollary}

We can now express the Baker function $\Psi^-_0(x)$ and its dual $\Psi^+_0(x)$, as well as $K(x, x')$, in terms of 
the series $T$. 
Defining
\be
\xi(x,{\bf s)}:=\sum_{k=1}^\infty s_k x^k,\ee
 we have
\begin{proposition}\label{Proppsiphi}
\bea\label{psiprod}
\Psi^+_0(x)&=&\gamma e^{T\left(D-1\right)}\left(e^{\beta^{-1}\xi(x,s)}\right),\\
\Psi^-_0(x)&=&e^{-T\left(-D\right)}\left(e^{-\beta^{-1}\xi(x,s)}\right), \\
K(x,x')&=&e^{T\left(D\right)-T\left(-D'-1\right)}
\left({e^{\beta^{-1}\xi({\bf s},x) - \beta^{-1}\xi ({\bf s},x')}\over x-x'}\right).
\label{K_x_x_prime}
\eea
\end{proposition}
\begin{proof}From the Cauchy-Littlewood identity and Lemma~\ref{lemma:evalSchur} we have
\be
\sum_{a,b\geq 0} x^a(-x')^b s_{(a|b)}(\beta^{-1}{\bf s}) = \frac{e^{\beta^{-1}\xi({\bf s},x) - \beta^{-1}\xi ({\bf s},x')}}{x-x'}.
\ee
Expression (\ref{K_x_x_prime}) for $K(x,x')$ then follows from the definition of the operators $T(D), T(D-1)$
and Proposition~\ref{prop:explicitFunctions}.
The other two equalities are proved similarly.
\end{proof}

\begin{definition}
Using the operators $R_{\pm}$ defined in (\ref{Rdef}), we define
\bea
\Psi^\pm_k(x):=R_\pm^k \Psi^\pm_0(x),  \quad k \in \Zb.
\label{Psi+_k_rec}
\eea
\end{definition}
It then follows from (\ref{Tcom}) that
 \begin{proposition}\label{Proppsiphik}
\bea
\Psi^+_k(x)=\gamma e^{T\left(D-1\right)}\left(x^k\,e^{\beta^{-1}\xi(x,s)}\right) ,
\label{Psikpr} \\
\Psi^-_k(x)=e^{-T\left(-D\right)}\left(x^k\,e^{-\beta^{-1}\xi(x,s)}\right).
\label{Psi-kpr} 
\eea
\end{proposition}
\noindent and we have the following identifications:
\begin{proposition}
\label{adapted_basis_series}
The functions $\{\Psi^\pm_k(x)\}_{k\in \Zb}$ defined in (\ref{Psi+_k_rec}) coincide with those defined in (\ref{Psi+_k}) and (\ref{Psi-_k}). Thus  we have the series expansions
\bea
\Psi^+_k(x) =\gamma  \sum_{j=0}^\infty   x^{j+k} h_j(\beta^{-1}{\bf s}) \rho_{j+k-1}, 
\label{Psi_k_plus}\\
\Psi^-_k(x) = \sum_{j=0}^\infty x^{j+k} h_j(-\beta^{-1}{\bf s}) \rho_{-j-k}^{-1}.
\label{Psi_k_minus}
\eea
\end{proposition}
\begin{proof}
From the definition \ref{rho_j_gamma_G} of the coefficients $\rho_j$ it follows that
\bea
 e^{T\left(D-1\right)} x^{k} =\rho_{k-1} x^{k}\\
 e^{-T\left(-D\right)} x^{k}=\rho_{-k}^{-1} x^{k}.
\eea
The statement then follows from the series expansion of the expressions in the parentheses in (\ref{Psikpr})-(\ref{Psi-kpr}).
\end{proof}

\begin{remark}
The operator representation (\ref{Psikpr}), (\ref{Psi-kpr}) coincides with the generalized convolution action 
on the standard orthonormal (monomial basis) elements of $\HH$, as explained in Section \ref{convolution_action}, 
and given in eqs.~(\ref{Psi+_k}), (\ref{Psi-_k}). Using different methods, the series expansions 
(\ref{Psi_k_plus}), (\ref{Psi_k_minus}) were  also  derived in the companion papers \cite{ACEH1, ACEH2}. 
\end{remark}

From (\ref{Psikpr}), (\ref{Psi-kpr}) we deduce
\begin{proposition}
For all $m>0$ and $k \in  \Zb$, we have
\bea\label{sderiv}
\Psi^\pm_{k+m}(x)=\pm\beta \frac{d}{d s_m}\Psi^\pm_{k}(x).
\eea
\end{proposition}
\begin{proof}
This  follows from the series expansions (\ref{Psi_k_plus}) (\ref{Psi_k_minus}) using the property
\be
{\partial h_j  ({\bf s})\over \partial s_i} =h_{j-i}({\bf s}).
\ee
\end{proof}

%%%%%%%%%%%%%%%% Section 5.2. The Christoffel-Darboux relation  %%%%%%%%%%%%%%%%

\subsection{The Christoffel-Darboux relation}\label{sec:CDR}
\begin{definition}
Define the quantity $S(z)$ by
\be\label{Sdefin}
S(z):=z \frac{d}{d z}\xi(z,{\bf s})=\sum_{k=1}^\infty k s_k z^k.
\ee
\end{definition}
In the following, we assume that only a finite number of variables $\{s_k\}$ are nonzero; i.e.,  that $S(z)$ is a polynomial in $z$
of degree
\be
L := \deg S.
\ee
We also assume that the model is not degenerate; i.e., $ML> 1$.

\begin{definition}
Define the operators
\be
\Delta_\pm(x):=\pm\beta e^{\mp\beta^{-1}\xi(x,s)} \circ D \circ e^{\pm\beta^{-1}\xi(x,s)}=S(x)\pm \beta D
\label{eq:defDeltapm}
\ee
and
\be
V_\pm(x):=\gamma^{-1}x^{-1}e^{\mp\beta^{-1}\xi(x,s)}\circ R_{\pm}\circ e^{\pm\beta^{-1}\xi(x,s)}=G(\Delta_\pm(x)).
\label{eq:defVpm}
\ee
\end{definition}

\begin{definition}\label{Adefinition}
The polynomial $A(r,t)$ of degree $L M-1$ in each variable $(r,t)$ and the
	$LM \times LM$ matrix $\mathbf{A}=(A_{ij})_{0\leq i,j \leq LM-1}$ are defined by
\be
	A(r,t):=\left(r\, V_-(t)-t\, V_+(r)\right)\left( \frac{1} {r-t}\right) = \sum_{i=0}^{LM-1} \sum_{j=0}^{LM-1} A_{ij}r^i t^j .
\label{eq:defA}
\ee
\end{definition}

The highest total degree term is
\be
g_M \, (L s_L)^M \frac{r\, t^{LM}-t\, r^{LM}}{r-t}=-g_M \, (L s_L)^M \sum_{j=1}^{LM-1} r^j t^{LM-j}, 
\ee
so 
\be
\det \mathbf{A} =(-1)^\frac{LM(LM-1)}{2}\, g_M^{LM-1} \, (L s_L)^{M(LM-1)}.
\ee
Therefore the matrix $\mathbf{A}$ is nonsingular.

The following two results are proved in~\cite{ACEH2}. For the convenience of the reader, a short proof of the first result using only 
the notation of the present paper is also given in Appendix \ref{app_A2}.
\begin{theorem}[{\cite{ACEH2}}]
\label{thm:CD} The following ``Christoffel-Darboux'' relation holds:
\be
K(x,x')=\frac{1}{x-x'} A(R_+,R_-')\Psi^+_0(x)\Psi^-_0(x') =\frac{1}{x-x'}\sum_{i,j=0}^{LM-1}{A}_{ij}  \Psi^+_{i}(x)\Psi^-_{j}(x').
\label{eq:CD}
\ee
\end{theorem}
\begin{proposition}[\cite{ACEH2}]
The Christoffel-Darboux matrix elements $A_{ij}$
entering in \eqref{eq:defA} have the explicit expression
\bea
A_{ij} &\&=  - \sum_{k=-i}^{j} G(k \beta) h_{j-k} (-{\beta^{-1}\bf s}) h_{i+k}(\beta^{-1}{\bf s}), 
\quad i,j = 1,2, \dots, \label{Acoefficients}
\\
A_{00} &\&=1, \quad A_{0j} = A_{i 0} = 0.
\nonumber
\eea
\end{proposition}
For example, if $L=M =2$, the relation~\eqref{eq:CD} takes the form
\bea
K(x,x')=\frac{1}{x-x'}\sum_{i,j=0}^{3}{A}_{ij}  \Psi^+_{i}(x)\Psi^-_{j}(x'),
\eea
where
\be
\bf{A} = \begin{pmatrix}
 1& 0& 0 & 0  \cr
0 & -2\,s_{{2}}g_{{1}}-{s_{{1}}}^{2}g_{{2}}  &-4\,s_{{1}}s_{{2}} g_{{2}} & -4\,{s_{{2}}}^{2}g_{{2}}  \cr
0 & -4\,s_{{1}}s_{{2}} g_{{2}}  & -4\,{s_{{2}}}^{2}g_{{2}} & 0 & \cr
0 & -4\,{s_{{2}}}^{2}g_{{2}} &0  & 0   \cr
\end{pmatrix}.
\ee
Another interesting example is when $s_{k}=\delta_{k,1}$.
Then, from (\ref{tauR}) we have
\bea\label{sdelta}
\tau\left([x]-[x'],\beta^{-1}{\bf s}\right)
&=&\gamma e^{T\left(D-1\right)-T\left(-D\right)}\left(G\left(-D'\right)x-G\left(\beta D\right)x'\right) \frac{e^{\beta^{-1}(x-x')}}{x-x'} \\
&=&\gamma e^{T\left(D-1\right)-T\left(-D'\right)}
\frac{G\left(-\beta D'\right)D+G\left(\beta D\right)D'}{D+D'}
e^{\beta^{-1}(x-x')}\\
&=&\frac{G\left(-\beta D'\right)D+G\left(\beta D\right)D'}{D+D'}\Psi^+_0(x) \Psi^-_0(x').
\eea
Thus the $\tau$-function can be represented in terms of a polynomial in the operators $D$ and $D'$  acting on
$\Psi^+_0(x) \Psi^-_0(x')$.

%%%%%%%%%%%%%%%% Section 6  Quantum spectral curve %%%%%%%%%%%%%%%%

\section{The quantum spectral curve and first order linear differential systems}
\label{section5}

%%%%%%%%%%%%%%%% Section 6.1. The quantum spectral curve  %%%%%%%%%%%%%%%%

\subsection{The quantum spectral curve}

\begin{theorem}
\label{thm:kspectral0}
The functions $\Psi^\pm_0(x)$ satisfy the quantum curve equation
\be
\left(\beta D-S(R_+)\right)\Psi^+_0(x)=0,
\label{eq:quantumCurve}
\ee
and its dual 
\be
\left(\beta D+S(R_-)\right)\Psi^-_0(x)=0.
\label{eq:quantumCurveDual}
\ee
\end{theorem}
\begin{proof}
The statement follows immediately from Lemma \ref{lemmarec} and Proposition \ref{Proppsiphi}, 
\bea
\beta D \Psi^+_0(x)&=&\beta D \gamma e^{T\left(D-1\right)}\left(e^{\beta^{-1}\xi(x,s)}\right)\\
&=&  \gamma e^{T\left(D-1\right)} S (x) e^{\beta^{-1}\xi(x,s)}\\
&=& S(R_+)  \Psi^+_0(x),
\eea
and similarly for $\beta D \Psi^-_0(x)$.
\end{proof}
\begin{remark}
	Eq.~\eqref{eq:quantumCurve} can also be proved directly from the combinatorial viewpoint of
 constellations using an edge-removal decomposition. The proof is similar to the one
  we give for the spectral curve equation (Proposition~\ref{prop:eqY} below).
  \end{remark}
\begin{remark}\label{rem:orbifold} Quantum spectral curves for some particular types of Hurwitz numbers were constructed
in \cite{ALS,Zhou,MSS,Do}.
The case of strictly monotone orbifold Hurwitz numbers corresponding to $G(z)=1+z$ and $S(z)=z^r$,
investigated in \cite{Do, Dunin},  gives the quantum curve equation (\ref{eq:quantumCurve}) of the form
\be
\left(\beta D-\left(\gamma x (1+\beta D)\right)^r\right)\Psi^+_0(x)=0.
\ee
To identify this with the quantum spectral curve, derived in \cite{Do}, let
 \be
 \hbar:=\beta. \quad \hat{x}:=x^{-1}, \quad \hat{y}:=\beta x^2 \frac{\partial}{\partial x}
 \ee
  and put $\gamma=1$. Then the quantum spectral curve operator is 
\be
\hat{x}\hat{y}-\left(\hat{x}^{-1}+\hat{y}\right)^r=x^{\hbar^{-1}}\left(\hat{x}\hat{y}-1-\hat{y}^r\right)x^{-\hbar^{-1}}
\ee
 where the operator in brackets on the RHS  is the negative of the quantum curve operator in Theorem 1 of \cite{Do}.
\end{remark}

More generally, we have the corresponding equations satisfied by the other elements
of the adapted basis $\{\Psi^\pm_i\}_{i\in \Zb}$
\begin{theorem}\label{thm:kspectral}
\bea
\left[\beta D \mp S(R_\pm)\right] \Psi_k^\pm(x) &\&= \beta  k \Psi^\pm_k(x).
\label{Psi+_k_spectral_eq}
\eea
\end{theorem}
\begin{proof}
 As shown in \cite{ACEH2}, this follows either from applying the Euler operator $D$ to the representation of $\Psi_k^\pm$ given in Proposition \ref{Proppsiphik}, or directly to the series expansions (\ref{Psi_k_plus}), (\ref{Psi_k_minus}) of Proposition \ref{adapted_basis_series}.
\end{proof}

%%%%%%%%%%%%%%%% Section 6.2. The infinite differential system  %%%%%%%%%%%%%%%%

\subsection{The infinite differential system}\label{subsec:matricesPQ}
\begin{definition}
Define  two doubly infinite column vectors whose components are the functions $\Psi^+_k(x)$ and $\Psi^-_k(x)$.
\be
\vec{{\Psi}}^+_{\infty}:=\begin{pmatrix}
\vdots\\
{\Psi}^+_{-1}\\
{\Psi}^+_{0}\\
{\Psi}^+_{1}\\
\vdots\\
\end{pmatrix},\,\,\,\,\,\,\,\,\,\,\,\,\,
\vec{{\Psi}}^-_\infty:=\begin{pmatrix}
\vdots\\
{\Psi}^-_{-1}\\
{\Psi}^-_{0}\\
{\Psi}^-_{1}\\
\vdots\\
\end{pmatrix}.
\ee 
and four doubly infinite  matrices $Q^\pm$, $P^\pm$ that are constant  in $x$,  with matrix elements
\bea
P^\pm_{ij} 
&\&= \begin{cases} \pm\beta j \delta_{ij} + (j-i) s_{j-i}, \quad j\ge i +1 \cr
   0, \ \quad j\le i. 
 \end{cases} 
 \label{deriv_rec_matrices_general} 
\eea
\bea
Q^\pm_{ij} &\&:= \begin{cases} \gamma \sum_{k=i-1}^{j} r^{(G, \beta)}_{\pm k} h_{k-i+1}(\pm \beta^{-1}{\bf s})  h_{j-k}(\mp\beta^{-1}{\bf s}) , \quad j\ge i-1
\cr
   0, \ \quad j\le i-2. 
 \end{cases} 
\label{Q_pm_hypergeometric}
\eea
\end{definition}
Note that the matrices  $P^\pm$ are  upper triangular, whereas $Q^\pm$ are
almost upper triangular, with  one nonvanishing diagonal just  below the principal one.

\begin{theorem}\label{thm:infiniteSystem}
The basis elements $\{\Psi^\pm_k\}$ satisfy the following recursion relations under multiplication by $\frac{1}{\gamma x}$
and differential relations upon application of the Euler operator $D:= x {d \over dx}$
\bea
\frac{1}{\gamma x}\vec{{\Psi}}^\pm_\infty= Q^{\pm}\vec{{\Psi}}^\pm_\infty.
\label{eq:Qpm}
\\
\pm\beta D \vec{{\Psi}}^\pm_\infty= P^{\pm}\vec{{\Psi}}^\pm_\infty,
\label{eq:Ppm}
\eea
\end{theorem}
Detailed proofs of this theorem are given in
the companion paper~\cite{ACEH2}. 
We also give an alternative proof in  Appendix \ref{app_A3}.

%%%%%%%%%%% Section 6.3 . Folding: finite-dimensional linear differential system %%%%%%%%%%%%%%%%

\subsection{Folding: finite-dimensional linear differential system}
\label{sec:finiteSystem}

We now consider a finite-dimensional version of the differential system (\ref{eq:Ppm}), obtained by using  (\ref{eq:Qpm}) 
to ``fold'' all the higher and lower components into the $LM$-dimensional window between $0$ and $LM-1$.
Define two column vectors of dimension $LM$ by 
\be
\vec{{\Psi}}^+:=\begin{pmatrix}
{\Psi}^+_{0}\\
{\Psi}^+_{1}\\
{\Psi}^+_{2}\\
\dots\\
{\Psi}^+_{ML-1}\\
\end{pmatrix},\quad \quad 
\vec{{\Psi}}^-:=\begin{pmatrix}
{\Psi}^-_{0}\\
{\Psi}^-_{1}\\
{\Psi}^-_{2}\\
\dots\\
{\Psi}^-_{ML-1}\\
\end{pmatrix}.
\ee 
In terms of these, the Christoffel-Darboux relation can be expressed as
\be\label{CDagain}
K(x,x')=\frac{\vec{{\Psi}}^+(x)^T{\bf A} \vec{{\Psi}}^-(x') }{x-x'},
\ee
where $\mathbf{A}$ is the $LM\times LM$ Christoffel-Darboux matrix introduced in Definition~\ref{Adefinition}.

Now define three further $LM \times LM$ dimensional matrices $\tilde\Eb(x)$, $\Eb^\pm(x)$, all 
first degree polynomials in $1/x$, as follows.
\begin{definition}
The matrix elements $\{\tilde{\Eb}_{ij}\}_{0\leq i,j \leq LM-1}$ of $\tilde{\Eb}$ are defined by 
the generating function
\be
\sum_{i,j=0}^{LM-1}\tilde{\Eb}(x)_{ij}r^{i}t^{j} := \left(\Delta_+(r)rV_-(t)-\Delta_-(t)tV_+(r)\right)\left(\frac{1}{r-t}\right) -\frac{1}{\gamma x} rt \frac{S(r)-S(t)}{r-t},
\label{tildeEE}
\ee
while $\Eb^\pm(x)$  are defined as
\bea
\Eb^-(x) &\& := {\bf A}^{-1} \tilde{\Eb}(x)
\label{def_EE}\\
\Eb^+(x) &\& := ({\bf A}^T)^{-1} \tilde{\Eb}^T
\label{def_EE'}(x).
\eea
\end{definition}

From this, it follows that they satisfy the following duality relation:
\be
\label{maste}
\mathbf{A}{\mathbf{E}}^-(x)-{\mathbf{E}}^+(x)^{T}\mathbf{A}=0.
\ee

From (\ref{tildeEE}) it follows that 
\be\label{Eeqeq}
\sum_{i,j=1}^{LM}{\tilde{\mathbf{E}}}(x)_{ij}r^{i-1}t^{j-1}
=\frac{r S(r)G(S(t))-tS(t)G(S(r))}{r-t}-\frac{1}{\gamma x} rt \frac{S(r)-S(t)}{r-t}+O(\beta),
\ee
where the $\beta$-corrections do not depend on $x$.

We then have the finite dimensional ``folded'' projection of eq.~((\ref{eq:Ppm}).
\begin{theorem}\label{prop:finiteSystem}
The following differential systems hold
\be
\label{finite_D_Psi_eq}
\pm\beta D \vec{{\Psi}}^\pm = {\mathbf{E}}^\pm(x)\vec{{\Psi}}^\pm.
\ee
\end{theorem}
The proof is given in Appendix \ref{app_A3}. It proceeds by ``folding'' the finite band  infinite constant 
coefficient differential system (\ref{eq:Ppm}) into  finite ones with rational coefficients  using the recursion relations (\ref{eq:Qpm}). 
The Christoffel-Darboux relation (\ref{CDagain}) implies that
\be
K(x+\epsilon,x)=\frac{\vec{{\Psi}}^+(x)^T\mathbf{A} \vec{{\Psi}}^-(x) }{\epsilon}
+\left(\frac{d}{d x}\vec{{\Psi}}^+(x)^T\right)\mathbf{A} \vec{{\Psi}}^-(x) +O(\epsilon).
\ee
Therefore, from Proposition~\ref{prop:detConnected},  the correlation function $\tilde{W}_1(x)$ may be expressed as
\be
\tilde{W}_1(x)=\frac{1}{\beta x} \vec{{\Psi}}^+(x)^T\tilde{\Eb}(x) \vec{{\Psi}}^-(x).
\label{W1_Psi_bilinear}
\ee

%%%%%%%%%%%%%%%%%%%%  Subsection 6.4  %%%%%%%%%%%%%%%%%%%%
\subsection{Adjoint differential system}

The results of this section and the next were announced in~\cite{ACEH1}. They are not used in the rest of the present paper.

\begin{definition}
Let $\mat{M}(x)$ be the rank-$1$, $LM \times LM$ matrix  defined by
\be
\label{defM}
	\mat{M}(x) := \vec{\Psi}^-(x) \vec{\Psi}^+(x)^T \, \mat{A},
\ee
with entries 
\be\label{eq:entryM}
	\mat{M}(x)_{ij} = \Psi^-_i(x) \sum_{k=0}^{LM-1} \Psi^+_k(x) A_{kj} 
\ee
viewed as elements of $\Kb[x,s,\beta,\beta^{-1}][[\gamma]]$.
\end{definition}

The matrix $\mat{M}(x)$ has the following properties:
\begin{proposition}\label{prop:Mpositive}
The entries of $\mat{M}(x)$ are elements of $\Kb[x,s,\beta][[\gamma]]$, \textit{i.e.} they contain no negative power of $\beta$.
Moreover,  $\mat{M}(x)$ is a rank-$1$ projector:
\be
\mat{M}(x)^2=\mat{M}(x)
\quad , \quad
\Tr \,\mat{M}(x)=1.
\label{M2_M}
\ee
and satisfies the adjoint differential system
\be\label{eq:adjODE}
\beta x \frac{d}{dx} \mat{M}(x) = [\mat{M}(x), \mat{E}^-(x)].
\ee
\end{proposition}
\begin{proof} The property $\mat{M}(x)\in \Kb[x,s,\beta][[\gamma]]$ follows the fact that
the matrix elements $A_{ij}$ are polynomials in $\beta$ (of degree no greater then $M$), 
as can be seen from eqs.~(\ref{eq:defVpm}),  (\ref{eq:defA}),  and, by Lemma \ref{lemma:PsiPhiStruct} below,
the quantities $ \check\Psi^\pm_i(x)$, defined in eq.~(\ref{eq:defCheckPsi}) below, belong to
$\Kb[x,{\bf s},\beta][[\gamma]]$.

The fact that $\mat{M}(x)$ a rank-$1$ projector follows from its definition (\ref{defM}), and the relation
\be
 \vec{\Psi}^+(x)  \mat{A} \vec{\Psi}^-(x)^T =1,
\ee
which is equivalent to
\be
\lim_{x' \ra x} (x-x')K(x, x') =1.
\ee
The  adjoint equation (\ref{eq:adjODE}) follows from eq.~(\ref{finite_D_Psi_eq}).
\end{proof}

%%%%%%%%%%%%%%%%%%%%  Subsection 6.5  %%%%%%%%%%%%%%%%%%%%
\subsection{The current correlators $\tilde W_n$}

It follows  from the Proposition \ref{prop:detConnected} and the Christoffel-Darboux Theorem \ref{thm:CD} 
 that:
\begin{proposition} \label{prop:WgntoM}
\bea
\tilde W_1(x) 
	&\&= \frac{1}{\beta x} \Tr ( \mat{M}(x) \mat{E}^-(x)), 
	\label{trace_M_1}
	\\
\tilde W_2(x_1,x_2) 
&\&= \frac{\Tr  \mat{M}(x_1)\mat{M}(x_2) }{(x_1-x_2)^2} - \frac{1}{(x_1-x_2)^2}, 
\label{trace_prod_M_2}
\\
\text {for } n\geq 3, \quad 
\tilde W_n(x_1,\dots,x_n) 
&\& = \sum_{\sigma \in \mathfrak S_n^{\rm 1-cycle}}
(-1)^\sigma \frac{\Tr \left( \prod_{i} \mat{M}(x_{\sigma(i)}) \right) }{\prod_i (x_{\sigma(i)} - x_{\sigma(i+1)})}.
\label{trace_prod_M_n}
\eea
\end{proposition}
\begin{proof}
Cyclically reordering the terms in the trace products in (\ref{trace_prod_M_2}),  (\ref{trace_prod_M_n})
and using Theorem \ref{thm:CD}, we obtain
\be
 \frac{\Tr \left( \prod_{i} \mat{M}(x_{\sigma(i)}) \right) }{\prod_i (x_{\sigma(i)} - x_{\sigma(i+1))})}
 = \prod_{i=1}^n K(x_{\sigma(i)}, x_{\sigma(i+1)}).
\ee
Substituting this in (\ref{eq:detConnected2}) and (\ref{eq:detConnectedn}) of  Proposition
\ref{prop:detConnected} gives (\ref{trace_prod_M_2})
and (\ref{trace_prod_M_n}).  By eq.~(\ref{W1_Psi_bilinear}),  eq.~(\ref{trace_M_1}) is equivalent to (\ref{eq:detConnected1}).
\end{proof}

\begin{remark}
In \cite{ACEH1}, a WKB expansion in the parameter $\beta$ was indicated for $\mat{M}(x)$.
In the present work, this is superseded by the WKB expansion developed below in Section \ref{sec:WKB}.

\end{remark}

%%%%%%%%%%%%%%%% Section 7. Classical spectral curve  %%%%%%%%%%%%%%%%

\section{Classical spectral curve and local expansions}
\label{sec:classicalCurve}

%%%%%%%%%%%%%%%% Section 7.1. The classical spectral curve %%%%%%%%%%%%%%%%
\subsection{The classical spectral curve}

The classical spectral curve is the equation satisfied by $\tilde W_{0,1}$. We write:
\be 
y(x):=\tilde W_{0,1}(x).
\ee
The following result is easily proved (for example, by an edge-removal decomposition on constellations: see Appendix \ref{app_A4} ).

\begin{proposition}\label{prop:eqY}
$y(x)$ is the unique solution in 
$\gamma\Kb[x,{\bf s}][[\gamma]]$ of the equation
\be\label{eq:TutteY}
x y(x) = S\left( \gamma x G(x y(x)) \right).
\ee
\end{proposition}

\begin{definition}
We define the \emph{spectral curve} as the complex algebraic plane curve $\{(x,y)\} \ss {\bf C} \times {\bf C}$ given by
\be
x y = S\left( \gamma x G(xy) \right).
\ee
Its compactification (in $\Pb^1 \times \Pb^1$) is a genus $0$ curve (thus $\mathbf {P}^1$) that admits a parametrization by the rational functions $X(z)$ and $Y(z)$ defined as:
\bea\label{zcurve}
	X(z)&:=&\frac{z}{\gamma G(S(z))} ,
		\label{X_z}
		\\
	Y(z)&:=& \frac{S(z)}{z}\,\gamma G(S(z)).
	\label{Y_z}
\eea
\end{definition}
We observe that we have
\be\label{xyrel}
X(z) Y(z)=S(z).
\ee

\begin{corollary}\label{coroll:W01}
The 1-form $\tilde W_{0,1}(x)dx=y(x)dx$ can  be written
\be
\tilde W_{0,1}(x)dx = \frac{S(z)}{z}\,\left(1-\frac{z S'(z) G'(S(z))}{G(S(z))} \right)\, dz.
\ee
\end{corollary}
In particular, this enables us to view this as a meromorphic $1$-form on the spectral curve, not only on a neighbourhood of $x=0$.

%%%%%%%%%%%%%%%%  Subsection 7.2 Some Geometry %%%%%%%%%

\subsection{Some geometry}

%%%%%%%%%%%%%%%%  Subsubsection 7.2.1 Branchpoints and genericity assumption %%%%%%%%%

\subsubsection{Branch points and ramification points}

Since the spectral curve is a plane algebraic curve that admits a rational parametrization, its compactification is a Riemann surface of genus $0$: the Riemann sphere $\bar\Cb=\Cb\cup\{\infty\}=\Pb^1$.
In the following, if $f$ is a function and $k>0$, we write $\operatorname{order}_z f=k$ (resp. $=-k$) if f has a zero of order $k$ (resp. a pole of order $k$) at $z$.

\begin{definition}[Branch points and ramification points]\label{defbp}

A \emph{ramification point} on the spectral curve is a point $z\in\bar\Cb$ at which the map $X: z \ra  X(z)$
 is not invertible in a neighbourhood of $z$.
It is either such that $\operatorname{order}_z (X-X(z)) >1$ in which case $z$ is a zero of $X'$, 
or $\operatorname{order}_z X <-1$, in which case $X$ has a pole at $z$ of degree $\geq 2$. 
	
A \emph{branch point} is the image under $X$ of a ramification point.
	
\end{definition}

\begin{definition}\label{phidef} Define the  function 
\be 
\phi(z) := G(S(z))-zG'(S(z))S'(z).
\ee
If $G$ is a polynomial of degree $M$  and $S$ is a polynomial of degree $L$, then $\phi(z)$ is also a polynomial, of degree $LM$.
\end{definition}

\begin{definition}
Let $\mathcal L$ denote the set of zeros of $\phi$.
\end{definition}

\begin{lemma}
The set of ramification points is $\mathcal L$, together with the point $z=\infty$ (for $LM>2$).
\end{lemma}

\begin{proof}
We have
\be\label{difX}
	X'(z)=\frac{G(S(z))-zG'(S(z))S'(z)}{\gamma G(S(z))^2} = \frac{\phi(z)}{\gamma G(S(z))^2}.
\ee
There are 2 cases: 

-  If $z$ is a zero of $\phi(z)$ and not a zero of $G(S(z))$, then it is not a pole of $X(z)$, but it is a zero of $X'(z)$.
It is therefore a ramification point.

- If $z$ is a zero of $\phi(z)$ and a zero of $G(S(z))$, it cannot be $z=0$, because $G(S(0))=1$, so it must be a zero of $S'(z)G'(S(z))$, i.e. it must be a multiple zero of $G(S(z))$, and thus a pole of $X(z)$ of degree $\geq 2$. Therefore it is a ramification point.

The converse holds for the same reasons, finite ramification points must be zeros of $\phi$.
Since the point $z=\infty$ is a zero of $X$ of degree $LM-1$, it is a ramification point if $LM>2$.
\end{proof}

Self-intersections of the spectral curve are pairs of distinct points $(z_+,z_-)$ such that 
\be
X(z_+)=X(z_-) , \quad Y(z_+)=Y(z_-)
\ee
simultaneously.
Since $z=\gamma\,X(z) \, G(X(z)Y(z))$, this implies $z_+=z_-$, and thus we see that the spectral curve has no self--intersection points, it is a \emph{smooth} genus zero curve.
This means that the map $\bar{\Cb}\to \bar{\Cb}\times \bar{\Cb}$, $z\mapsto (X(z),Y(z))$ is actually an embedding.

%%%%%%%%%%%%%%%%  Subsubsection 7.2.2 Labelling the roots %%%%%%%%%

\subsubsection{Labelling the roots}

Note that the map $X$ is rational, with degree $LM$, so for generic $x\in \bar\Cb$ we have $\# X^{-1}(x)=LM$.
The preimages  cannot  be globally ordered,  only locally, so let us define:

\begin{definition}\label{def:orderdomain}
Let $\mathcal L^* = X^{-1}(\{ \text{branch points}\})$.
An open domain $U$ of $\bar\Cb $ will be called \textbf{orderable} if it is a connected, simply connected open domain of $\bar\Cb\setminus \mathcal L^*$.
An \textbf{ordered} domain $U$ is an orderable domain together with a map, called an \textbf{ordering},
\bea
U & \rightarrow & \bar\Cb^{LM-1} \cr
z & \mapsto &
 X^{-1}(X(z))\setminus \{z\} = \{ z^{(1)}(z), \dots, z^{(LM-1)}(z) \}
\eea
that is analytic over $U$.
We also define the map $z\mapsto z^{(0)}(z)$ to be the identity
\be
z^{(0)}(z) := z.
\ee
\end{definition}
Note that every orderable domain admits $(LM-1)!$ orderings, obtained by permutations of the preimages.
It is not possible in general  to extend  an ordering analytically to a non--simply connected domain.

As an important illustration, we may chose as orderable domain $U_0$  a small disc around $z=0$ (where ``small'' means
 not containing any point of  $\mathcal L^*$ except for $z=0$), and cut along a segment starting from $0$.
In such $U_0$, there is a canonical labelling of roots such that, as $z\to 0$, we have for $k=1,\dots,LM-1$
\be
z^{(k)}(z) \mathop{\sim}_{z\to 0}  cst\cdot e^{\frac{2i\pi k}{LM-1}} \,\,z^{-\frac{1}{LM-1}} (1+o(1)).
\ee
Over $X(U_0)$, the map $X$ has $LM$ inverses, which we denote 
\be
\tilde z^{(i)}(x),  \quad \text{for } i=0,\dots,LM -1.
\ee
Only one of them, the \emph{physical sheet} $\tilde z^{(0)}(x)$ is a formal power series in $\gamma$ starting with $O(\gamma)$:
\be \label{eq:z0dev}
\tilde z^{(0)}(x)=\gamma x + g_1 s_1 \gamma^2 x^2 + O(\gamma^3 x^3)  \in \gamma x \Kb[s][[\gamma x]].
\ee
All  other branches $\tilde z^{(k)}(x)$, $k = 1, \dots LM -1$ are algebraic functions of $\gamma x$ and can be viewed as formal Puiseux series in $\gamma x$, starting with the expansion 
\be\label{eq:labelRoots}
\tilde z^{(k)}(x) \sim  cst\cdot e^{\frac{2i\pi k}{LM-1}} \,\,(\gamma x)^{-\frac{1}{LM-1}}.
\ee

The following statement is immediate from \eqref{eq:TutteY}.
\begin{proposition}
The formal series
$y(x)=Y(\tilde z^{(0)}(x)) = \tilde W_{0,1}(x)\in \gamma \Kb[x,s][[\gamma]]$ has a finite radius of convergence.
It is in fact a series in $\gamma x$, and its radius of convergence  in $\gamma x$ is 
\be
\rho={\rm min} \left|\frac{a}{G(S(a))}\right|_{a\in \branch}>0.
\ee
\end{proposition}

%%%%%%%%%%%%%%%%  Subsubsection 7.2.3 Galois involutions %%%%%%%%%

\subsubsection{Galois involutions}

\begin{definition}[Galois involution]\label{def:Galois}
Let $a\in \mathcal L$ be a simple ramification point; i.e. $\operatorname{order}_z (X-X(z)) =2$ or $\operatorname{order}_z X=-2$.
There exists a neighbourhood $U_a$ of $a$ and an analytic map $\sigma_a: U_a \to U_a$, different from the identity, such that
\be
X(\sigma_a(z))=X(z).
\ee
It is an involution 
\be
\sigma_a(\sigma_a(z))=z, \quad \quad \sigma_a(a)=a, 
\ee
and $a$ is its only fixed point. We call $\sigma_a$ the Galois involution at $a$.

\end{definition}

\begin{remark}\label{rem:Galois}
For higher order branch points, there is no such involution; instead, there is a local Galois group $\mathcal G_a$ that permutes the preimages of $X(z)$ that lie in a neighbourhood $U_a$ of $a$.
For simple branch points we have that $\mathcal G_a=\{\Ib,\sigma_a\}\sim  \Zb_2$.
\end{remark}

%%%%%%%%%%%%%%%% Section 8. $\beta$-expansions and their poles %%%%%%%%%%%%%%%%

\section{WKB $\beta$-expansions and their poles}
\label{sec:WKB}

%%%%%%%%  Subsection 8.1.  $\beta$-expansion of $\Psi^+_k$, $\Psi^-_k$ and $K$  %%%%%%%%%%%%%

\subsection{WKB $\beta$-expansion of $\Psi^+_k$, $\Psi^-_k$ and $K$}\label{subsec:betaPsik}

In this section we derive a simple recursion for the $\beta$-expansion of the fermionic functions $\Psi^+_k$, $\Psi^-_k$ and $K$. This  allows us to establish the properties of the coefficients of the expansion, which will be useful in the study of the $\beta$-expansions of the bosonic functions in the next section.

\begin{definition} Define, as power series in $\gamma$,
	\bea \label{eq:defCheckPsi}
\check \Psi^\pm_i(x) :=\Psi^\pm_i(x) \exp\left(\mp\beta^{-1}\int_0^x y(u)du\right).
%\check \Psi^-_i(x) := \Psi^-_i(x) \exp\left(\beta^{-1}\int_0^x y(u)du\right).
\eea
\end{definition}

\begin{lemma}\label{lemma:PsiPhiStruct}
For each $i\in  \Zb$, the functions $ \check \Psi^\pm_i(x)$, which are \textit{a priori} elements of $\Kb[x,{\bf s},\beta, \beta^{-1}][[\gamma]]$, are in fact elements of $\Kb[x,{\bf s},\beta][[\gamma]]$, \textit{i.e.} they involve no negative powers of $\beta$. 
%The same is true for the function $\check \Psi^-_i(x)$. 
\end{lemma}
A full proof is given in Appendix \ref{app_A5}. 
Roughly speaking, this follows from the fact that negative powers of~$\beta$ can only come from connected components of the constellations/coverings having genus $0$ and one boundary, which are ``removed'' by the exponential factors in~\eqref{eq:defCheckPsi}. 

For the next lemma, which gives the action of the operators $D$ and $\beta \frac{d}{d\beta}$ on the $\check\Psi^{\pm}_k$'s, 
we need several definitions:
\begin{definition} For $m>0$, define
\bea
\tilde G^{(m)}(x)&\&:= \frac{d^m}{dx^m} \log G(x)- \frac{d^m}{dx^m} \log G(x)\big|_{x=0}
\label{G_tilde_m}
\\
&\& =\sum_{k=m+1}^\infty A_k  \frac{(-1)^{k-1} k! }{(k-m)!}x^{k-m},
\label{G_tilde_m_series}
\eea 
where the $A_k$'s are defined in eq.(\ref{A_k_def}) as normalized power sums over the $c_i$'s.
We also define
\be
\tilde G^{(0)}(x):= \log G(x)
\ee
and $\tilde G^{(-1)}$ as the primitive of $\tilde G^{(0)}$ that vanishes at $0$:
\be
\tilde G^{(-1)}(x)=\sum_{k=1}^\infty \frac{A_k \beta^k (-x)^{k+1}}{k+1}.
\ee
\end{definition}
\begin{remark}
It follows that $\tilde G^{(m)}(0)=0$ for all $m\geq -1$.
\end{remark}

\begin{definition}
Define the operators
\be
O^\pm:=\gamma xG (xy(x) \pm \beta D)
\ee
and
\bea
U^\pm &:=& \frac{\pm 1}{\beta} \left( \int_0^x y(u) du -\sum_{j=1}^L s_j \left({O^\pm}\right)^j \right) \cr
&& + \sum_{m=0}^\infty \beta^{m-1} (-1)^m \frac{B_m}{m!} \Big( (xy(x)\pm \beta D) \tilde G^{(m)}(xy(x)\pm \beta D) \cr
&& + (m-1) \tilde G^{(m-1)}(xy(x)\pm \beta D) \Big),
\eea
where $B_m$ are the Bernoulli numbers, defined by the generating series
\be
{x \over e^x -1} = \sum_{n=0}^\infty {B_n x^n\over n!}.
\label{bernoulli_B_n_gen}
\ee
\end{definition}

\begin{lemma}\label{lemma:actionsOnCheckedPsik}
The functions $\check\Psi^\pm_k(x)$ satisfy the equations
\bea
\label{beta_diff_Psik_check}
\beta \frac{d}{d\beta} \check\Psi^\pm_k(x)=U^\pm \check\Psi^\pm_k(x),
\\
\pm\beta D \check\Psi^\pm_k(x) = (S(O^\pm)-S(z) \pm \beta k) \check\Psi^\pm_k(x).
\label{D_diff_Psik_check}
\eea
\end{lemma}

\begin{remark}
Equations (\ref{beta_diff_Psik_check}), (\ref{D_diff_Psik_check}) depend on $k$ only through the term $\pm \beta k$ in~(\ref{D_diff_Psik_check}).
Eq.~(\ref{D_diff_Psik_check}) is really just a re-writing of eq.(\ref{Psi+_k_spectral_eq}).
\end{remark}

To prove Lemma~\ref{lemma:actionsOnCheckedPsik}, first note that \eqref{D_diff_Psik_check} 
follows from Proposition~\ref{Proppsiphik}. To obtain \eqref{beta_diff_Psik_check}, it is enough to prove that the functions $\Psi^\pm_i(x)$ satisfy the equations
\bea\label{betadif}
\beta \frac{d}{d\beta} \Psi^+_k(x)=\left( \beta \frac{d}{d\beta} T( D-1))-\frac{ 1}{\beta}\xi(R_+,s)\right) \Psi^+_k(x),\\
\beta \frac{d}{d\beta} \Psi^-_k(x)=\left(- \beta \frac{d}{d\beta} T(- D))+\frac{ 1}{\beta}\xi(R_-,s)\right) \Psi^-_k(x),
\label{betadif1}
\eea
where
\bea\label{eq:diffTbeta}
\beta \frac{d}{d\beta} T(x-1) 
&=& \sum_{m=0}^\infty \beta^{m-1} (-1)^m \frac{B_m}{m!} \left( \beta x \tilde G^{(m)}(\beta x) + (m-1) \tilde G^{(m-1)}(\beta x) \right).
\eea
Equations \eqref{betadif}-\eqref{betadif1} also follow from Proposition \ref{Proppsiphik}. Equation \eqref{eq:diffTbeta} follows from  (\ref{Tdef})
and from manipulations involving Bernouilli polynomials. (See Appendix \ref{app_A5}  for a full proof.)

The following lemma is the main technical step of this section.
\begin{lemma}\label{lemma:operatorStructure}
We have the following expansion in powers of $\beta$
\be\label{Uexpan}
O^\pm = \sum_{m=0}^\infty  (\pm \beta)^m O_{m}
\qquad , \qquad
U^\pm = \sum_{m=-1}^\infty  \beta^m U_{m}^\pm.
\ee
	Moreover, after the change of variable $x=X(z)$, $D=\frac{X(z)}{X'(z)}\frac{d}{dz}$, we have 
\be
O_m = \sum_{j=0}^m O_{m,j}(z) \left( z\frac{d}{dz}\right)^j
\qquad , \qquad
U_m^\pm = \sum_{j=0}^{m+1} U_{m,j}^\pm(z) \left( z\frac{d}{dz}\right)^j,
\ee
where the coefficients
$O_{m,j}(z)$ are  rational functions of $z$ with poles only at the zeros of $\phi(z)$,
and  the $U_{m,j}^\pm(z)$'s are rational functions of $z$ with poles only at the zeros of $\phi(z)$ and/or the zeros of $G(S(z))$.
In particular,
\be
U_{-1}^\pm(z)=U_0^\pm(z)=0,
\ee
and
\be
O_0 = z.
\ee
For large $|z|$ the coefficients $U_{j,l}^\pm$ have the asymptotic form
\be
U_{j,l}^\pm=c_{j,l}^\pm z^{-jL}+O(z^{-jL-1}),\,\,\,\,\,\, l>0
\ee
and
\be
U_{j,0}^\pm=c_{j,0}^\pm+O(z^{-1})
\ee
for some constants $c_{j,l}^\pm$.
\end{lemma}
\begin{remark} By slight abuse of notation, we use the same symbols $O_m$ , $U_m^\pm$ to denote operators ``in the variable $x$'' or ``in the variable $z$''. For the functions on which these operators act, we always use unambiguous notation,  so no confusion should occur.
\end{remark}

 The proof proceeds by making the change of variable $x\rightarrow z$, whose Jacobian is
\be
dz/dx = \frac{1}{X'(z)} = \frac{\gamma G(S(z))^2}{\phi(z)},
\ee
 which naturally introduces poles at the zeros of $\phi$. The computation  of $O_0$ and  $U_{-1}^\pm$ is an explicit check, while the computation of $U_0^\pm$ involves a computation with Bernouilli polynomials. The full proof is given in  Appendix \ref{app_A5}.

Now consider the WKB $\beta$-expansions
\be
\check \Psi^\pm_k (x) =:\sum_{j=0}^\infty\beta^j \check \Psi^{\pm{(j)}}_k (x),
\ee
where, by Lemma~\ref{lemma:PsiPhiStruct}, we have $\Psi^{\pm{(j)}}_k (x)\in\Kb[x,{\bf s}][[\gamma]]$. Substituting this expression in~\eqref{beta_diff_Psik_check},
we  obtain the following recursion relations on the $\check\Psi_k^{\pm(m)}$'s
\bea\label{eq:recPsiCheck}
\check\Psi^{\pm(m)}_k (x) =\frac{1}{m} \sum_{j=1}^m U_j^\pm \check\Psi^{\pm (m-j)}_k (x).
\eea 

In particular, in the course of the proof an explicit equation is obtained for the functions $\check\Psi_k^{\pm(0)}$ that is explicitly solved in terms of the $z$ variable:	
\bea\label{eq:PhiPsi0}
\check\Psi^{\pm(0)}_k (X(z)) = \frac{z^k}{\sqrt{\phi(z)}}.
\eea
 (See the proof of Lemma~\ref{lemma:operatorStructure} in Appendix~\ref{app_A5}.)
More generally, as shown in Appendix \ref{app_A5}, the recursion~\eqref{eq:recPsiCheck} can be analyzed using the previous results, which easily imply:
\begin{lemma}
\label{thm:betaPsi_k}
	The functions $\sqrt{\phi(z)} \check\Psi^{\pm(m)}_k (X(z))$ are rational functions of $z$ with poles possible only at the zeros of $\phi(z)$, at $z=\infty$ (for positive $k$) or at $z=0$ (for negative $k$).
\end{lemma}

\begin{remark}\label{rem:squareRootSign}
	The functions $\check\Psi^{\pm(m)}_k (X(z))$ are originally defined only in a neighbourhood of $z=0$, but in view of Lemma \ref{thm:betaPsi_k}, we may view them as analytic functions of $z$ in any orderable domain, up to the choice of a determination of the square-root. This will play a crucial role in the proofs. More conceptually, the necessity of a choice of determination of the square-root is related to the {\em fermionic} nature of the functions $\Psi_k^\pm$. 
\end{remark}

We now arrive at the main results of this section. Using  the Christoffel-Darboux relation (Theorem~\ref{thm:CD}), 
we can transfer information about the $\Psi^\pm_k$'s  to information about the kernel $K(x,x')$. 
Introduce
\be\label{eq:defKcheck}
\tilde{K}(x,x'):=x\,x'\, K(x,x')
\quad , \quad
\check{\tilde{K}}(x,x'):=\exp\left(\beta^{-1}\int_{x}^{x'} y(u) du\right)\tilde{K}(x,x).
\ee
and its series expansion
\be\label{Kchecks}
\check{\tilde{K}}(x,x')=:\sum_{m=0}^\infty \beta^m \check{\tilde{K}}^{(m)}(x,x').
\ee
We then have:
\begin{proposition}\label{prop:diffKbeta}
	\be\label{eq:diffKbeta}
\left(\beta \frac{d}{d\beta}-\sum_{k=1}^\infty \beta ^k (U_k^-+U_k^+)\right)\check{\tilde{K}}(x,x')=0.
\ee
\end{proposition}

\begin{proof}  Proposition \ref{Proppsiphi} implies that
\bea
\tilde{K}(x,x')&=&x\, x'\, e^{T\left(D\right)-T\left(-D'-1\right)}
\left({e^{\beta^{-1}\xi({\bf s},x) - \beta^{-1}\xi ({\bf s},x')}\over x-x'}\right)\\
&=& e^{T\left(D-1\right)-T\left(-D'\right)}
\left(x x' {e^{\beta^{-1}\xi({\bf s},x) - \beta^{-1}\xi ({\bf s},x')}\over x-x'}\right),
\eea
and hence
\be
\beta \frac{d}{d \beta}\tilde{K}(x,x')=\left( \beta \frac{d}{d\beta} \left(T(D-1)- T(-D')\right)-\frac{1}{\beta}\xi(R_+,s)+\frac{1}{\beta}\xi(R_-',s)\right) \tilde{K}(x,x').
\ee
Equation  \eqref{eq:diffKbeta} follows from comparison of this equation with (\ref{betadif}) and (\ref{betadif1}) and the definition of the operators $U_k^\pm$.
\end{proof}

\begin{proposition}\label{prop:betaK}
	All  the $\sqrt{\phi(z)\phi(z')} \check{\tilde{K}}^{(j)}(X(z),X(z'))$'s are rational functions of $z$ and $z'$, with poles only at the zeros of $\phi(z)$, $\phi(z')$, and $z=z'$. Moreover there are no poles at $z=z'$ for $j>0$.
\end{proposition}

\begin{proof}
	Proposition~\ref{prop:diffKbeta} implies the recursion relations
\be
\check{\tilde{K}}^{(m)}(X(z),X(z'))=\frac{1}{m} \sum_{j=1}^m \left(U_j^+(z)+U_j^-(z')\right) \check{\tilde{K}}^{(m-j)}(X(z),X(z')), 
\ee
which enable us to compute the $\check{\tilde{K}}^{(m)}(z,z')$ by induction. 
The induction starts from
\be\label{eq:Kcheck0}
\check{\tilde{K}}^{(0)}(X(z),X(z'))=  \frac{1}{\gamma\sqrt{\phi(z)\phi(z')}} \frac{ z z'}{z-z'}
=\frac{X(z)X(z')}{(z-z')\sqrt{X'(z)X'(z')}},
\ee
which follows from~\eqref{eq:PhiPsi0} and  the Christoffel-Darboux relation (Proposition \ref{thm:CD}, eq.~(\ref{eq:CD})). Using Lemma \ref{lemma:operatorStructure}, we conclude that the $\sqrt{\phi(z)\phi(z')} \check{\tilde{K}}^{(j)}(X(z),X(z'))$'s are rational functions, with poles possible only at $z=z'$ and at the zeros of $\phi(z)$, $\phi(z')$, $G(S(z))$, and $G(S(z'))$.
 
From the Christoffel-Darboux relation it also follows that the poles of the functions $\sqrt{\phi(z)\phi(z')} \check{\tilde{K}}^{(j)}(X(z),X(z'))$  can appear only at the poles of $\sqrt{\phi(z)}\check{\Psi}_k^{\pm(m)}$ and at $x=x'$, so from Lemma \ref{thm:betaPsi_k} the poles at the zeros of $G(S(z))$, and $G(S(z'))$ are excluded.
The last sentence of the Proposition follows from Proposition \ref{prop:explicitFunctions}.
\end{proof}

%%%%%%%%% Section 8.2 Definition and structure of $\tilde \omega_{g,n}(z_1,\dots,z_n)$%%%%%%%%%%%%%

\subsection{Definition and poles of $\tilde \omega_{g,n}$}\label{subsec:defomegagn}
\label{subsec:polesomegagn}

The goal of this section is to analyse the analytic properties of the functions $\tilde W_{g,n}$ or, more precisely, the closely related 
symmetric rank-$n$ differential forms $\tilde\omega_{g,n}(z_1,\dots,z_n)$ on the spectral curve that are defined below. 
This will be important in the proof of topological recursion in Section~\ref{sec:toprec}.

First, from the structural results on $\check{\tilde K}(X(z),X(z'))$ obtained in the previous section and from Proposition~\ref{prop:detConnected} in Section~\ref{subsec:fermionsToBosons}, we have the following fact.

\begin{proposition}\label{prop:Wgnrational}
For $g\geq 0, n\geq 1$, the series $\tilde W_{g,n}(X(z_1),\dots,X(z_n))$ is rational in the variables $z_1,\dots,z_n$.
More precisely:
	\be \label{eq:Wgnrational}
	\tilde W_{g,n}(X(z_1),\dots,X(z_n))	\in \gamma^n \Kb(\mathbf{s},z_1,\dots,z_n).
\ee

\end{proposition}
\begin{remark}The nontrivial content of Proposition~\ref{prop:Wgnrational} is the dependency on $\mathbf{s}$ and the $z_i$'s. The fact that $\tilde W_{g,n}(X(z_1),\dots,X(z_n))$ is a multiple of $\gamma^n$ is clear from the fact that $\tilde F_{g,n}(X(z_1),\dots,X(z_n))$ is independent of $\gamma$, which follows from the fact that, in each monomial appearing in \eqref{eq:defFgn_tilde}, the total power in $x_i$ is equal to the power in $\gamma$, and  that $\gamma X(z)$ is a function of $z$ independent of $\gamma$.
\end{remark}

This proposition allows us to redefine the generating functions $\tilde W_{g,n}$ as differential forms in the $z_i$ variables:

\begin{definition}\label{def:omegagn}

	Define $\tilde\omega_{g,n}(z_1,\dots,z_n) \in \Kb(\mathbf{s})(z_1,\dots,z_n) dz_1\dots dz_n$ by
	\bea\nonumber
	\tilde \omega_{g,n}(z_1,\dots,z_n) =&& \tilde W_{g,n}(X(z_1),\dots,X(z_n)) 
	X'(z_1)\dots X'(z_n)
	dz_1\dots dz_n\\
	&&+\delta_{n,2}\delta_{g,0}\frac{X'(z_1)X'(z_2)}{(X(z_1)-X(z_2))^2} dz_1 dz_2\label{defomegagn}
	.
\eea

\end{definition}
%For example one has:
%\be\label{def:omega01}
%\tilde \omega_{0,1}(z) = Y(z) X'(z) dz
%\ee
The addition of the extra double-pole term for $(g,n)=(0,2)$ in Definition~\ref{def:omegagn} will prove convenient for
the statement of the topological recursion formulae; it is standard in other similar cases, see e.g.~\cite{EO1}.
%The following theorem is a crucial element for what follows.
%\begin{theorem}\label{thm:omegagnPoles}
%	For $2g+n-2 >0$, the rational function $\frac{1}{dz_1\dots dz_n} \tilde \omega_{g,n}(z_1,\dots,z_n)$ has poles only when $z_i \in \mathcal{L}$ for $i=1,\dots,n$. Moreover it behaves like $O(z_i^{-2})$ at $z_i=\infty$, which means that the differential $\tilde \omega_{g,n}(z_1,\dots,z_n)$ is analytic at $z_i=\infty$.
%\end{theorem}
%The proof of these statements is given in Section~\ref{subsec:polesomegagn}. 
%In preparation, we need  to study the analytic behavior of the $\beta$-expansion of the functions $\Psi^+_k(x),\Psi^-_k(x)$, and $K(x,x')$.

%\rmkB{Here you see that all square roots $\sqrt{\phi(z)}$ disappear, indeed the $W_{g,n}$ are bosonic. They belong to integer powers of the canonical line bundle. However, we have to say it somewhere. A suggestion is to define everything in orderable domains (regarding fiber bundles, this means in a chart where the bundle can be trivialized), and in the end, since we obtain only rational functions, we extend them to the full spectral curve. A posteriori.}

A finer study of the poles leads to the following theorem.

\begin{theorem}
	\be\label{eq:omega0102}
	\tilde \omega_{0,1}(z) = Y(z)X'(z) dz \quad , \quad \tilde \omega_{0,2}(z_1,z_2) = \frac{dz_1dz_2}{(z_1-z_2)^2},
\ee
	and for $(g,n) \not \in \{(0,1),(0,2)\}$, the quantity
	\begin{equation}
\frac{\tilde \omega_{g,n}(z_1,\dots,z_n)}{dz_1 \dots dz_n}
\end{equation}
is a rational function of its variables $z_i$, with poles only at  the zeros $\phi(z_i)=0$
of $\phi$. In particular, there is no pole at $z_i=0$ nor at $z_i=\infty$, nor $z_i=z_j$.
\end{theorem}

\begin{remark}[Continuation of Remark~\ref{rem:squareRootSign}]
Although the right-hand side of \eqref{defomegagn} makes sense   {\it a priori} only as a power series, or for the $z_i$ in a small neighborhood of zero\footnote{Note that for each pair $(g,n)$ we have an upper bound of the form $(ct)^d$ for the sum of the coefficients appearing in $\tilde W_{g,n}$ at  order $d$ in the variables $x_i$. Such a bound is clear from the graphical interpretation, since the number of embedded graphs of given Euler characteristic grows no more than exponentially in the number of edges. This ensures that the series 
$\tilde W_{g,n}$ are convergent in a neighbourhood of $x_i=0$, and hence that the $\tilde \omega_{g,n}(z_1,\dots,z_n)/(dz_1\dots dz_n)$ are convergent in a neighborhood of $z_i=0$.}, the theorem enables us to view $\frac{\tilde \omega_{g,n}(z_1,\dots,z_n)}{dz_1 \dots dz_n}$, and hence $\tilde W_{g,n}(X(z_1),\dots,X(z_n))$, as a \emph{globally defined} function of the variables $z_i$. In particular, although the proof makes use of the fermionic functions of the previous section, the result does not depend on the determination of the square root $\sqrt{\phi(z)}$. This is related to the {\em bosonic} nature of the functions $\tilde W_{g,n}$.
\end{remark}

\begin{proof}
We  use the notation 
	\be K^{(m)}(x,x') := \frac{1}{xx'} \check{\tilde{K}}^{(m)}(x,x').
	\ee
	Using the first equation of Proposition~\ref{prop:detConnected} we have, 
	\bea\label{w1proof}
\sum_g \beta^{2g-1} \frac{\tilde \omega_{g,1}(z)  }{dz} 
&=& \sum_g \beta^{2g-1} \tilde W_{g,1}(X(z))  X'(z) \cr
&=& X'(z) \lim_{z'\to z} \left(K(X(z),X(z')) - \frac{1}{X(z)-X(z')} \right).
	\eea
From the relation~\eqref{eq:defKcheck} between $K$ and $\check{\tilde{K}}$, inserting the $\beta$-expansion~\eqref{Kchecks}, and using the explicit expression~\eqref{eq:Kcheck0}, we get
\bea
	\eqref{w1proof}	&=&	X'(z) \lim_{z'\to z} \left(\frac{e^{\beta^{-1}\int_{z'}^{z} Y(u)X'(u)du}}{(z-z')\sqrt{X'(z)X'(z')}} - \frac{1}{X(z)-X(z')} +\sum_{m\geq 1} \beta^{m} {K}^{(m)}(X(z),X(z')) \right)\cr
	&=& Y(z)X'(z) +X'(z) \sum_{m\geq 1} \beta^{m} K^{(m)}(X(z),X(z)).\label{w1proof2} 
\eea
	By Proposition~\ref{prop:betaK}, and using
	\be
	\frac{X'(z)}{\phi(z)X(z)^2}=\frac{\gamma}{z^2},
	\ee
	 this shows that $\tilde \omega_{g,1}(z)/dz$ is a rational function of $z$, with poles only at the ramification points 
	for $g>0$. We already know that there is no pole at $z=0$  since, by definition, we are working with well-defined power series in $z$. Note also that the order $\beta^0$ term in~\eqref{w1proof2} is the first part of~\eqref{eq:omega0102}.

For $n=2$ we use the second equation of Proposition~\ref{prop:detConnected} to compute
\bea
\sum_g \beta^{2g} \frac{\tilde \omega_{g,2}(z_1,z_2)  }{dz_1 dz_2} 
&=& \frac{X'(z_1)X'(z_2)}{(X(z_1)-X(z_2))^2} + \sum_g \beta^{2g} \tilde W_{g,2}(X(z_1),X(z_2))  X'(z_1) X'(z_2) \cr
&=& - K(X(z_1),X(z_2)) K(X(z_2),X(z_1)) \ X'(z_1) \ X'(z_2) \cr
&=& \frac{1}{(z_1-z_2)^2} 
+ \frac{\sqrt{X'(z_1)X'(z_2)}}{z_1-z_2} \sum_{m\geq 1} \beta^m K^{(m)}(X(z_1),X(z_2)) \cr
&& + \frac{\sqrt{X'(z_1)X'(z_2)}}{z_2-z_1} \sum_{m\geq 1} \beta^m K^{(m)}(X(z_2),X(z_1)) \cr 
&& - \sum_{m,m'\geq 1} \beta^{m+m'}  K^{(m)}(X(z_1),X(z_2)) K^{(m')}(X(z_2),X(z_1)) \ X'(z_1) \ X'(z_2) \cr
&&
\eea

	By Proposition~\ref{prop:betaK}, this is also a rational function of $z_1,z_2$. The order $\beta^0$ term gives the second part of~\eqref{eq:omega0102}.
	Higher orders in $\beta$ could {\it a priori} have simple poles at $z_1=z_2$, but this cannot occur since, by definition, they are symmetric functions of $z_1,z_2$.
	Thus to each order in powers of $\beta$, except $\beta^0$, the poles are only at the ramification points as claimed.
	 (Poles at $z_i=0$ are again excluded since we work with valid power series.)

Finally for $n\geq 3$ we use the third equation of Proposition~\ref{prop:detConnected}:
	\bea
\sum_g \beta^{2g-2+n} \frac{\tilde \omega_{g,n}(z_1,\dots,z_n)  }{dz_1 \dots dz_n} 
&=& \sum_g \beta^{2g-2+n} \tilde W_{g,n}(X(z_1), \dots ,X(z_n))  X'(z_1) \dots  X'(z_n) \cr
&=& \sum_{\sigma\in \mathfrak S_{n}^\text{1-cycle}} \sgn(\sigma) \prod_{i=1}^n
 K(X(z_i),X(z_{\sigma(i)})) \ X'(z_1) \dots  X'(z_n). \cr
 &&
\eea
	Since each $K$ has a simple pole at coinciding points, each  $\tilde \omega_{g,n}$  could {\it a priori} have poles (at most simple) at such points, but a symmetric function of $(z_i, z_j)$ can have no simple pole at $z_i=z_j$, so the proposition is proved.
\end{proof}

%%%%%%%%%%%%%%%  Section 9.  Loop equations %%%%%%%%%%%%%%

\section{Fundamental system and loop equations}\label{sec:loop}
\label{loop_eqs_fundamental_sys}

%%%%%%%%%%%%%%%  Subsection 9.1 The fundamental system %%%%%%%%%%%%%%
\subsection{The fundamental system}
\label{fund_sys}

According to Theorem \ref{thm:kspectral0}, $\Psi_0^\pm(x)$ is annihilated by an order $LM$ differential operator 
(eqs.~(\ref{eq:quantumCurve}), (\ref{eq:quantumCurveDual})).
There therefore exist $LM-1$ other linearly independent solutions to the same equations. In this section, we 
express them explicitly as power series in $\beta$.

\begin{definition}

In an ordered domain $U$ (see Def.~\ref{def:orderdomain}), define the following 
diagonal matrices of size $LM\times LM$:
\be
\mathbf Y(z) := \text{diag}(Y(z^{(j)}(z)))
\quad , \quad
{\mathbf \Phi}(z) := \text{diag}(\phi(z^{(j)}(z))) ,\quad j=0, \dots , LM-1
\ee
and the Vandermonde matrix
\be
\mathbf V(z)_{i,j} := (z^{(j)}(z))^i, \quad  i, j= 0, \dots, LM-1.
\ee
Since $U$ is simply connected and avoids the zeros of $\phi$, we may choose a sign for the square root $\sqrt{\phi(z^{(j)}(z))}$, so that the matrix $\sqrt{\Phi(z)}$ is analytic and well defined over $U$.

We further define the $LM \times LM$ matrices
	\bea\label{subtleDef}
	\check {\mathbf \Psi}^\pm(z)_{i,j}  :=  \left.\check \Psi^\pm_i( X(z ) ) \right|_{z=z^{(j)}(z)},
%\check {\mathbf \Psi}^-(z)_{i,j}  :=  \check \Psi^-_i(\tilde z^{(j)}(z)) 
%\in \Kb(\{z^{(j)}(z)\})[[\beta]]
% \cr
%&\&
\eea
where the symbol $|_{z=z^{(j)}(z)}$ means that we substitute $z$ by $z^{(j)}(z)$ in each coefficient of the $\beta$-expansion of the function $\check \Psi^\pm_i( X(z ) )$. 
\end{definition}
We emphasize the fact that, from Lemma~\ref{thm:betaPsi_k} and Remark~\ref{rem:squareRootSign}, each coefficient in the $\beta$-expansion of the function $\check \Psi^\pm_i( X(z ) )$ is a globally defined function of $z$ (up to the choice of the square root), hence the substitution performed in~\eqref{subtleDef} is well-defined. We also insist on the fact that, although $X(z^{(j)}(z))=X(z)$ for any value of $j$, in general $\left.\check \Psi^\pm_i( X(z ) ) \right|_{z=z^{(j)}(z)}$ is not equal to $\check \Psi^\pm_i( X(z ) )$. 

\begin{proposition}\label{prop:diffMatrices}
To leading order in $\beta$, these matrices satisfy
\begin{equation}
\sqrt{{\mathbf \Phi}(z)} \ \check{\mathbf \Psi}^\pm(z) \sim \mathbf V(z) +O(\beta).
%, \quad
%\sqrt{{\mathbf \Phi}(z)} \ \check{\mathbf \Psi}^-(z) \sim \mathbf V(z) +O(\beta).
\end{equation}
Moreover, they satisfy the differential equations
\bea
\label{eqDPsimat}
	\pm \beta X(z) \frac{d}{dX(z)} \check{\mathbf \Psi}^\pm(z) =  {\bf E}^\pm(X(z)) \check{\mathbf \Psi}^\pm(z) - 
	X(z) \check{\mathbf \Psi}^\pm(z) \mathbf Y(z)
%	-\beta X(z) \frac{d}{dX(z)} \check{\mathbf \Psi}^-(z) = {\bf E'}(X(z)) \check{\mathbf \Psi}^-(z) 
%	- X(z) \check{\mathbf \Psi}^-(z) \mathbf Y(z)
\eea
and
\begin{equation}\label{eqPsiAPsi}
	\check{\mathbf \Psi}^-(z)^T {\bf A}^T \check{\mathbf \Psi}^+(z) = \Ib,
\end{equation}
which implies, in particular, that they are invertible.
\end{proposition}

\begin{proof}
The leading order term follows from \eqref{eq:PhiPsi0}.
From Theorem \ref{prop:finiteSystem} we know that the first column vector of $\check {\mathbf{\Psi}}^+(z)$ satisfies the first column of  \eqref{eqDPsimat}. 
It remains to prove that all other columns do.
	The $j^{\rm th}$ column vector of $\check {\mathbf{\Psi}}^+(z)$ is equal to the first column vector of $\left.\check {\mathbf{\Psi}}^+(z)\right|_{z=z^{(j)}(z)}$, so it satisfies
\be
	\beta X(z^{(j)}(z)) \frac{d}{dX(z^{(j)}(z))} \check{\mathbf \Psi}^+_{i,j}(z) =  \sum_{l} {\bf E}^+_{i,l}(X(z^{(j)}(z))) \check{\mathbf \Psi}^+_{l,j}(z) - 
X(z^{(j)}(z))\check{\mathbf \Psi}^+_{i,j}(z) Y(z^{(j)}(z)), 
\ee
and since $X(z^{(j)}(z))=X(z)$, it satisfies the $j^{\rm th}$ column of \eqref{eqDPsimat}.
Here we have used the fact  that $X(z^{(j)}(z))=X(z)$ implies
$X'(z^{(j)}(z)) \,z^{(j)'}(z) = X'(z)$, and thus
\be
\frac{d}{dX(z^{(j)}(z))}=\frac{1}{X'(z^{(j)}(z)) \,z^{(j)'}(z)}\ \frac{d}{dz}  = \frac{1}{X'(z)} \ \frac{d}{dz} = \frac{d}{dX(z)}.
\ee

	We now prove \eqref{eqPsiAPsi}. We use the notation 
	\be\mathbf{\Lambda}(z):=\sum_{m\geq 0} \mathbf{\Lambda}^{(m)}(z) \beta^m := \check{\mathbf \Psi}^-(z)^T {\bf A}^T \check{\mathbf \Psi}^+(z),
	\ee and we will prove that $\mathbf{\Lambda}(z)=\Ib$ in three steps.
	First, the Christoffel-Darboux relation (Theorem~\ref{thm:CD}) shows that all the diagonal terms of $\mathbf{\Lambda}(z)$ are equal to $1$. Secondly, we claim that $\mathbf{\Lambda}^{(0)}(z)$ is equal to $\Ib$. To see this, recall~\eqref{eq:PhiPsi0}, so the claim is equivalent to the fact that for all $0\leq a,b\leq LM-1$ one has
	\be\label{eq:secondstep}
	\sum_{0\leq i,j<LM} \mathbf{A}^{(0)}_{i,j} z^{(a)}(z)^i z^{(b)}(z)^j = \delta_{a,b}\sqrt{\phi(z^{(a)}(z))\phi(z^{(b)}(z))},
	\ee
where
	$\mathbf{A}^{(0)}_{i,j}$ is the coefficient of $\beta^0$ in $\mathbf{A}_{i,j}$. Note that we only need to prove it for $a\neq b$, since diagonal terms were dealt with in the first step.
Now, from \eqref{eq:defDeltapm}--\eqref{eq:defVpm} one has $V_\pm(x)=G(S(x))+O(\beta)$, so that the generating polynomial $A(r,t)$ of the entries of the matrix $A$ satisfies, by~\eqref{eq:defA}:
\be
A(r,t) = \frac{rG(S(t))-tG(S(r))}{r-t} + O(\beta), 
\ee
which implies that, for $a\neq b$,
\be
\sum_{0\leq i,j<LM} \mathbf{A}^{(0)}_{i,j} z^{(a)}(z)^i z^{(b)}(z)^j=
\frac{rG(S(r'))-r'G(S(r))}{r-r'}
\ee
	where  $r=z^{(a)}(z)$ and $r'=z^{(b)}(z)$. The numerator is equal to zero by definition of the $z^{(i)}(z)$, and since $r\neq r'$ this proves \eqref{eq:secondstep} and the claim. 
Thirdly, we write
\bea
	&&\beta X(z) \frac{d}{dX(z)}  \left( \check{\mathbf \Psi}^-(z)^T {\bf A}^T \check{\mathbf \Psi}^+(z)
	\right) \cr
	&&= \check{\mathbf \Psi}^-(z)^T ({\bf A}^T {\bf E^+}(X(z)) - {\bf E^-}(X(z))^T {\bf A}^T) \check{\mathbf \Psi}^+(z) 
	 + X(z) [ {\bf Y}(z), \check{\mathbf \Psi}^-(z)^T {\bf A}^T \check{\mathbf \Psi}^+(z) ]
	 \cr
	&&=X(z) [ {\bf Y}(z), \check{\mathbf \Psi}^-(z)^T {\bf A}^T \check{\mathbf \Psi}^+(z) ],
\eea
	where we used~\eqref{maste}. Equivalently, we have for $m\geq 1$
	\be
	X(z) \frac{d}{dX(z)} \mathbf{\Lambda}^{(m-1)}(z)  = 
	X(z) [ {\bf Y}(z), \mathbf{\Lambda}^{(m)}(z)].
	\ee
	If for some $m\geq 1$ we assume that $\mathbf{\Lambda}^{(m-1)}(z)$ is a constant matrix (which we know is true for $m=1$, by the second step), this equation implies that $\mathbf{\Lambda}^{(m)}(z)$ commutes with ${\bf Y}(z)$, which is a diagonal matrix with distinct entries on the diagonal. This implies that $\mathbf{\Lambda}^{(m)}(z)$ is diagonal, which implies that it is equal to zero by the conclusion of first step. Therefore we can apply induction on $m\geq 1$ and conclude that $\mathbf{\Lambda}^{(m)}(z)$ is constant and equal to zero for all $m\geq 1$. This concludes the proof of \eqref{eqPsiAPsi}.
\end{proof}

By Theorem~\ref{thm:CD} we obtain:
\begin{corollary}\label{cor:KPsiinv}
	\be
	\check K(X(z),X(z')):= K(X(z),X(z')) e^{\beta^{-1}\int_{z}^{z'} Y(u)X'(u)du}=  \left( \frac{\check{\mathbf \Psi}^-(z)^{-1} \ \check{\mathbf \Psi}^-(z') } {X(z)-X(z')} \right)_{0,0}.
\ee
\end{corollary}

%%%%%%%%%%%%%%%%% Subsection 9.2. Loop equations %%%%%%%%%%%%%%%

\subsection{Loop equations}
\label{loop_eqs}
In this subsection we give the main preparatory result needed for deriving the topological recursion
 equations: the {\em Loop Equations}.

\begin{definition}
Define the matrix
	\be
 {\bf D}(x): = \frac{1}{\beta x} {\bf E^-}(x), 
\ee
whose coefficients are rational functions of $x$, 
%It follows as a consequence of Corollary~\ref{cor:KPsiinv}, that 
%\be
%\left(\frac{d}{dx}\check{\mathbf \Psi}^+(x) e^{\beta\int_0^z \mathbf Y(z') X'(z')dz'} \right) e^{-\beta\int_0^z \mathbf Y(z') X'(z')dz'} \check{\mathbf \Psi}^+(x)^{-1} = {\bf D}(x).
%\ee
%\rmkA{By definition, $\check{\mathbf \Psi}^+(z)$ is a function of $z$, not $x$. Is the statement obvious? I do not understend, why it follows form \ref{cor:KPsiinv}.}
as well as the matrix
	\bea\label{eq:defD}
\tilde{{\bf D}}_n(x) 
&:=& {\bf D}(x) + \sum_{i=3}^n \epsilon_i \frac{ {\bf M}(z_i)}{(x-X(z_i))(X(z_i)-x)} \cr
	&& + \sum_{3\leq i,j\leq n\atop  i\neq j} \epsilon_i \epsilon_j \frac{ {\bf M}(z_i){\bf M}(z_j)}{(x-X(z_i))(X(z_i)-X(z_j))(X(z_j)-x)} \cr
&& + \sum_{k=3}^{n} \sum_{ 3 \leq i_1,\dots,i_k \leq n \atop i_1\neq \dots \neq i_k  } \epsilon_{i_1}\dots \epsilon_{i_k} \frac{ {\bf M}(z_{i_1})\dots {\bf M}(z_{i_k})}{(x-X(z_{i_1}))(X(z_{i_1})-X(z_{i_2})) \dots (X(z_{i_k})-x)} , \cr
&\&
\eea
where
	\be\label{eq:defM}
	{\bf M}(z):= \check{\mathbf \Psi}^-(z)^{-1} \ \text{diag}(1,0,0,\dots,0) \ \check{\mathbf \Psi}^-(z).
\ee
\end{definition}
%	Note that the coeficients of $\mathbf{D}(x)$ are rational functions of $x$.

\begin{theorem}\label{thm:loopeqs}
Let $U$ be an ordered domain.
The following equation, called the \textbf{first loop equation}, is satisfied for all $z\in U$
\be\label{loopeq11}
\sum_{k=0}^{LM-1} 
\sum_{g=0}^\infty \beta^{2g-1} \tilde \omega_{g,1}(z^{(k)}(z))
= \operatorname{Tr} {\bf D}(X(z)) ,
%\ -\delta_{g,0}\delta_{n,1}\frac{L g_{M-1}}{g_{M}}\, \frac{ dX(z)}{X(z)} + \delta_{g,0}\delta_{n,2} \frac{dX(z)dX(z_2)}{(X(z)-X(z_2))^2}.
\ee
and if $n\geq 2$
\be\label{loopeq12}
\sum_{k=0}^{LM-1} 
 \tilde \omega_{g,n}(z^{(k)}(z),z_2, \dots ,z_n)
=  \delta_{g,0}\delta_{n,2} \frac{dX(z)dX(z_2)}{(X(z)-X(z_2))^2}.
\ee
The \textbf{second loop equation} is the statement that for all $z\in U$  and  $n\geq 2$,
\begin{multline}\label{loopeq2}
Q_{g,n}(X(z);z_3,\dots,z_n) 
:= \frac{1}{dX(z)^2}\sum_{0\leq k<l\leq LM-1} \Big(
\tilde\omega_{g-1,n}(z^{(k)}(z),z^{(l)}(z);z_3,\dots,z_n) \cr
+ \sum_{g_1+g_2=g}\sum_{I_1\uplus I_2=\{z_3,\dots,z_n\}}
\tilde\omega_{g_1,1+|I_1|}(z^{(k)}(z),I_1)\ \tilde\omega_{g_2,1+|I_2|}(z^{(l)}(z),I_2) \Big) \end{multline}
is a rational function of $X(z)$ with no poles at the branch points.

More precisely
\be
\sum_g \beta^{2g} Q_{g,2}(x) = \frac12 \left(\Tr {\bf D}(x)^2 - (\Tr {\bf D}(x))^2\right),
\ee
and for $n>2$
\be
\sum_g \beta^{2g+n} Q_{g,n}(x;z_3,\dots,z_n) = [\epsilon_3\dots\epsilon_n] \frac12 \left(\Tr \tilde{{\bf D}}_n(x)^2 - (\Tr \tilde{{\bf D}}_n(x))^2\right).
\ee
\end{theorem}
The proof of this theorem follows along the same  lines as in \cite{Belliard-Eynard-Marchal2016}; a self-contained version is given in Appendix \ref{app_A6}. 

\begin{remark}
	Because ${\bf E}^+$ is conjugate to ${{\bf E}^-}^T$, we could replace the matrix ${\bf E}^-$ by ${\bf E^+}^T$ in the definition of ${\bf D}(x)$, and $\check{\mathbf \Psi}^+$ by $\check{\mathbf \Psi}^-$ in the definition of ${\bf M}$, and the same loop equations would hold.
\end{remark}

%%%%%%%%%%%%%%%% Section 10. Topological recursion %%%%%%%%%%%%%%%%

\section{Topological recursion}
\label{sec:toprec}

Since the $\tilde\omega_{g,n}$'s satisfy the linear and quadratic loop equations, and have poles only at branch points, from the theorem in \cite{EO1} (and more generally of \cite{BorotEynardOrantin1303.5808}), this implies that the $\omega_{g,n}$'s satisfy the topological recursion relations.

%%%%%%%%%%%%%%%% Subsection 10.1 %%%%%%%%%%%%%%%%

\subsection{Topological recursion for the $\tilde \omega_{g,n}$'s}
\label{top_rec_hurwitz}

Define
\be
\mathcal{K}(p;z_1,z_2) := \frac12 \ 
\Big[ \frac{dp}{z_1-p} - \frac{dp}{z_2-p}  \Big] \,\frac{z_1 G(S(z_1)) }{(S(z_1)-S(z_2))\,\phi(z_1) dz_1}.
\label{K2_def}
\ee
%and, for $l\ge 2$,
%\be
%K^{(l)}_a(p;z_1,z_2,\dots,z_l) := 
% \frac{dp}{z_1-p}  \,\prod_{i=2}^l \frac{z_1 G(S(z_1)) }{(S(z_1)-S(z_i))\,\phi(z_1) dz_1}.
%\ee
The main result of the  paper is the following theorem.
\begin{theorem}\label{thm:toprec}
Assume that all ramification points $a\in\mathcal L$ are simple, with local Galois involution $\sigma_a$.
	Then the  $\tilde\omega_{g,n}$s satisfy the following topological recursion equations
%\begin{itemize}
%\item[$\bullet$] If
\be
\tilde\omega_{g,n}(z_1,\dots,z_n) 
= -\,\sum_{a\in \mathcal L} \Res_{z\to a} 
		\mathcal{K}(z_1;z,\sigma_a(z)) 
\ {\mathcal W}_{g,n}(z,\sigma_a(z);z_2,\dots,z_n),
\label{omega_tilde_K}
\ee
%\rmkB{Check the overall sign in front}
where
\begin{multline}
{\mathcal W}_{g,n}(z,z';z_2,\dots,z_n)
= \tilde \omega_{g-1,n+1}(z,z',z_2,\dots,z_n) \\
+ \sum'_{g_1+g_2=g,\, I_1\uplus I_2=\{z_2,\dots,z_n\}}
\tilde \omega_{g_1,1+|I_1|}(z,I_1)\tilde \omega_{g_2,1+|I_2|}(z',I_2)
\label{W2gn_TR}
\end{multline}
where $\sum'$ means that we exclude the 2 terms $(g_1,I_1)=(0,\emptyset)$ and $(g_2,I_2)=(0,\emptyset)$.
%\item[$\bullet$] If some ramification points have higher order, then the $\tilde\omega_{g,n}$s satisfy the higher order version of the topological recursion defined in \cite{BouchardEynardlocalglobal},  which involves the local Galois group $\mathcal G_a$ of remark.~\ref{rem:Galois}. 
%
%\begin{multline}
%{\mathcal W}^{(l)}_{g,n}(z_1,z_2,\dots,z_l;p_2,\dots,p_n)
%= \sum_{\mu \vdash \{z_1,\dots,z_l\}} \; 
%\sum_{I_1\uplus \dots \uplus I_{\ell(\mu)} \atop = \{p_2,\dots,p_n\}} \;
%	\sum'_{g_1+g_2+\dots+g_{\ell(\mu)} \atop = g-1+\ell(\mu)}\;
%\prod_{i=1}^{\ell(\mu)} 
%\tilde \omega_{g_i,\#\mu_i+\# I_i}(\mu_i \cup I_i), \\
%\label{Wlgn_TR}
%\end{multline}
%where $\sum_{\mu \vdash \{z_1,\dots,z_l\}}$ means the sum over all ways of partitioning $\{z_1,\dots,z_l\}$ into non empty parts, and $\ell(\mu)$ is the number of parts,
%and $\sum'$ means that we exclude all terms containing a $\tilde\omega_{0,1}$ factor.
%Then we have
%\be
%\tilde\omega_{g,n}(z_1,\dots,z_n) 
%= -\,\sum_{a\in \mathcal L} \Res_{z\to a} 
%\sum_{P\subset \mathcal G_a\setminus \{\Ib\} }
%K^{(1+\# P)}_{a(z_1;z,\{\sigma(z) | z\in P\} }
%\ {\mathcal W}^{(1+\# P)}_{g,n}(z,\{\sigma(z) | z\in P\};z_2,\dots,z_n) 
%\label{tilde_omega_gn}
%\ee
%\end{itemize}
\end{theorem}
\begin{proof}
In view of what has already been proved,  this result follows from arguments similar  to those in~\cite{EO1}. 
For the convenience of the reader we provide the details in Appendix \ref{app_A7}.
\end{proof}
\begin{remark}
It is natural to expect that, if  branch points of  higher order occur, the higher order version of the topological recursion 
relations introduced in~\cite{BouchardEynardlocalglobal} holds, and that the ideas of~\cite{BouchardEynardlocalglobal} together 
with our intermediate results may be used to  prove the corresponding generalization. However, for the sake of brevity, we do not address this here, leaving it rather as an open problem for future work.
\end{remark}

%%%%%%%%%%%%%%%% Subsection 10.2 %%%%%%%%%%%%%%%%
 
\subsection{Applications, examples, and further comments.}
\label{sec:F03}

Topological recursion has many consequences (see e.g. the review~\cite{EO-review}). The first of these is that it enables one to compute each of the generating functions $\tilde W_{g,n}$ or $\tilde F_{g,n}$ for given $g$ and $n$ in closed form as a rational function of the variables $z_i$. As an illustration, we give here the value of the function $\tilde F_{0,3}$, which in the case $L=1$ had been conjectured in the context of combinatorial enumeration by John Irving\footnote{Personal communication to G.C.}. Explicit details of the (short) calculation 
needed to verify it are given in Appendix \ref{app_A7}.
\begin{proposition}\label{prop:F03}
	Assuming all branch points simple, we have
	\be \label{eq:F03}
\tilde F_{0,3}(X(z_1),X(z_2),X(z_3))
= -\sum_{i=1}^3 \frac{1}{\prod_{j\neq i} (z_i-z_j)} \  \frac{z_i^2 G'(S(z_i))}{\phi(z_i)  }  .
\ee
\end{proposition}
Here, as everywhere in the paper, we have assumed that $G$ and $S$ are polynomials. However, let us consider momentarily  the case where $G$ and $S$ are both formal power series. From the definitions, it follows that each coefficient of a fixed order in the $z_i$'s on the left-hand side of $\eqref{eq:F03}$ is a polynomial in (finitely many) coefficients of $G$ and $S$. The same is true for the right-hand side, by direct inspection. This observation implies that the expression $\eqref{eq:F03}$, which we proved for polynomial $G$ and $S$,  is in fact true for the case of arbitrary power series $G$ and $S$, for example in the case of $S(z)=z$ and $G(z)=e^z$ corresponding to classical Hurwitz numbers -- for which the topological recursion is already known~\cite{BM,BEMS, EMS}.
Developing further such ``projective limits'' arguments would lead too far from our main subject and we leave to the reader the task of examining special cases or finding general assumptions under which they hold in full generality.

In ongoing work, we plan to study further consequences of the topological recursion relations for weighted Hurwitz numbers, extend the class
of multiparametric weights beyond the case of polynomial weight generating functions and derive certain explicit ELSV-like formulae
for the general weighted case.

 \bigskip
\noindent 
\small{ {\it Acknowledgements.} 
The work of A.A. was supported by IBS-R003-D1,  by RFBR grant 18-01-00926 and by European Research Council (QUASIFT grant agreement 677368).  B.E. was  supported by the ERC Starting Grant no. 335739, Quantum fields and knot homologies, funded by the European Research Council under the European Unions Seventh Framework Programme, and also partly supported by the ANR grant Quantact : ANR-16-CE40-0017.
G.C. acknowledges support from the Agence Nationale de la Recherche, grant ANR 12-JS02-001-01 ``Cartaplus'' and from the City of Paris, grant ``\'Emergences 2013, Combinatoire \`a Paris''. This project has received funding from the European Research Council (ERC) under the European Union’s Horizon 2020 research and innovation programme (grant agreement No. ERC-2016-STG 716083, ``CombiTop'').
The work of J.H. was partially supported by the Natural Sciences and Engineering Research Council of Canada (NSERC) and the Fonds de recherche du Qu\'ebec, Nature et technologies (FRQNT).   A.A. and J.H. wish to thank the Institut des Hautes \'Etudes Scientifiques for their kind hospitality for extended periods in 2017 - 2018,
 when much of this work was completed. B.E. wishes to thank the Centre de recherches math\'ematiques, Montr\'eal, for the Aisenstadt Chair grant, 
and the FQRNT grant from the Qu\'ebec government, that partially supported this joint project.
The authors would all like to thank the organizers of the January - March, 2017  thematic semester ``Combinatorics and interactions''  
at the Institut Henri Poincar\'e, where they were participants during the completion of this work. A.A. would also like to thank S. Shadrin for useful discussions. The authors thank an anonymous referee for a very careful reading of the paper and for  suggestions.}

\bigskip
\bigskip 

\break
%%%%%%%%%%%%%%%% Appendices%%%%%%%%%%%%%%%%

%\part*{Appendices }
%\addcontentsline{toc}{part}{Appendices}

%%%%%%%%%%%%%%%%%%%%%%%%%%%%%%%%%%%%%%%%%%%
%%%%%%%% Appendices %%%%%%%%%
%%%%%%%%%%%%%%%%%%%%%%%%%%%%%%%%%%%%%%%%%%%
\begin{appendix}

%%%%%%%%%%%%%% Appendix %%%%%%%%%%%%
\section{Appendices: Proofs}
\label{app_A}

In these appendices we provide all proofs that were omitted in the body of the paper.

%%%%%%%%%%%%%% Appendix A.1 %%%%%%%%%%%%
\subsection{Section~\ref{sec:fermionicAndBosonic}}
\label{app_A1}
\begin{proof}[Proof of Proposition~\ref{prop:detConnected}]
We apply Proposition~\ref{prop:Wndeterminant}. For this we  need the expansions of the quantities appearing in~\eqref{eq:Wndeterminant}:
\be\label{eq:devK}
K(x,x-\epsilon)=\frac{1}{\epsilon}
+\OO(1)
\ee
where $\OO(1)$ means that coefficients are finite at $\epsilon=0$ (order by order in $\gamma$). Moreover 
\bea\label{eq:devCauchy}
	1/\det\left(\frac{1}{x_i-x_j+\epsilon_j}\right)_{1\leq i,j\leq n} 
	&=&\displaystyle \epsilon_1\dots \epsilon_n\prod_{i<j} \frac{(x_i-x_j+\epsilon_j)(x_j-x_i+\epsilon_i)}{(x_i-x_j)(x_j-x_i+\epsilon_i-\epsilon_j)}\\
	&=& \epsilon_1\dots\epsilon_n \left(1+\OO(1) \right).\label{eq:devCauchy2}
\eea

Now, applying Proposition~\ref{prop:Wndeterminant} we get, for $n=1$, using \eqref{eq:devK}:
\be \label{eq:W1int}
W_1(x)=
[\epsilon_1] K(x,x-\epsilon_1) \cdot \epsilon_1
=\lim_{\epsilon_1\rightarrow 0} \left(K(x,x-\epsilon_1)-\frac{1}{\epsilon_1} \right)
\ee 
and  \eqref{eq:detConnected1} follows, using the equality $\tilde W_1(x)=W_1(x)$.

%Similarly, using Proposition~\ref{prop:Wndeterminant} with $n=2$, expanding the determinant, and using~\eqref{eq:devCauchy}we get: 
%\bea
%W_2(x_1,x_2)&=&
% [\epsilon_1 \epsilon_2]
%\epsilon_1\epsilon_2\left(1-\frac{\epsilon_1\epsilon_2}{(x_1-x_2)^2}+\dots \right)\\
%&&\quad\times\left(K(x_1,x_1-\epsilon_1) K(x_2,x_2-\epsilon_2)- K(x_1,x_2-\epsilon_2)K(x_2,x_1-\epsilon_1)\right)\\
%&=&
%W_1(x_1)W_1(x_2) - K(x_1,x_2)K(x_2,x_1)-\frac{1}{(x_2-x_1)^2},
%\eea
%where in the first term we have used the first equality in~\eqref{eq:W1int} twice (for $x=x_1$ and $x=x_2$), and in the rest we use~\eqref{eq:devK}. 
%It follows from the definition of connected functions that:
%\be \label{explog-k=2}
%W_2(x_1,x_2)=W_1(x_1)W_1(x_2)+\tilde W_2(x_1,x_2),
%\ee
%and the second claim \eqref{eq:detConnected2} follows by comparing the last two equations.

For higher values of $n$, we use Proposition~\ref{prop:Wndeterminant} and expand the determinant. 
	\bea
	W_n(x_1,\dots,x_n)&=& [\epsilon_1\dots\epsilon_n] \det (K(x_i,x_j-\epsilon_j) /\det\left(\frac{1}{x_i-x_j+\epsilon_j}\right)\\
	&=&\sum_{\sigma\in\mathfrak{S}_n}\sgn(\sigma)
	 [\epsilon_1\dots\epsilon_n] \left(\prod_i K(x_{\sigma(i)},x_{i}-\epsilon_{i})\right) /\det\left(\frac{1}{x_i-x_j+\epsilon_j}\right). \cr
	 &&
	 \label{eq:proof1cycleinter}
\eea
	Note that $K(x_{\sigma(i)},x_{i}-\epsilon_i)$ is regular at $\epsilon_{\sigma(i)}=0$ if $\sigma(i)\neq i$, and that for $\sigma(i)=i$,
	 it has a pole of order $1$ given by:
	\be \label{eq:devKW1}
	K(x_i,x_{i}-\epsilon_i) = \frac{1}{\epsilon_i} + \tilde W_1(x_i) + O(\epsilon_i),
	\ee
	where again $\OO$ is understood coefficient by coefficient in $\gamma$.
	 Since in \eqref{eq:devCauchy2} all $\epsilon_i$'s appear to  power at least one, we deduce that in order to extract the coefficient of $[\epsilon_1\dots\epsilon_n]$ in~\eqref{eq:proof1cycleinter}, we have to choose, for each fixed point of $\sigma$,  either the first or second term in \eqref{eq:devKW1}. Letting $I_\sigma$ be the set of fixed points of $\sigma$ for which we choose the first term, and $J_\sigma$ the  complementary set of fixed points, we can thus write 
\bea
	\eqref{eq:proof1cycleinter} &=& 
	\sum_{\sigma\in\mathfrak{S}_n}\sgn(\sigma)\!\!\!
	\sum_{I_\sigma \uplus J_\sigma=\atop {\{\mbox{\tiny fixed points of }\sigma\}}}
	 \prod_{i\in J_\sigma} \tilde W_1(x_i) \prod_{i\in (I_\sigma\uplus J_\sigma)^c}  K(x_i,x_{\sigma_i})
	  \left [\prod_{i\in I_\sigma}\epsilon_i^2 \prod_{i\in I_\sigma^c}\epsilon_i\right] 1/\det\left(\frac{1}{x_i-x_j+\epsilon_j}\right), \cr
	  &&
	  \label{eq:proof1cycleinter2}
	   \eea
where, as before, the square brackets denote a coefficient extraction.

Now observe that
\be
\frac{(x_i-x_j+\epsilon_j)(x_j-x_i+\epsilon_i)}{(x_i-x_j)(x_j-x_i+\epsilon_i-\epsilon_j)}
=1-\frac{\epsilon_i\epsilon_j}{(x_i-x_j)^2} + O(\epsilon_i^2\epsilon_j)+ O(\epsilon_i\epsilon_j^2),
\ee
which implies, from~\eqref{eq:devCauchy}, that
\be
\left[\prod_{i\in I_\sigma}\epsilon_i^2 \prod_{i\in I_\sigma^c}\epsilon_i\right] {1\over \det\left(\frac{1}{x_i-x_j+\epsilon_j}\right)}
=\sum_{\pi_\sigma}
	\prod_{\{i,j\}\in\pi_\sigma}\frac{-1}{(x_i-x_j)^2},
\ee
where the sum is taken over all pairings $\pi_\sigma$ of $I_\sigma$, and the first brackets again denote a coefficient extraction. Given a pairing $\pi_\sigma$, we can modify the permutation $\sigma$ by transforming each pair of fixed points of $\pi_\sigma$ into a $2$-cycle. The permutation $\tilde \sigma$ thus created now has two types of $2$-cycles $(a,b)$ that come either with a weight $-K(x_a,x_b)K(x_b,x_a)$ (the original cycles of $\sigma$; note the sign coming from the contribution of this $2$-cycle to the signature of $\sigma$) or  a weight $\frac{-1}{(x_a-x_b)^2}$ (the cycles coming from $\pi$). 
Equivalently, each $2$-cycle $(a,b)$ of $\tilde{\sigma}$ carries a weight which is the sum of these two,
\be
-K(x_a,x_b)K(x_b,x_a) - \frac{1}{(x_a-x_b)^2}. 
\ee
We can therefore rewrite~\eqref{eq:proof1cycleinter2} as a sum over the permutations $\tilde{\sigma}$ with these modified weights, and it will be convenient to do it by summing instead over the cycle decomposition of $\tilde{\sigma}$. Namely, writing $\{\CC_1, \dots,\CC_\ell\}$ for the cycle decomposition of $\tilde\sigma$, and writing $I_i$ for the support of the cycle $\CC_i$, we can rewrite~\eqref{eq:proof1cycleinter2} as:
\bea\label{eq:cumulant1}
W_n(x_1,\dots,x_n)&=&
\sum_{\ell\geq 1} \sum_{I_1\uplus\dots\uplus I_\ell = \{1,\dots,n\}}
\prod_{i=1}^\ell \sum_{\CC_i \in S_{I_i}\atop \mbox{\tiny with 1-cycle}} L_{\CC_i}(x_j,j\in I_i),
\eea
where
%the right-hand side of \eqref{eq:detConnected1}, \eqref{eq:detConnected2}, respectively
%for $| I_i |=1,2$, and where 
\be
L_{\CC_i}(x_j,j\in I_i) =
\left\{\begin{array}{cl}
	\tilde{W}_1(x_j) & \mbox{ if } |I_i|=1 \text{ and } I_i=\{j\};
	\\
	-K(x_a,x_b)K(x_b,x_a) - \frac{1}{(x_a-x_b)^2} & \mbox{ if }|I_i|=2 \text{ and } I_i=\{a,b\};
	\\
	(-1)^{|I_i |-1}\prod_{j \in I_i} K(x_j, x_{\CC_i(j)}) & \mbox{ if }| I_i |\geq 3.
\end{array}\right.
\ee
Formulae (\ref{eq:detConnected1}), (\ref{eq:detConnected2}) and (\ref{eq:detConnectedn}) of Proposition~\ref{prop:detConnected}
 follow by comparing~\eqref{eq:cumulant1} with the expansion~\eqref{eq:cumulant2} of the nonconnected  correlator in terms of the connected ones and using induction on $n$.
\end{proof}

%%%%%%%%%%%%%% Appendix A2 %%%%%%%%%%%%
\subsection{Section~\ref{sec:recursionOperators}}
\label{app_A2}

\begin{proof}[Proof of Lemma~\ref{lemmarec}]

$T(x)$ is defined by $T(0)=0$ and
\bea
T(x)-T(x-1) 
&=& \log \gamma G(\beta x) \cr
&=& \log \gamma + \sum_{i=1}^M \log(1+c_i \beta x) \cr
&=& \log \gamma +  \sum_{k=1}^\infty {(-1)^{k-1}} A_k \beta^k x^k
\eea
where $A_k = \frac{1}{k}\sum_i c_i^k$.
The $k^{\rm th}$ Bernoulli polynomial $B_k(x)$ satisfies 
\be\label{Bernoulpr}
B_k(x+1)-B_k(x)= k x^{k-1}\quad 
\text{ and} \quad B_k(1-x)=(-1)^k B_k(x) . 
\ee
It can be expressed as
\be\label{Bernexp}
B_k(x) = (-1)^k B_k + \sum_{m=0}^{k-1} \frac{ k!  B_m}{(k-m)! m!} \ x^{k-m},
\ee
	where $B_m=B_m(1)=(-1)^m  B_m(0)$ are the Bernoulli numbers.
We thus get
\be
T(x) = x\log\gamma  + \sum_{k=1}^\infty (-1)^k A_k \beta^k \frac{B_{k+1}(x+1)-B_{k+1}}{k+1} . 
\ee
Then $T(x)-x\log \gamma\in \Kb[x][[\beta]]$, and (\ref{TGrel}) is satisfied, therefore
\be
T(0)=0, \quad T(i) = T_i.
\label{T_interpol}
\ee 

Note  that the series $T(x)$ is unique up to a linear term of the form $2\pi i k x$, $k\in  \Zb$ or any periodic function $f(x)$ such that $f(x)=f(x+1)$. However, if we require the coefficient of each power of $\beta$ to be a polynomial in $x$,
 we  can only add a linear term, and if we require the interpolating property (\ref{T_interpol}), this must vanish.

\end{proof}

\begin{proof}[Proof of Theorem~\ref{thm:CD}]
%	A proof was already given in~\cite[Proposition 4.5]{ACEH2}, but the following alternative proof is self-contained.
	From Proposition~\ref{Proppsiphi} we know that 
\bea
\tau\left([x]-[x'],\beta^{-1}{\bf s}\right)&=&(x-x')K(x,x')\cr
&=&(x-x')e^{T\left(D\right)-T\left(-D'-1\right)}\frac{e^{\beta^{-1}\xi({\bf s},x) - \beta^{-1}\xi ({\bf s},x')}}{x-x'}\cr
&=&\left(e^{T\left(D-1\right)-T\left(-D'-1\right)}x-e^{T\left(D\right)-T\left(-D'\right)}x'\right)\frac{e^{\beta^{-1}\xi({\bf s},x) - \beta^{-1}\xi ({\bf s},x')}}{x-x'}\cr
&=&\gamma e^{T\left(D-1\right)-T\left(-D'\right)}\left(G\left(-\beta D'\right)x-G\left(\beta D\right)x'\right) \frac{e^{\beta^{-1}\xi({\bf s},x) - \beta^{-1}\xi ({\bf s},x')}}{x-x'}\cr
&&
\eea
Thus
\bea\label{tauR}
\tau\left([x]-[x'],\beta^{-1}{\bf s}\right) 
 =\gamma e^{T\left(D-1\right)-T\left(-D'\right)}\left(e^{\beta^{-1}\xi({\bf s},x) - \beta^{-1}\xi ({\bf s},x')}A(x,x')\right).
\eea
%(Just because $\tau\left([x]-[x'],{\bf s}\right)$ is a formal series and $e^{T\left(D-1\right)-T\left(-D'\right)}$ is invertible). 
From (\ref{Tcom}) and (\ref{Tcom1}) it follows that
\bea\label{tauTop}
\tau\left([x]-[x'],\beta^{-1}{\bf s}\right)&=&\gamma\, A(R_+,R_-')e^{T\left(D-1\right)-T\left(-D'\right)}\left( e^{\beta^{-1}\xi({\bf s},x) - \beta^{-1}\xi ({\bf s},x')}\right)\cr
&=&A(R_+,R_-') \Psi^+_0(x) \Psi^-_0(x').\hfill \qedhere
\eea
\end{proof}

%%%%%%%%%%%%%% Appendix A.3 %%%%%%%%%%%%

\subsection{Section~\ref{section5}}
\label{app_A3}

\begin{proof}[Proof of Theorem~\ref{thm:infiniteSystem}]
Detailed proofs of this theorem are given in
the companion paper ref. \cite{ACEH2}. 
In the following we give an alternative proof and construct the four matrices explicitly.
For polynomial $S$ and $G$ these matrices are finite-band. From Theorem \ref{thm:kspectral0}
and the commutation relation (\ref{Rcom}) we conclude that for any series 
\be
f(x)= \sum_{i=0}^\infty f_i x^i,
\ee
with constant coefficients, we have
\bea\label{Darbpsi}
\pm\beta D \,f(R_\pm) \Psi^\pm_0 = \left(f(R_\pm)S(R_\pm)\pm\beta R_\pm f'(R_\pm)\right)\Psi^\pm_0.
%-\beta D \,f(R_-) \Psi^-+0 = \left(f(R_-)S(R_-)-\beta R_-f'(R_-)\right)\Psi^+_0.
\eea
Thus, for any function $g$
\bea\label{gf}
g(\pm\beta D)f(R_\pm)\Psi^\pm_0=p_\pm(R_\pm)\Psi^\pm_0,
%g(-\beta D)f(R_-)\Psi^-_0=p_-(R_-)\Psi^-_0,
\eea
where
\be\label{defpm}
p_\pm(r):=g\left(\Delta_\pm(r)\right) f(r).
\ee

In particular, if $g(r)=r$ and $f(r)=r^k$, we have
\bea\label{Dlong}
\pm \beta D {\Psi}^\pm_k=P_\pm^{(k)}(R_\pm)\Psi^\pm_0,
%-\beta D {\Psi}^-_k=P_-^{(k)}(R_-)\Psi^-_0,
\eea
where
\be
P_\pm^{(k)}(r):=\Delta_{\pm} r^k.
\ee
We have thus constructed the matrices $P^{\pm}$ satisfying \eqref{eq:Ppm}. 
These matrices,
\bea
P^{\pm}=\bordermatrix{
~&~&\cdots&-1&0&1&\cdots\nonumber 
\\ 
~&\vdots&\ddots&\vdots&\vdots&\vdots&\vdots&\iddots\nonumber
\\
\vdots&\cdots&\mp2\beta&s_1&2s_2&3s_3&4s_4&\cdots\nonumber
\\
-1&\cdots&0&\mp\beta&s_1&2s_2&3s_3&\cdots\nonumber
\\
0&\cdots&0&0&0&s_1&2s_2&\cdots\nonumber
\\
1&\cdots&0&0&0&\pm\beta&s_1&\cdots\nonumber
\\
\vdots&\cdots&0&0&0&0&\pm2\beta&\cdots\nonumber
\\
~&\iddots&\vdots&\vdots&\vdots&\vdots&\vdots&\ddots\nonumber
\\
}
\eea
are upper-triangular with $L+1$ bands.

Moreover, we have
\bea
G(\pm \beta D){\Psi}^\pm _k =  Q^{(k)}_\pm (R_\pm ) \Psi^\pm_0,
%G(-\beta D){\Psi}^-_k =  Q^{(k)}_-(R_-) \Psi^-_0,
\eea
where
\be
Q^{(k)}_\pm(r):=V_\pm(r) r^k.
\ee
From the definition of $\Psi^+_k$ and $\Psi^-_k$, for any $k\in  \Zb$, we have
\bea\label{Qoper}
\frac{1}{\gamma x}{\Psi}^\pm_k= Q^{(k-1)}_\pm(R_\pm){\Psi}^\pm_0.
%\frac{1}{\gamma x}{\Psi}^-_k= Q^{(k-1)}_-(R_-){\Psi}^-_0.
\eea
The last two equations give a construction of the matrices ${Q}^{\pm}$ satisfying \eqref{eq:Qpm}.
Note that since the polynomials $Q^{(k-1)}_\pm(r)$ involve monomials of degrees from $k-1$ to $k-1+LM$, the matrices 
${Q}^{\pm}$ have bands of width $LM$, the nonzero diagonal being just below the principal one.

To compute the first few bands explicitly, note that since 
\bea
Q_\pm^{(k)}(r)=r^k\left(G(\pm \beta k)\pm r\frac{s_1}{\beta}\left(G(\pm\beta(k+2)-G(\pm\beta(k))\right)+O(r^2)\right),
\eea

\bea
Q^{\pm}=\bordermatrix{
~&~&\cdots&-1&0&1&\cdots\nonumber 
\\ 
~&\vdots&\ddots&\vdots&\vdots&\vdots&\vdots&\iddots\nonumber
\\
\vdots&\cdots&a_\pm^{(0)}(-2)&a_\pm^{(1)}(-2)&a_\pm^{(2)}(-2)&a_\pm^{(2)}(-2)&a_\pm^{(2)}(-2)&\cdots\nonumber
\\
-1&\cdots&a_\pm^{(-1)}(-1)&a_\pm^{(0)}(-1)&a_\pm^{(1)}(-1)&a_\pm^{(2)}(-1)&a_\pm^{(3)}(-1)&\cdots\nonumber
\\
0&\cdots&0&a_\pm^{(-1)}(0)&a_\pm^{(0)}(0)&a_\pm^{(1)}(0)&a_\pm^{(2)}(0)&\cdots\nonumber
\\
1&\cdots&0&0&a_\pm^{(-1)}(1)&a_\pm^{(0)}(1)&a_\pm^{(1)}(1)&\cdots\nonumber
\\
\vdots&\cdots&0&0&0&a_\pm^{(-1)}(2)&a_\pm^{(0)}(2)&\cdots\nonumber
\\
~&\iddots&\vdots&\vdots&\vdots&\vdots&\vdots&\ddots\nonumber
\\
}
\eea
where
\bea
a_\pm^{(-1)}(k)&=&G(\pm \beta (k-1)),\\
a_\pm^{(0)}(k)&=& \pm\frac{s_1}{\beta}\left(G(\pm\beta (k+1))-G(\pm\beta(k-1))\right),\\
a_\pm^{(1)}(k)&=&\pm\frac{s_2}{\beta}\left(G(\pm\beta (k+1)-G(\pm\beta(k-1))\right)\\&&+\frac{s_1^2}{2\beta^2}\left(G(\pm\beta (k+1)-2G(\pm\beta k)+G(\pm\beta(k-1))\right).
\eea
\end{proof}

%\begin{proof}[Proof of Proposition~\ref{prop:commutePQ}]
% For simplicity we consider only one relation; the other is derived similarly. For the generating functions 
%\bea
%{\mathcal P}^{\pm}(t,r)=\sum_{i,j=-\infty}^\infty P_{ij}^{\pm}t^{i}r^{j}=\Delta_\pm(r)\cdot\sum_{i=-\infty}^{\infty} (tr)^i,\\
%{\mathcal Q}^{\pm}(t,r)=\sum_{i,j=-\infty}^\infty Q_{ij}^{\pm}t^{i}r^{j}=V_\pm(r) r^{-1} \cdot\sum_{i=-\infty}^{\infty} (tr)^i.
%\eea
%Thus, the generating function for the commutator of $P^\pm$ and $Q^\pm$ is
%\bea
%\frac{1}{2\pi i} \oint ({\mathcal P}^{+}(t,q){\mathcal Q}^{+}(q^{-1},r)-{\mathcal Q}^{+}(t,q) {\mathcal P}^{+}(q^{-1},r))\frac{d q}{q}\\
%=\left[V_+(r) r^{-1},\delta_+(r)\right]\sum_{i=-\infty}^{\infty} (tr)^i\\
%=-\beta V_+(r) r^{-1}\sum_{i=-\infty}^{\infty} (tr)^i\\
%=-\beta {\mathcal Q}^{+}(t,r).
%\eea
%\end{proof}

\begin{proof}[Proof of Theorem~\ref{prop:finiteSystem}]
From Theorem \ref{thm:CD} and  (\ref{Darbpsi}) it follows that 
\bea\label{DtoA}
\beta D\, \tau\left([x]-[x'],\beta^{-1}{\bf s}\right)=A^+(R_+,R_-') \Psi^+_0(x)\Psi^-_0(x'),\\
-\beta D'\, \tau\left([x]-[x'],\beta^{-1}{\bf s}\right)=A^-(R_+,R_-') \Psi^+_0(x)\Psi^-_0(x'),
\eea
where
\bea
A^+(r,t)=\Delta_+(r)A(r,t),\\
A^-(r,t)=\Delta_-(t)A(r,t).
\eea
It is clear that $A^{\pm}(r,t)$ are polynomials in $r$ and $t$; $A^{+}(r,t)$ is a polynomial of degree $LM+L-1$ in the variable $r$ and degree $LM-1$ in the variable $t$, while $A^{-}(r,t)$ is a polynomial  of degree $LM-1$ in the variable $r$ and
degree $LM+L-1$ in the variable $t$.
	Let $\mathbf{B}$ be the symmetric $LM \times LM$ matrix  with entries 
\be
	\mathbf{B} = \begin{pmatrix}
 0& 0& 0 & 0& \dots \cr
0 & s_1  & 2s_2 & 3s_3  & \cdots\cr
0 & 2s_2 &3s_3 & 4s_4 & \cdots &\cr
0 & 3s_3 &4s_4  & 5s_5&\cdots  \cr
\vdots & \vdots&\vdots &\vdots& \ddots \cr
\end{pmatrix}
\ee
that are the coefficients of the polynomial
\be
	B(r,t)=rt  \frac{S(r)-S(t)}{r-t} = \sum_{i=0}^{LM-1} \sum_{j=0}^{LM-1} \mathbf{B}_{ij} r^i t^j. 
\ee
We therefore have
\be
	\vec{{\Psi}}^-(x')^t \mathbf{B} \vec{{\Psi}}^+(x)= R_+ R_-' \frac{S(R_+)-S(R_-')}{R_+-R_-'}  \left( \Psi^+_0(x)\Psi^-_0(x')\right).
\ee
From (\ref{Qoper}) it follows that
\bea\label{Btildop}
\frac{1}{\gamma x} \vec{{\Psi}}^-(x')^T \mathbf{B} \vec{{\Psi}}^+(x)= {\mathbf{B}}^+ (R_+,R_-')\left( \Psi^+_0(x)\Psi^-_0(x')\right),\\
\frac{1}{\gamma x'} \vec{{\Psi}}^-(x')^T \mathbf{B} \vec{{\Psi}}^+(x)= {\mathbf{B}}^- (R_+,R_-') \left(\Psi^+_0(x)\Psi^-_0(x')\right),
\eea
where
\bea
{\mathbf{B}}^+ (r,t) :=t V_+(r) \left(\frac{S(r)-S(t)}{r-t} \right),\\
{\mathbf{B}}^- (r,t) :=r V_-(t) \left(\frac{S(r)-S(t)}{r-t}\right).
\eea
Combining (\ref{DtoA}) and (\ref{Btildop}) we get
\bea\label{CBeq}
\beta D\, \tau\left([x]-[x'],\beta^{-1}{\bf s}\right)=C^+(R_+,R_-')\left(\Psi^+_0(x)\Psi^-_0(x')\right)
-\frac{1}{\gamma x}\left( \vec{{\Psi}}^-(x')^T \mathbf{B} \vec{{\Psi}}^+(x)\right),\\
-\beta D'\, \tau\left([x]-[x'],\beta^{-1}{\bf s}\right)=C^-(R_+,R_-') \left(\Psi^+_0(x)\Psi^-_0(x')\right) 
-\frac{1}{\gamma x'} \left(\vec{{\Psi}}^-(x')^T \mathbf{B} \vec{{\Psi}}^+(x)\right), \label{CBeq2}
\eea
where
\be
C^\pm(r,t) =A^{\pm}(r,t)+B^{\pm}(r,t).
\ee
It is easy to see that $C^+(r,t)$, $C^-(r,t)$ are equal polynomials of degree at most  $LM-1$ in each of the variables $r$ and $t$:
\be
C(r,t) := C^+(r,t)=C^-(r,t)=\left(\Delta_+(r)rV_-(t)-\Delta_-(t)tV_+(r)\right)\left(\frac{1}{r-t}\right).
\ee
We can therefore rewrite (\ref{CBeq}) and (\ref{CBeq2})  as 
\bea\label{Deqmat}
\beta D\, \tau\left([x]-[x'],{\bf s}\right)= \vec{{\Psi}}^+(x)^T\left(\mathbf{C}-\frac{1}{\gamma x} \mathbf{B}\right) \vec{{\Psi}}^-(x'),\\
-\beta D'\, \tau\left([x]-[x'],{\bf s}\right)= \vec{{\Psi}}^+(x)^T \left(\mathbf{C}-\frac{1}{\gamma x'} \mathbf{B}\right) \vec{{\Psi}}^-(x'). \label{Deqmat2}
\eea

Using the Christoffel-Darboux relation we can rewrite (\ref{Deqmat2}) as
\be\label{Phieqz}
\vec{{\Psi}}^+(x)^T\left(\beta \mathbf{A} D' +\mathbf{C}-\frac{1}{\gamma x'} \mathbf{B}\right) \vec{{\Psi}}^-(x')=0,
\ee
where
\be
\vec{x}:=\begin{pmatrix}
1\\
x\\
x^2\\
\dots\\
x^{ML-1}\\
\end{pmatrix},
\ee
which can equivalently be rewritten
\be
W \left(\vec{x}^T \left(\beta \mathbf{A} D' +\mathbf{C}+\frac{1}{\gamma x'} \mathbf{B}\right) \vec{{\Psi}}^-(x')\right)=0,
\ee
where
\be
W:=\gamma e^{T(D-1)}\left(e^{\beta^{-1}\xi({\bf s},x)}\right).
\ee
Since the operator $W$ is invertible,  eq.~(\ref{Phieqz}) is equivalent to
\be
\vec{x}^T \left(\beta \mathbf{A} D' +\mathbf{C}-\frac{1}{\gamma x'} \mathbf{B}\right) \vec{{\Psi}}^-(x')=0.
\ee
This is a polynomial in $x$,  which must be identically $0$ since all its coefficients vanish.
Therefore
\be
\left(\beta \mathbf{A} D' +\mathbf{C}-\frac{1}{\gamma x'} \mathbf{B}\right) \vec{{\Psi}}^-(x')=\vec{0},
\ee
or
\be
-\beta D \vec{{\Psi}}^-(x) =\mathbf{A}^{-1}\left(\mathbf{C}-\frac{1}{\gamma x} \mathbf{B}\right) \vec{{\Psi}}^-(x).
\ee
It follows similarly that
\be
\beta D  \vec{{\Psi}}^+(x)^T  = \vec{{\Psi}}^+(x)^t \left(\mathbf{C}-\frac{1}{\gamma x} \mathbf{B}\right) \mathbf{A}^{-1}.
\ee
\end{proof}

%%%%%%%%%%%%%%%%%%%%%%%%%%%  Appendix A.4 %%%%%%%%%%%%%%%%

\subsection{Section~\ref{sec:classicalCurve}}
\label{app_A4}

\begin{proof}[Proof of Proposition~\ref{prop:eqY}]
We give a combinatorial proof on constellations.  By definition, for each $i=1, \dots M$,   
\be
x y(x)=x\frac{d}{dx}\tilde F_{0,1}(x)
\ee
 is the generating function of constellations of genus~$0$ with a unique vertex of colour $0$, and a marked corner of colour $i$. Indeed a constellation of size $n\geq 1$ has $n$ such corners, and  such an object is counted with a weight $x^n$ in $\tilde F_{1,0}(x)$, and therefore
	a weight $n x^n$ in $xy(x)$. By deleting this vertex and edges incident to it, we obtain an object which we call a \emph{preconstellation} in this proof. A preconstellation is a genus~$0$ graph with one face, hence a tree. Such objects can be decomposed recursively to obtain polynomial equations for their generating functions as we now show.

We first claim that the generating function of  preconstellations carrying a marked vertex $v_0$  of colour $\infty$ and degree $1$ is given by 
\be
\gamma x G(x y(x))=\gamma x\prod_{i=1}^M (1+c_i x y(x)).
\ee
 To see this, note that the unique star vertex adjacent to $v_0$ is incident to $M$ edges in addition to the one linking it to $v_0$. For $i=1, \dots, M$, the $i$-th edge is attached either to a single vertex of colour $i$ (contribution to the generating function:~$1$) or to a nontrivial preconstellation marked at a corner of colour $i$ (contribution to the generating function:~$c_i xy(x)$). The factor $\gamma x$ takes into account the contribution of the central star vertex, and the claim follows.

Now $xy(x)$ is also the generating function of preconstellations with a marked corner of colour $\infty$. Given such an object, call $v$ the vertex incident to the marked corner and $k$ its degree. By exploding $v$ into $k$ vertices, we obtain $k$ preconstellations, each of them carrying a marked vertex of colour $\infty$ and degree $1$. Recalling that the vertex $v$ comes equipped,
 by definition, with a weight $k s_k$, the total contribution is thus, from the previous claim:
$$
xy(x)=\sum_{k\geq 1} k s_k \left(\gamma x\prod_{i=1}^M (1+c_i x y(x))\right)^k
=S\left( \gamma x G(x y(x)) \right).
$$
Finally, the uniqueness of the solution in $\gamma\Kb[x,s][[\gamma]]$ is clear, since coefficients of $y(x)$ can be computed by induction from~\eqref{eq:TutteY}.
\end{proof}

%%%%%%%%%%%%%%%%%%%%%%%%%%%  Appendix A.5 %%%%%%%%%%%%%%%%

\subsection{Section~\ref{sec:WKB}}
\label{app_A5}

\begin{proof}[Proof of Lemma~\ref{lemma:PsiPhiStruct}]
	We first observe that in~\eqref{eq:PsiRefined}, the only value of $(g,n)$ for which the exponent $2g-2+n$ can be negative is $(g,n)=(0,1)$, and in this case $2g-2+n=-1$.
Therefore the function
\be \label{eq:PsiRegular}
\Psi^+_0(x) \exp\left(-\beta^{-1}\tilde F_{0,1}(x)\right) = \exp\left(\sum_{g\geq 0,n\geq 1, \atop (g,n)\neq (0,1)} \frac{\beta^{2g-2+n}}{n!} \tilde F_{g,n}( x,\dots,x) \right)
\ee
involves no negative powers of $\beta$. Now observe that, from the definitions, we have 
\be \int_0^x y(u) du = \tilde F_{0,1}(x)\,,
\ee so the statement is proved for $\check \Psi^+_0(x)$. 
From
\be
\beta x \frac{d}{dx}  \exp\left(-\beta^{-1}\int_0^x y(u)du\right) = - x y(x) \exp\left(-\beta^{-1}\int_0^x y(u)du\right),
\ee
we see that for any function $F(x)\in \Kb[x,{\bf s},\beta,\beta^{-1}]$, if $F(x)\exp\left(-\beta^{-1}\int_0^x y(u)du\right)$ has no negative powers of $\beta$, then the same is true for $\beta x \frac{d}{dx} F(x)$. Given the recursive relation  \eqref{Psi+_k_rec} 
for $ \Psi^+_i(x)$ in terms of  $\Psi_0^+(x)$, we obtain the statement for $\check \Psi^+_i(x)$ for any $i\in  \Zb$. 
The corresponding statement for the functions $\check\Psi^-_i(x)$ follows by replacing $\beta$ by $-\beta$.
\end{proof}

\begin{proof}[Proof of Lemma~\ref{lemma:actionsOnCheckedPsik}]
It  follows from Proposition \ref{Proppsiphik} that the functions $\Psi^\pm_k(x)$ obey equations (\ref{betadif}), (\ref{betadif1}),
%\bea
%\beta \frac{d}{d\beta} \Psi^+_k(x)=\left( \beta \frac{d}{d\beta} T(D-1)-\frac{1}{\beta}\xi(R_+,s)\right) \Psi^+_k(x),\\
%\beta \frac{d}{d\beta} \Psi^-_k(x)=\left( - \beta \frac{d}{d\beta} T(-D)+\frac{1}{\beta}\xi(R_-,s)\right ) \Psi^-_k(x),
%\eea
where we recall that $T(x)$ is defined by (\ref{Tdef}). Then, from (\ref{Bernoulpr}) and (\ref {Bernexp}) we get
\bea
T(x-1)  
&=&(x-1)\log\gamma  +  \sum_{k=1}^\infty A_k \beta^k \frac{B_{k+1}(-x)-B_{k+1}}{k+1} \cr
&=&(x-1)\log\gamma  +  \sum_{k=1}^\infty \sum_{m=0}^{k} \frac{B_m}{m!} A_k \beta^k \frac{k! (-x)^{k+1-m}}{(k+1-m)!}  \cr
&=& (x-1)\log\gamma + \sum_{m=0}^\infty \beta^{m-1} (-1)^m \frac{B_m}{m!} \tilde G^{(m-1)}(\beta x). 
\eea
%%%%%%%%%%%%%%%%%%%
Applying $\beta d/d\beta$ gives
\bea
\beta \frac{d}{d\beta} T(x-1) 
&=& \sum_{m=0}^\infty \beta^{m-1} (-1)^m \frac{B_m}{m!} \left( \beta x \tilde G^{(m)}(\beta x) + (m-1) \tilde G^{(m-1)}(\beta x) \right).
\eea
The equation for $\Psi_k^-$ is obtained by replacing $\beta\to -\beta$.

\end{proof}

\begin{proof}[Proof of Lemma~\ref{lemma:operatorStructure}]
Since the pair ($O^+$, $O^-$) is related by the replacement $\beta\to -\beta$, it is sufficient to prove
it for the case $O^+$. Using the parametrization $x=X(z)$, $y(x)=Y(z)$, the relations
\be
xy(x)=S(z)
\ee
and 
\be
D = x\frac{d}{dx} = \frac{X(z)}{X'(z)} \ \frac{d}{dz} = \frac{G(S(z))}{\phi(z)} \ z\frac{ d}{dz},
\ee
we have
\be
O^+ = \frac{z}{G(S(z))} \sum_{l=0}^M g_l (S(z) + \beta D)^l,
\ee
and thus $O_0=z$.
For $k\geq 1$ we have
\be
O_k = \frac{z}{\phi(z)} \ \sum_{l=k}^M g_l \sum_{j_0+j_1+j_2+\dots+j_k=l-k} S(z)^{j_0} z\frac{d}{dz} S(z)^{j_1} \frac{G(S(z))}{\phi(z)}  z\frac{d}{dz}  \dots S(z)^{j_{k-1}} \frac{G(S(z))}{\phi(z)} S(z)^{j_k}.
\ee
These are clearly polynomials in  $z\frac{d}{dz}$, whose coefficients are rational functions of $z$, 
and whose denominators are a power of $\phi(z)$. We have
\bea
O_1 
&=&  \frac{z}{\phi(z)} \ \sum_{l=1}^M g_l \sum_{j_0+j_1=l-1} S(z)^{j_0} z\frac{d}{dz} S(z)^{j_1} \cr
&=&  \frac{z}{\phi(z)} \ \sum_{l=1}^M  l g_l \left(S(z)^{l-1} z\frac{d}{dz} + \frac{l-1}{2} S(z)^{l-2} z S'(z)  \right) \cr
&=&  \frac{z}{\phi(z)} \ \left(G'(S(z)) z\frac{d}{dz} + \frac12 G''(S(z)) z S'(z)  \right).
\eea

For the first term of the expansion of $U^+$ in (\ref{Uexpan}) we have
\be
U_{-1}^+ = S(z)\log G(S(z)) - \tilde G^{(-1)}(S(z)) + \int_0^z Y(z')X'(z')dz'- \xi(z).
\ee
Applying $z{d \over dz}$ gives
\bea
z \frac{d}{dz} U_{-1}^+ 
&=& \frac{S(z)S'(z)G'(S(z))}{G(S(z))} + S'(z) \log G(S(z)) - S'(z) \log G(S(z))+ Y(z)X'(z)- \frac{S(z)}{z} \cr
&=& \frac{S(z)S'(z)G'(S(z))}{G(S(z))} + \frac{S(z) \phi(z)}{z G(S(z))} - \frac{S(z)}{z} \cr
&=& \frac{S(z)}{zG(S(z))} \left( zS'(z)G'(S(z)) +\phi(z) - G(S(z))\right) = 0.
\eea
Since $U_{-1}^+$ vanishes at $z=0$  it follows  that $U_{-1}^+=0$.

We now compute $U_0^+$, using the fact that the Bernoulli polynomial has leading terms
\be
B_k(x) = x^k - \frac{k}{2} x^{k-1} + O(x^{k-2}).
\ee
\bea
U_0^+ 
&=& \sum_{j=1}^\infty  (-1)^{j-1} j A_j \left( \frac{-S(z)^j}{2} + \frac{1}{j+1} \sum_{m=0}^j S(z)^{j-m} D S(z)^m \right) -\sum_{j=1}^L s_j \sum_{m=0}^{j-1} z^{j-1-m} O_1 z^m \cr
&=& \frac{-1}{2} \frac{S(z)G'(S(z))}{G(S(z))} + \frac{G(S(z))}{\phi(z)} \sum_{j=1}^\infty  (-1)^{j-1} j A_j \left(  S(z)^j z\frac{d}{dz} + \frac{j}{2}  z S'(z) S(z)^{j-1}  \right) \cr
&& -\frac{z}{\phi(z)} \sum_{j=1}^L s_j \sum_{m=0}^{j-1} z^{j-1-m} \left( G'(S(z)) z \frac{d}{dz} + \frac12 zS'(z)G''(S(z)) \right) z^m \cr
&=& \frac{-1}{2} \frac{S(z)G'(S(z))}{G(S(z))} 
+   \frac{S(z)G'(S(z))}{\phi(z)} z\frac{d}{dz}  \cr
&& +\frac12 z S'(z) \frac{G(S(z))}{\phi(z)}  \left(   \frac{G'(S(z)) + S(z) G''(S(z))}{G(S(z))} - \frac{S(z) G'(S(z))^2}{G(S(z))^2} \right) \cr
&& -\frac{z}{\phi(z)} \sum_{j=1}^L j s_j  z^{j-1} \left(  G'(S(z)) z \frac{d}{dz} + \frac{j-1}{2} G'(S(z)) + \frac12 zS'(z)G''(S(z)) \right)  \cr
&=& \frac{S(z)G'(S(z))}{\phi(z)} z \frac{d}{dz} \cr
&& - \frac{1}{2} \frac{S(z)G'(S(z))}{\phi(z)} + \frac12 \frac{z S'(z) G'(S(z))}{\phi(z)} + \frac12 \frac{z S(z) S'(z) G''(S(z))}{\phi(z)} \cr
&& -\frac{S(z)G'(S(z))}{\phi(z)}  z \frac{d}{dz}   -\frac12 (zS'(z)-S(z)) \frac{G'(S(z))}{\phi(z)} -\frac12 \frac{zS'(z) S(z) G''(S(z))}{\phi(z)} \cr
&=&   0 ,
\eea
where we have used 
\be
\sum_{j=1}^\infty  (-1)^{j-1} j A_j x^j = \sum_{i=1}^M \frac{c_i x}{1+c_i x} = \frac{x G'(x)}{G(x)},
\ee
and hence
\be
\sum_{j=1}^\infty  (-1)^{j-1} j^2 A_j x^{j-1} =\frac{d}{dx} \sum_{i=1}^M \frac{c_i x}{1+c_i x} = \frac{G'(x)}{G(x)}+\frac{xG''(x)}{G(x)}- \frac{xG'(x)^2}{G(x)^2}.
\ee
Computations for $U_{-1}^-$ and $U_0^-$ are completely analogous.

From the definition of the $U_j^\pm$'s,  it easily follows that they are differential operators of the form
\be
U_j^\pm=\sum_{l=0}^{j+1} U_{j,l}^\pm(z) \left(z\frac{d}{dz}\right)^l,
\ee
where $\{U_{j,l}^\pm\in {\Kb}({\bf s},z)\}$ are rational functions of $z$, with poles possibly only at the zeros of $\phi(z)$ and $G(S(z))$. For large $|z|$, their leading asympotic term is
\be
U_{j,l}^\pm=c_{j,l}^\pm z^{-jL}(1+O(z^{-1})),\,\,\,\,\,\, l>0
\ee
and
\be
U_{j,0}^\pm=c_{j,0}^\pm(1+O(z^{-1}))
\ee
for some constants $c_{j,l}^\pm$.

\end{proof}

\begin{proof}[Proof of Lemma~\ref{thm:betaPsi_k}]

Recall that $\check\Psi^\pm_k(x)$ satisfies  equations~\eqref{beta_diff_Psik_check}-\eqref{D_diff_Psik_check}.
%If we substitute the expansions
%\be
%\check \Psi^+_k =\sum_{j=0}^\infty\beta^j \check \Psi^{+{(j)}_k}
%\quad , \quad
%\check \Psi^-_k =\sum_{j=0}^\infty\beta^j \check \Psi^{-(j)}_k, 
%\ee
%we obtain recursion relations
%\bea
%\check\Psi^{\pm(m)}_k=\frac{1}{m} \sum_{j=1}^m U_j^\pm \check\Psi^{\pm (m-j)}_k,\\
%\eea 
%for the coefficients, which allow us to determine all the $\check\Psi^{\pm(m)}_k$'s from $\check\Psi^{\pm(0)}_k$.
From Theorem \ref{thm:kspectral} we conclude that
%Equations (\ref{Dlong}) can be expressed 
%\bea
%\pm \beta D \Psi^\pm_k &=\left(S(R_\pm)\pm\beta k \right) \Psi^\pm_k,
%\label{dPsik_eq}
%\eea
%and thus
\bea
\pm \beta D \check\Psi^\pm_k &=\left(S(O^\pm)-S(z)\pm\beta k \right) \check\Psi^\pm_k.
\label{dPsik_check_eq}
\eea
From the structure of the operators $O^\pm$, we have 
\be
S(O^\pm) =S(z)+\sum_{j=1}^{ML} (\pm \beta)^j S^{(j)},
\ee
where
\be
S^{(j)}=\sum_{k=0}^j S^{(j)}_k (z)  \left(z\frac{\pp}{\pp z}\right)^k, 
\ee
and the $S^{(j)}_k (z)$'s are rational functions with poles only at the zeros of $\phi(z)$. In particular
\be
S^{(1)}=\frac{1}{2} \frac{z^2}{\phi(z)} \left(\frac{\pp^2}{\pp z^2}G(S(z))\right)+  \frac{z^2}{\phi(z)} \left(\frac{\pp}{\pp z}G(S(z))\right)\frac{\pp}{\pp z}.
\ee
Equating the $\beta^0$ terms in (\ref{dPsik_check_eq}), we get
\bea
\left(z\frac{\pp}{\pp z}-k +\frac{1}{2} \frac{z\phi'(z)}{\phi(z)}\right) \check\Psi^{\pm(0)}_k=0,
\eea
which implies that
	\bea%\label{eq:PhiPsi0}
\check\Psi^{\pm(0)}_k \propto \frac{z^k}{\sqrt\phi(z)},
\eea
	where the proportionality constant is independent of $z$. From Proposition~\ref{adapted_basis_series}  and the comparison of the leading terms in the the vicinity of $z=0$, the proportionality constant is equal to $1$, which proves~\eqref{eq:PhiPsi0}.
For $m>0$, we get
\bea
\left(z\frac{\pp}{\pp z}-k +\frac{1}{2} \frac{z\phi'(z)}{\phi(z)}\right)
 \check\Psi^{+(m)}_k =\sum_{j=2}^{\min(LM, m-1)} S^{(j)} \check\Psi^{+(m-j-1)}_k,\\
\left(z\frac{\pp}{\pp z}-k +\frac{1}{2} \frac{z\phi'(z)}{\phi(z)}\right)
 \check\Psi^{-(m)}_k =\sum_{j=2}^{\min(LM, m-1)} (-1)^{j+1}S^{(j)} \check\Psi^{-(m-j-1)}_k.
\eea
Since the operators $S^{(j)}$ have poles only at the zeros of $\phi$, and not at the zeros of $G(S(z))$ that are not zeros of 
$\phi$, we conclude that $ \check\Psi^{-(m)}_k$ can have poles at most at the zeros of $\phi$, and at $z=0$ or $z=\infty$ depending on $k$.
\end{proof}

%%%%%%%%%%%%%% Appendix A.6 %%%%%%%%%%%%

\subsection{Section~\ref{sec:loop}}
\label{app_A6}

\begin{proof}[Proof of Theorem~\ref{thm:loopeqs}]
%	It follows from Proposition~\ref{prop:Wgnrational} that the connected functions $\tilde W_n(x_1,\dots,x_n)$ are, order by order in $\beta$, rational functions of variables $z_i$ such that $x_i=X(z_i)$, so that we can view them as multivariate functions of the $x_i$. In what follows we use the notation $\stackrel{m_i}{x_i}$ to denote that such a function is evaluated on the $m_i$-th sheet, i.e. the previous rational function is evaluated at $(z^{(m_1)}(z_1), \dots, z^{(m_n)}(z_n))$. Of course, this notation is local to some fixed orderable domain $U$.

	We fix an orderable domain $U$. All the variables $x,x',x_1,\dots$ are implicitly assumed to belong to $U$, and all variables $z,z', z_1, \dots$ to $X^{-1}(U)$.

	We start by defining ``modified'' functions $W^{mod,k}_{n+k}$ that are closely related to the $W_{n}$'s: the function $W^{mod,k}_{n+k}$ is defined
	in terms of the connected functions $\tilde W_i$ for $i\leq n+k$ in the same way as the nonconnected function $W_{n+k}$ is, except that the connected $2$-point functions are modified by a double pole, and that certain terms in the summation are omitted.
	Precisely, for $n, k \geq 0$ we let
%	We expand the nonconnected function $W_n(x_1,\dots,x_n)$ in terms of the connected functions $\tilde W_k$ for $k\leq n$, and replace each 2-point connected function by its singular version obtained by adding a double-pole. More precisely, we let
%	\be\label{eq:Wnmod}
%	W^{mod,0}_{n}(x_1,\dots,x_n) :=\sum_{P\in \mathcal{P}_n} \prod_{I \in P} \left(\tilde{W}_{|I|}(x_i,i\in I) +\frac{\delta_{|I|,2}}{(x_i-x_{i'})^2}\right)
%\ee
	\be\label{eq:Wnkmod}
	W^{mod,k}_{n+k}(x_1,\dots,x_k,x_{k+1},\dots,x_{n+k}) :=\sum_{P\in \widetilde{\mathcal{P}}_{n,k}} \prod_{I \in P} \left(\tilde{W}_{|I|}(x_i,i\in I) +\frac{\delta_{|I|,2}}{(x_i-x_{i'})^2}\right),
\ee
	where
	$\widetilde{\mathcal{P}}_{n,k}$ is the set of partitions of the set $\{1,\dots,n+k\}$, 	
	 that are such that  each part contains at least one element of $\{1,\dots,k\}$, and where in the product for $|I|=2$ we use the local notation $I=\{i,i'\}$. Observe that $W^{mod,0}_n = \delta_{n,0}$.

	From the definition~\eqref{eq:Wnkmod} and from Proposition~\ref{prop:detConnected} we directly obtain, for $k\geq 0$
	that 
	\be \label{eq:det1}
	W_{n+k}^{mod,k}(x_1,\dots,x_{n+k}) =\det' ( R_{i,j}(x_1,\dots,x_{n+k})	)_{1\leq i,j\leq n+k}
	\ee
	where 
	\be
	R_{i,j}(x_1,\dots,x_{n+k}) =\left\{\begin{array}{ll} K(x_i,x_j) & \mbox{if }i\neq j,\\ 
		\tilde W_1(x_i)
	%\lim_{x'\to x_i} \left(K(x_i,x')-\frac{1}{x_i-x'}\right)
	&\mbox{if }i=j,\end{array}\right.
\ee
	and where $\det'$ means that when we expand the determinant as a sum over permutations of $\{1,\dots,n+k\}$, we exclude permutations that have at least one cycle stabilizing $\{k+1,\dots,n+k\}$ (when $n=0$ it is the usual determinant).
	We now let 
	\be
	\Omega^{mod,k}_{n+k} (z_1,\dots,z_{n+k}):=W^{mod,k}_{n+k}(X(z_1),\dots,X(z_{n+k})),
	\ee
	which, from Proposition~\ref{prop:Wgnrational} is, order by order in $\beta$, a rational function of the $z_i$. Since the contribution of exponential factors cancel out along each cycle of the permutation, we can replace $K$ by $\check K$ in~\eqref{eq:det1} and obtain:
	\be\label{eq:OmegaModDet}
	\Omega^{mod,k}_{n+k} (z_1,\dots,z_{n+k})=\det' ( \check R_{i,j}(z_1,\dots,z_{n+k})	)_{1\leq i,j\leq n+k},
	\ee
	\be
	\check R_{i,j}(z_1,\dots,z_{n+k}) =\left\{\begin{array}{ll} \check K(z_i,z_j) & \mbox{if }i\neq j,\\
	%\lim_{x'\to x_i} \left(K(x_i,x')-\frac{1}{x_i-x'}\right)
	\tilde W_1(X(z_i))	
	&\mbox{if }i=j.\end{array}\right.
\ee

Consider the following functions:
\bea
P_{k,n}(z;z_1,\dots,z_n)
	&\&= \sum_{m_1<\dots<m_k} \Omega_{n+k}^{mod,k}(z^{(m_1)}(z),\dots,z^{(m_k)}(z),z_1,\dots,z_{n}) \\
P_{n}(z,y;z_1,\dots,z_n)
&\&= \sum_{k=0}^{LM} (-1)^k y^{LM-k} P_{k,n}(z;z_1,\dots,z_n).
\eea

	We now compute $P_n$,  starting with the case $n=0$. First introduce the matrix
	\be\label{eq:matK}
	\hat{\mathbf{K}}(z,z'):=\frac{\check{\mathbf \Psi}^-(z)^{-1} \ \check{\mathbf \Psi}^-(z')  } {X(z)-X(z')}, 
\ee
	and recall (Corollary~\ref{cor:KPsiinv}) that  $\hat{\mathbf{K}}(z,z')_{0,0}=\check K(z,z')$ and, more generally, 
	\be
	\hat{\mathbf{K}}(z,z')_{m,m'}=\check K(z^{(m)}(z),z^{(m')}(z')). 
	\ee

	Observe that, for  $m_1<\dots<m_k$, and for $i\neq j$, we have:
	\bea
	\lim_{z_j\rightarrow z_i} \check R_{i,j}(z^{(m_1)}(z_1), \dots,z^{(m_k)}(z_k))
	&&= \lim_{z_j\rightarrow z_i} (\hat{\bf K}(z_i,z_j))_{m_i,m_j}\cr
	&&= \left(\lim_{z_j\to z_i} \frac{\check{\mathbf{\Psi}}^-(z_i)^{-1} \check{\mathbf{\Psi}}^-(z_j)}{X(z^{(m_i)}(z_i))-X(z^{(m_i)}(z_j))} \right)_{m_i,m_i}\cr
	&&=\left( \check{\mathbf{\Psi}}^-(z_i)^{-1} {\mathbf D}(X(z_i)) \check{\mathbf{\Psi}}^-(z_i)-X(z_i)\mathbf{Y}(z_i) \right)_{m_i,m_j}\cr
	&&=\left( \check{\mathbf{\Psi}}^-(z_i)^{-1} {\mathbf D}(X(z_i)) \check{\mathbf{\Psi}}^-(z_i) \right)_{m_i,m_j}\cr \label{eq:Rij}
	&&
	\eea
from L'H\^opital's rule and Proposition~\ref{prop:diffMatrices},  and because $\mathbf{Y}$ does not contribute to nondiagonal terms. Similarly,  from the first equation of Proposition~\ref{prop:detConnected} we have 
\bea
\check R_{i,i}(z^{(m_1)}(z_1), \dots,z^{(m_k)}(z_k)) 
	&\&=	\lim_{z'\to z^{(m_i)}(z_i)} \left( e^{{\int_{z^{(m_i)}(z_i)}^{z'} Y(u)X'(u)du}}\check{{K}}(z^{(m_i)}(z_i),z')-\frac{1}{X(z^{(m_i)}(z_i))-X(z')}\right)\cr 
	&\&= \left(\lim_{z'\to z} \frac{\check{\mathbf{\Psi}}^-(z_i)^{-1} \check{\mathbf{\Psi}}^-(z')-{\rm Id}}{X(z^{(m_i)}(z_i))-X(z')} \right)_{m_i,m_i} +Y(z^{(m_i)}(z_i))\cr
	&\& \cr
	&\&=\left( \check{\mathbf{\Psi}}^-(z_i)^{-1} {\mathbf D}(X(z_i)) \check{\mathbf{\Psi}}^-(z_i) \right)_{m_i,m_i},\label{eq:Rii}
	\eea
where we have again used L'H\^opital's rule and Proposition~\ref{prop:diffMatrices}.
We thus obtain
	\bea\label{eq:P0(x,y)}
P_{0}(z,y)
&\&= \sum_{k=0}^{LM} (-1)^k y^{LM-k}  \sum_{m_1<\dots<m_k} 
	\sum_\sigma \sgn(\sigma)  \prod_i [\check{\mathbf \Psi}^-(z)^{-1}  {\mathbf D}(X(z)) \check{\mathbf \Psi}^-(z)]_{m_i,m_{\sigma(i)}}  \cr
&\& = \det(y{\rm Id}-\check{\mathbf \Psi}^-(z)^{-1}  {\mathbf D}(X(z)) \check{\mathbf \Psi}^-(z)) \cr
&\&= \det\left(y{\rm Id}-  {\mathbf D}(X(z))\right). 
\eea
Extracting the submaximal and subsubmaximal coefficient in $y$ gives, respectively, the basic case~\eqref{loopeq11} of the first loop equation, and the basic case ($n=2$) of the second loop equation~\eqref{loopeq2}.

\medskip
We now turn to the case $n>0$. We claim that 
\be\label{eq:OmegaBigDet}
\Omega_{n+k}^{\mathrm{mod,k}}(z^{(m_1)}(z),  \dots, z^{(m_k)}(z), z_{1},\dots,z_{n}) = \det' \left(\check S_{i,j}\right)_{i,j\in I\uplus J}
\ee
where $I=\{1,\dots,k\}$, $J=\{1',\dots,n'\}$, and 
\bea
	\label{eq:evalMatrixEntry1}
\check S_{i,j}  =\left( \check{\mathbf{\Psi}}^-(z)^{-1} {\mathbf D}(X(z)) \check{\mathbf{\Psi}}^-(z) \right)_{m_i,m_j} & \mbox{ if } i,j\in I \\
%	\check S_{i',j'} =	\left( \check{\mathbf{\Psi}}^-(z_i)^{-1} {\mathbf D}(X(z_i)) \check{\mathbf{\Psi}}^-(z_i) \right)_{0,0} & \mbox{ if } i',j'\in J \mbox{ and } i'=j' \\
	\label{eq:evalMatrixEntry3}
\check S_{i',j'} =	\left( \frac{\check{\mathbf{\Psi}}^-(z_i)^{-1} \check{\mathbf{\Psi}}^-(z_j)}{X(z_i)-X(z_j)} \right)_{0,0} & \mbox{ if } i',j'\in J \mbox{ and } i'\neq j' \\
	\label{eq:evalMatrixEntry4}
\check S_{i,j'}=	\left( \frac{\check{\mathbf{\Psi}}^-(z)^{-1} \check{\mathbf{\Psi}}^-(z_j)}{X(z)-X(z_j)} \right)_{m_i,0} & \mbox{ if } i\in I,j'\in J  \\
	\label{eq:evalMatrixEntry5}
\check S_{i',j} =	\left( \frac{\check{\mathbf{\Psi}}^-(z_i)^{-1} \check{\mathbf{\Psi}}^-(z)}{X(z_i)-X(z)} \right)_{0,m_j} & \mbox{ if } i'\in J,j\in I. 
\eea
This follows from the pseudo-determinantal formula~\eqref{eq:OmegaModDet}, and by examining what the matrix entries become when the $k+n$ variables are set to $(z^{(m_1)}(z),  \dots, z^{(m_k)}(z), z_{1},\dots,z_{n})$. More precisely the first equality~\eqref{eq:evalMatrixEntry1}  follows by the computations made in~\eqref{eq:Rii} and~\eqref{eq:Rij}, and the next one~\eqref{eq:evalMatrixEntry3} is also a consequence of $\eqref{eq:Rij}$ with $m_i=m_j=0$. The two remaining equalities are direct consequences of 
Proposition~\ref{prop:diffMatrices} and Theorem~\ref{thm:CD}, similarly as Corollary~\ref{cor:KPsiinv}.

\smallskip

Now expand the pseudo determinant $\det'$ in~\eqref{eq:OmegaBigDet} as a sum over permutations of $I\uplus J$, and recall that the meaning of the prime symbol is that we exclude permutations that have at least one cycle that stabilizes $J$.
Such a permutation can be seen as a directed graph consisting of cycles on the vertex set $I\uplus J$. It can be transformed into a permutation of $I$ by ``contracting'' all edges incident to an element of $J$. Conversely, given any permutation of $I$, we can transform it into a permutation of $I\uplus J$ by replacing each directed edge $i\rightarrow \sigma(i)$ into a sequence $i \rightarrow i_1 \rightarrow i_2 \rightarrow \dots \rightarrow i_k \rightarrow \sigma(i)$ for $k\geq 0$ and \emph{distinct} $i_1,\dots,i_k \in J$. By~\eqref{eq:evalMatrixEntry1},\eqref{eq:evalMatrixEntry4},\eqref{eq:evalMatrixEntry5} the contribution to the matrix elements of this subtitution for $k\geq 1$ is
\be
\left( \frac{\check{\mathbf{\Psi}}^-(z)^{-1} \check{\mathbf{\Psi}}^-(z_{i_1})}{X(z)-X(z_{i_1})} \right)_{m_i,0}
\prod_{p=1}^{k-1}
\left( \frac{\check{\mathbf{\Psi}}^-(z_{i_p})^{-1} \check{\mathbf{\Psi}}^-(z_{i_{p+1})}}{X(z_{i_p})-X(z_{i_{p+1}})} \right)_{0,0}
\left( \frac{\check{\mathbf{\Psi}}^-(z_{i_k})^{-1} \check{\mathbf{\Psi}}^-(z)}{X(z_{i_k})-X(z)} \right)_{0,m_{\sigma(i)}}. \label{eq:valCoeffM}
\ee
It follows by the definition~\eqref{eq:defD}  that~\eqref{eq:valCoeffM} is precisely the coefficient of $\epsilon_{i_1}\dots \epsilon_{i_k}$ in the entry $(m_i,m_{\sigma(i)})$ of the matrix 
\be
\check{\mathbf{\Psi}}^-(z)^{-1} \tilde {\mathbf D}_{n+2}(X(z)) \check{\mathbf{\Psi}}^-(z),
\ee
where here and below the indices of the variables $z_i$ are interpreted modulo $n$.
The same is true for $k=0$ by~\eqref{eq:evalMatrixEntry1}, interpreting the coefficient of $\epsilon_{i_1}\dots \epsilon_{i_k}$ as the constant coefficient in $\epsilon$'s ({\it i.e.}  $\epsilon_1=\dots=\epsilon_k=0$).

Since all permutations of $I\uplus J$ without cycle stabilizing $J$ can be obtained from a permutation of $I$ using this substitution procedure, it follows that~\eqref{eq:OmegaBigDet} can be rewritten as a sum over (unrestricted) permutations of $I$
\be
\Omega_{n+k}^{\mathrm{mod,k}}(z^{(m_1)}(z),  \dots, z^{(m_k)}(z), z_{1},\dots,z_{n}) =\sum_{\sigma \in \mathfrak{S}(I)} \epsilon(\sigma) [\epsilon_1 \dots \epsilon_n] \prod_i [\check{\mathbf \Psi}^-(z)^{-1}  \tilde {\mathbf D}_{n+2}(x) \check{\mathbf \Psi}^-(z)]_{m_i,m_{\sigma(i)}}\ee
where the extraction of coefficient ensures that each element of $J$ appears once in the graph, {\it i.e.} that we only consider contributions coming from permutations of $I\uplus J$.

Because the sum over permutations of $I$ is unrestricted, we can use the same computation as in~\eqref{eq:P0(x,y)} with usual (non-prime) determinants and we get
\bea
P_{n}(x,y; z_{1},\dots,z_{n})
&\& = [\epsilon_1 \dots \epsilon_n]\det(y{\rm Id}-\check{\mathbf \Psi}^-(z)^{-1} \tilde  {\mathbf D}_{n+2}(x) \check{\mathbf \Psi}^-(z))\cr
&\&= [\epsilon_1 \dots \epsilon_n]\det\left(y{\rm Id}- \tilde {\mathbf D}_{n+2}(x)\right). 
\eea
Extracting the submaximal and subsubmaximal coefficient in $y$ gives respectively, up to a shift of two in the value of $n$ and in the indices of the $z_i$,  the generic case of the first and second loop equations (respectively~\eqref{loopeq12} and~\eqref{loopeq2}).
\end{proof}

\subsection{Section~\ref{sec:toprec}}
\label{app_A7}

\begin{proof}[Proof of Theorem~\ref{thm:toprec}]
We recall here the proof of \cite{EO1} in the case when all ramification points are simple.
Consider $(g,n)\neq (0,1),(0,2)$ and $n\geq 1$.
By the Cauchy residue theorem, we have
\be
\tilde\omega_{g,n}(z_1,\dots,z_n) = \Res_{z\to z_1} \frac{dz_1}{z-z_1}\,\tilde\omega_{g,n}(z,z_2,\dots,z_n).
\ee
where $\tilde\omega_{g,n}(z_1,\dots,z_n)$ is a rational function of $z_1$.
Since it has poles only at ramification points, we may move the integration contour and get
\be
\tilde\omega_{g,n}(z_1,\dots,z_n) = -\sum_{a\in \mathcal L} \Res_{z\to a} \frac{dz_1}{z-z_1}\,\tilde\omega_{g,n}(z,z_2,\dots,z_n).
\ee
At a branch point $a$ (assumed to be generic here), there are exactly 2 branches that meet. Therefore
there exists a unique $i$ such that $z^{(i)}\sim z$.
Writing this as
\be
z^{(i)} = \sigma_a(z),
\ee
the map $z\mapsto \sigma_a(z)$ is locally an analytic involution in a neighbourhood of $a$,  with
$X(z)=X(\sigma_a(z))$, $\sigma_a(a)=a$.

Up to  a change of variable $z\to \sigma_a(z)$, 
	\bea\label{eq:inteqproof}
\tilde\omega_{g,n}(z_1,\dots,z_n) 
&=& -\sum_{a\in \mathcal L} \Res_{z\to a} \frac{dz_1}{\sigma_a(z)-z_1}\,\tilde\omega_{g,n}(\sigma_a(z),z_2,\dots,z_n) \cr
&=& \frac{-1}{2} \sum_{a\in \mathcal L} \Res_{z\to a} \Big[ \frac{dz_1}{z-z_1} \,\tilde\omega_{g,n}(z,z_2,\dots,z_n) \cr
&& \quad + \frac{dz_1}{\sigma_a(z)-z_1} \,\tilde\omega_{g,n}(\sigma_a(z),z_2,\dots,z_n) \Big] \cr
&&
\eea
Using the first loop equations we  have:
\be
	\tilde\omega_{g,n}(\sigma_a(z),z_2,\dots,z_n) = - \omega_{g,n}(z,z_2,\dots,z_n) - \sum_{j\neq 0,i} \tilde\omega_{g,n}(z^{(j)},z_2,\dots,z_n) + (rest),
\ee
where the rest can have poles only at the zeroes of $X(z)$. Substituting this in~\eqref{eq:inteqproof}, 
and observing that since $\tilde\omega_{g,n}(z^{(j)},z_2,\dots,z_n)$ and the rest have no pole at $z=a,$ they can be dropped from the residue
to obtain
 that
\bea
\tilde\omega_{g,n}(z_1,\dots,z_n) 
&=& \frac{-1}{2} \sum_{a\in \mathcal L} \Res_{z\to a} \Big[ \frac{dz_1}{z-z_1} - \frac{dz_1}{\sigma_a(z)-z_1}  \Big] \,\tilde\omega_{g,n}(z,z_2,\dots,z_n) .
\cr
&&
\eea
Rewriting this as
\bea
\tilde\omega_{g,n}(z_1,\dots,z_n) 
&=& \sum_{a\in \mathcal L} \Res_{z\to a} \Big[ \frac{dz_1}{z-z_1} - \frac{dz_1}{\sigma_a(z)-z_1}  \Big] \,\frac{\tilde\omega_{g,n}(z,z_2,\dots,z_n)\,(Y(\sigma_a(z))-Y(z))\,dX(z) }{2(Y(z)-Y(\sigma_a(z)))\,dX(z)} \cr
&=& \sum_{a\in \mathcal L} \Res_{z\to a}  \frac{(z-\sigma_a(z))\ dz_1}{(z-z_1)(z_1-\sigma_a(z))}   \,\frac{\tilde\omega_{g,n}(z,z_2,\dots,z_n)\,(Y(\sigma_a(z))-Y(z))\,dX(z) }{2(Y(z)-Y(\sigma_a(z)))\,dX(z)}, \cr
&&
\eea
consider the numerator:
\be
\NN:=\tilde\omega_{g,n}(z,z_2,\dots,z_n)\,\tilde\omega_{0,1}(\sigma_a(z)) - 
\tilde\omega_{g,n}(z,z_2,\dots,z_n)\,\tilde\omega_{0,1}(z).
\ee
We may add to this any rational function of $z$ that has no pole at $a$ (thus $O(1)$ in the Taylor expansion at $a$) without changing the residue.
Using the first loop equation, we have:
\be
\NN + O(1) =\tilde\omega_{g,n}(z,z_2,\dots,z_n)\,\tilde\omega_{0,1}(\sigma_a(z)) + 
\tilde\omega_{g,n}(\sigma_a(z),z_2,\dots,z_n)\,\tilde\omega_{0,1}(z) 
\ee
Observe the following:
\bea
&\& \sum_{0\leq k<l\leq d} \tilde\omega_{g,n}(z^{(k)},z_2,\dots,z_n) \tilde\omega_{0,1}(z^{(l)}) \cr
&\& = \tilde\omega_{g,n}(z^{(0)},z_2,\dots,z_n)\tilde\omega_{0,1}(z^{(i)})+\tilde\omega_{g,n}(z^{(i)},z_2,\dots,z_n)\tilde\omega_{0,1}(z^{(0)}) \cr
&\& + \sum_{k<l,\,\, k,l\neq 0,i} \tilde\omega_{g,n}(z^{(k)},z_2,\dots,z_n) \tilde\omega_{0,1}(z^{(l)}) \cr
&\& + \sum_{k\neq 0,i} (\tilde\omega_{g,n}(z^{(0)},z_2,\dots,z_n)+\tilde\omega_{g,n}(z^{(i)},z_2,\dots,z_n) )  \tilde\omega_{0,1}(z^{(k)}) \cr
&\& + \sum_{k\neq 0,i} (\tilde\tilde\omega_{0,1}(z^{(0)})+\tilde\omega_{0,1}(z^{(i)}) )  \tilde\omega_{0,1}(z^{(k)},z_2,\dots,z_n), 
\eea
where none of the last three lines have poles at $a$.
It follows that
\be
\NN = \sum_{0\leq k<l\leq d} \tilde\omega_{g,n}(z^{(k)},z_2,\dots,z_n) \tilde\omega_{0,1}(z^{(l)})
+ O(1).
\ee
Writing
\bea
Q_{g,n}(X(z);z_2,\dots,z_n)\,dX(z)^2 = \sum_{0\leq k<l\leq d} \tilde\omega_{g,n}(z^{(k)},z_2,\dots,z_n) \tilde\omega_{0,1}(z^{(l)}) + \mathcal W_{g,n}(z^{(k)},z^{(l)},z_2,\dots,z_n), \cr
&&
\eea
since $Q_{g,n}$ has no pole at the branch points, we  have
\bea
\NN 
&=& - \sum_{0\leq k<l\leq d}  \mathcal W_{g,n}(z^{(k)},z^{(l)},z_2,\dots,z_n) + O(1) \cr
&=& -   \mathcal W_{g,n}(z^{(0)},z^{(i)},z_2,\dots,z_n)  \cr
&& - \sum_{k<l,\,\,k,l\neq 0,i}  \mathcal W_{g,n}(z^{(k)},z^{(l)},z_2,\dots,z_n) +  \cr
&& - \sum_{k\neq 0,i}  \mathcal W_{g,n}(z^{(0)},z^{(k)},z_2,\dots,z_n)+ \mathcal W_{g,n}(z^{(i)},z^{(k)},z_2,\dots,z_n) 
 + O(1) \cr
 &&
\eea
The last three lines again have no poles at $a$, so we have
\bea
\NN 
&=& -   \mathcal W_{g,n}(z,\sigma_a(z),z_2,\dots,z_n)   + O(1) .
\eea
This implies
\bea
\tilde\omega_{g,n}(z_1,\dots,z_n) 
&=& -\,\sum_{a\in \mathcal L} \Res_{z\to a} \Big[ \frac{dz_1}{z-z_1} - \frac{dz_1}{\sigma_a(z)-z_1}  \Big] \,\frac{{\mathcal W}_{g,n}(z,\sigma_a(z),z_2,\dots,z_n)\ }{2(Y(z)-Y(\sigma_a(z)))\,dX(z)}, \cr
&&
\eea
which is the topological recursion relation.

%For higher order ramification points, the topological recursion relations of \cite{BouchardEynardlocalglobal} can be proved in a similar way, starting from the loop equations. It was also proved in \cite{BouchardEynardlocalglobal} that higher order topological recursion is the limit of simple topological recursion when simple ramification points coalesce into higher order ramification points. We refer to \cite{BouchardEynardlocalglobal} for details.

\end{proof}

\begin{proof}[Proof of Proposition~\ref{prop:F03}]
Topological recursion implies that
\be
\tilde F_{0,3}(X(z_1),X(z_2),X(z_3))
	= \sum_{a\in \mathcal L}\Res_{z\to a} \frac{dz}{(z-z_1)(z-z_2)(z-z_3)} \ \frac{1}{X'(z)Y'(z)},
\ee
which gives
\bea
\tilde F_{0,3}(X(z_1),X(z_2),X(z_3))
&=& \sum_{a\in \mathcal L}\Res_{z\to a} \frac{dz}{(z-z_1)(z-z_2)(z-z_3)} \ \frac{z^2 G(S(z))^2}{\phi(z) (zG(S(z)) S'(z) - S(z) \phi(z))} \cr
&=& \sum_{a\in \mathcal L}\Res_{z\to a} \frac{dz}{(z-z_1)(z-z_2)(z-z_3)} \ \frac{z G(S(z))}{\phi(z)  S'(z) } \cr
&=& -\sum_{i=1}^3 \frac{1}{\prod_{j\neq i} (z_i-z_j)} \ \frac{z_i G(S(z_i))}{\phi(z_i)  S'(z_i) } \cr
&& - \sum_{S'(b)=0}\Res_{z\to b} \frac{dz}{(z-z_1)(z-z_2)(z-z_3)} \ \frac{z G(S(z))}{\phi(z)  S'(z) } \cr
&=& -\sum_{i=1}^3 \frac{1}{\prod_{j\neq i} (z_i-z_j)} \ \frac{z_i G(S(z_i))}{\phi(z_i)  S'(z_i) } \cr
&& - \sum_{S'(b)=0} \frac{1}{(b-z_1)(b-z_2)(b-z_3)} \ \frac{b }{ S''(b) } \cr
&=& -\sum_{i=1}^3 \frac{1}{\prod_{j\neq i} (z_i-z_j)} \ \left( \frac{z_i G(S(z_i))}{\phi(z_i)  S'(z_i) }  - \frac{z_i}{S'(z_i)} \right) \cr
&=& -\sum_{i=1}^3 \frac{1}{\prod_{j\neq i} (z_i-z_j)} \  \frac{z_i^2 G'(S(z_i))}{\phi(z_i)  } .
\eea
\end{proof}

\end{appendix}

\bigskip
 
%%%%%%%%%%%%%%%% Bibliography %%%%%%%%%%%%%%%%

\newcommand{\arxiv}[1]{\href{http://arxiv.org/abs/#1}{arXiv:{#1}}}

\end{document}